\documentclass[a4paper, 11pt]{article}

\usepackage[utf8]{inputenc}

\usepackage{
amssymb, amsmath, amsthm, amstext, amsfonts, amsopn, amscd, amsbsy,comment
}
\usepackage{mathrsfs, dsfont, mathdots}
\usepackage{bm}
\usepackage{accents}

\usepackage{graphicx}
\usepackage[usenames,dvipsnames]{xcolor}

\usepackage[colorlinks=true, linkcolor=black, breaklinks=true, citecolor=blue, urlcolor=red]{hyperref} 
\usepackage{float}
\usepackage{subcaption}
\usepackage{caption}
\usepackage{tikz}

\usepackage[numbers]{natbib}
\usepackage{url}
\newlength{\dinwidth}
\newlength{\dinmargin}

\renewcommand{\emptyset}{\varnothing}

\newcommand{\textb}{\textcolor{blue}}
\newcommand{\textr}{\textcolor{red}}
\newcommand{\texto}{\textcolor{orange}}

\newcommand{\fA}{\mathfrak{A}}
\newcommand{\ua}{\textup a}
\newcommand{\fB}{\mathfrak{B}}

\newcommand{\cD}{\mathcal{D}}
\newcommand{\cE}{\mathcal{E}}
\newcommand{\fF}{\mathfrak{F}}\newcommand{\cF}{\mathcal{F}}
\newcommand{\fg}{\mathfrak{g}}
\newcommand{\fh}{\mathfrak{h}}
\newcommand{\fH}{\mathfrak{H}}\newcommand{\cH}{\mathcal{H}}
\newcommand{\cI}{\mathcal{I}}
\newcommand{\cK}{\mathcal{K}}
\newcommand{\cL}{\mathcal{L}}
\newcommand{\fl}{\mathfrak{l}}

\newcommand{\cO}{\mathcal{O}}
\newcommand{\fQ}{\mathfrak{Q}}
\newcommand{\fP}{\mathfrak{P}}

\newcommand{\fu}{\mathfrak{u}}

\newcommand{\CC}{\mathbb{C}}\newcommand{\C}{\mathbb{C}}
\newcommand{\Z}{\mathbb{Z}}
\newcommand{\T}{\mathbb{T}}
\newcommand{\RR}{\mathbb{R}}\newcommand{\R}{\mathbb{R}}
\newcommand{\NN}{\mathbb{N}}\newcommand{\N}{\mathbb{N}}

\newcommand{\vep}{\varepsilon}

\newcommand{\ltwo}{\ell^{2}}

\DeclareMathOperator{\Ad}{Ad}
\DeclareMathOperator{\ad}{ad}

\DeclareMathOperator{\1}{\mathds{1}}
\DeclareMathOperator{\Tr}{Tr}

\DeclareMathOperator{\pf}{pf}
\DeclareMathOperator{\vol}{vol}

\DeclareMathOperator{\Diff}{Diff}

\DeclareMathOperator{\im}{im}

\DeclareMathOperator{\supp}{supp}

\DeclareMathOperator{\CAR}{\textup{CAR}}
\DeclareMathOperator{\SDC}{\textup{SDC}}
\DeclareMathOperator{\TL}{\textup{TL}}
\DeclareMathOperator{\TIM}{\textup{TIM}}
\DeclareMathOperator{\sign}{\textup{sign}}
\DeclareMathOperator{\sinc}{sinc}

\DeclareMathOperator{\sd}{sd}
\DeclareMathOperator{\std}{std}
\DeclareMathOperator{\per}{per}
\DeclareMathOperator{\alg}{alg}

\DeclareMathOperator{\act}{\curvearrowright}

\def\Undertilde#1{\mathord{\vtop{\ialign{##\crcr
$\hfil\displaystyle{#1}\hfil$\crcr\noalign{\kern1.5pt\nointerlineskip}
$\hfil\widetilde{}\hfil$\crcr\noalign{\kern1.5pt}}}}}

\def\S2{S^{1(2)}}

\def\sl2{{{\rm SL}(2,\RR)}}
\def\psl2{{{\rm PSL}(2,\RR)}}
\def\u1{{{\rm V}(1)}}
\def\su2{{{\rm SV}(2)}}

\def\so3{{{\rm SO}(3)}}

\newtheorem{letterthm}{Theorem}

\newtheorem{thm}{Theorem} [section]
\newtheorem{prop}[thm]{Proposition}
\newtheorem{lemma}[thm]{Lemma}
\newtheorem{defn}[thm]{Definition}
\newtheorem{cor}[thm]{Corollary}
\newtheorem{rem}[thm]{Remark}
\newtheorem{example}[thm]{Example}
\newtheorem{criterion}[thm]{Criterion}
\newtheorem{conjecture}{Conjecture}
\newtheorem{assumption}{Assumption}

\newcommand{\bea}{\begin{assumption}}
	\newcommand{\eea}{\end{assumption}}
\newcommand{\beco}{\begin{conjecture} }
	\newcommand{\eeco}{\end{conjecture} }
\newcommand{\beq}{\begin{equation}}
	\newcommand{\eeq}{\end{equation}}
\newcommand{\beqa}{\begin{eqnarray}}
	\newcommand{\eeqa}{\end{eqnarray}}
\newcommand{\ben}{\begin{arabicenumerate}}
	\newcommand{\een}{\end{arabicenumerate}}
\newcommand{\bex}{\begin{example}}
	\newcommand{\eex}{\end{example}}
\newcommand{\ber}{\begin{remark}}
	\newcommand{\eer}{\end{remark}}
\newcommand{\bec}{\begin{corollary}}
	\newcommand{\eec}{\end{corollary}}
\newcommand{\bep}{\begin{proposition}}
	\newcommand{\eep}{\end{proposition}}
\newcommand{\becr}{\begin{criterion}}
	\newcommand{\eecr}{\end{criterion}}

\begin{document}

\tikzset{->-/.style={decoration={
  markings,
  mark=at position #1 with {\arrow{latex}}},postaction={decorate}}}

\title{Conformal field theory from lattice fermions}
\author{Tobias J. Osborne${}^{1}$, Alexander Stottmeister${}^{1}$}
\date{
\small{${}^{1}$ Institut f\"ur Theoretische Physik, Leibniz Universit\"at Hannover
\\Appelstra\ss e 2, 30167 Hannover, Germany} \\[0.5cm]
\today}
\maketitle
\begin{abstract}
We provide a rigorous lattice approximation of conformal field theories given in terms of lattice fermions in 1+1-dimensions, focussing on free fermion models and Wess-Zumino-Witten models. To this end, we utilize a recently introduced operator-algebraic framework for Wilson-Kadanoff renormalization. In this setting, we prove the convergence of the approximation of the Virasoro generators by the Koo-Saleur formula. From this, we deduce the convergence of lattice approximations of conformal correlation functions to their continuum limit. In addition, we show how these results lead to explicit error estimates pertaining to the quantum simulation of conformal field theories.
\end{abstract}

\tableofcontents

\section{Introduction}
\label{sec:intro}

A rigorous mathematical formulation of quantum field theory (QFT) is one the central challenges for the new millennium \cite{seibergNathanSeiberg20152014}. A variety of approaches to build a rigorous theory of QFT have been employed throughout the past decades and have lead to deep results and insights -- from \emph{constructive QFT} \cite{GlimmQuantumFieldTheory, glimmQuantumPhysicsFunctional1987, SummersAPerspectiveOn} to algebraic QFT \cite{HaagAnAlgebraicApproach, LechnerConstructionOfQuantum, BahnsLocalNetOf}, vertex operator algebras (VOAs) \cite{BorcherdsVertexAlgebrasKac, FrenkelVertexOperatorAlgebras, FrenkelBasicRepresentationsOf, TenerGeometricRealizationOf}, and beyond \cite{segalDefinitionConformalField2004, GrosseSolutionOfThe, GubinelliAPDEConstruction}. However, a completely satisfactory mathematical theory of QFT is still yet to be obtained. For this reason a proof of the rigorous existence of (and mass gap for) Yang-Mills theory was selected as one of the Clay maths prize problems. 

Several key difficulties to obtaining a rigorous formulation of QFT are already epitomized by $1+1$-dimensional gapless quantum field theories, in particular, \emph{conformal field theories} (CFTs) where a variety of constructions and classifications results has been obtained \cite{BelavinInfiniteConformalSymmetry, DiFrancescoCFTBook, EvansQuantumSymmetriesOn, WassermannOperatorAlgebrasAnd, BrunettiModularStructureAnd, GabbianiOperatorAlgebrasAnd, BuchholzTheCurrentAlgebra, KawahigashiClassificationOfLocal, CarpiFromVertexOperator}. 

A natural strategy to prove the existence of such a theory is to somehow realize it as the limit of a sequence of \emph{discretized approximations} (see \cite{JonesSomeUnitaryRepresentations, JonesANoGo, JonesScaleInvariantTransfer, BrothierConstructionsOfConformal, KlieschContinuumLimitsOf, OsborneQuantumFieldsFor} for recent attempts in a Hamiltonian setting), or \emph{lattice models}, on finite spatial lattices $\Lambda_{N} = \varepsilon_{N}\{-L_{N}, -L_{N} +1, \ldots, L_{N}-2, L_{N}-1\}\subset\varepsilon_{N}  \mathbb{Z}$, where $L=\varepsilon_N r_N$ is the length of the system. Associated to each lattice site $x\in \Lambda_N$ is a quantum degree of freedom, typically modeled by a Hilbert space $\mathcal{H}_x$, so that one assigns the \emph{total Hilbert space} $\mathcal{H}_{N} \subset \bigotimes_{x\in \Lambda_{N}} \mathcal{H}_x$ to the system\footnote{In the cases considered in this paper the total Hilbert space is provided by fermionic Fock space $\fF_{\ua}(\mathfrak{h}_{N})$ based on a one-particle space $\mathfrak{h}_{N}$ for a particle with spin hopping on $\Lambda_N$, which can be realised as subspace of a direct sum of tensor products of smaller Hilbert spaces.}. One requires, further, a \emph{Hamiltonian} $H_{0}^{(N)}$ for each discretization
\begin{align*}
H^{(N)}_0 = \tfrac{L}{\pi}\varepsilon_N\sum_{x\in\Lambda_N} h_x^{(N)},
\end{align*}
which is a sum of self-adjoint operators $h_x^{(N)}$ on $\mathcal{H}_N$ with finite support centered on $x$. A key additional ingredient that must be specified is a $C^{*}$-algebra $\mathfrak{A}_N$ (considered as a subalgebra of the bounded operators $B(\mathcal{H}_N)$ on $\cH_{N}$) of basic \emph{observables or fields}, corresponding to discretizations of their potential continuum counterparts. Given this data the task is to assign mathematical meaning to the \emph{scaling limit} of $\{\mathfrak{A}_N, \mathcal{H}_N, H^{(N)}_0\}$. 

A multitude of challenges must be overcome to realize this limit, in particular: (a) specifying the mathematical data to compare the discretizations with differing lattice size (and hence the topology for the limit); (b) identifying the correct lattice Hamiltonians for a target QFT; (c) proving the convergence of the sequence; and (d) proving that the requisite symmetry groups of the QFT are realized via projective unitary representations on the limit space. The recovery of expected continuum symmetries, i.e., requirement (d), typically poses severe difficulties (see \cite{SmirnovTowardsConformalInvariance, ChelkakUniversalityInThe, HonglerConformalFieldTheory, DuminilCopinRotationalInvarianceIn} for major advances in two-dimensional critical lattice models in the Euclidean framework).

The recently introduced \emph{operator-algebraic renormalization} (OAR) \cite{StottmeisterOperatorAlgebraicRenormalization, MorinelliScalingLimitsOf, BrothierConstructionsOfConformal} supplies a novel approach to defining, and proving the existence of, the above limit in Hamiltonian approach, thus overcoming the four challenges outlined in the previous paragraph. To do this one introduces the following additional structures:
\begin{itemize}
	\item \emph{Refining quantum channels}, which are unital completely positive (ucp) maps, $\alpha^{N}_{N+1}:\fA_{N}\rightarrow\fA_{N+1}$, connecting the observable algebras between the different scales. These induce dual \emph{coarse-graining} trace-preserving completely positive (tpcp) maps $\cE^{N+1}_{N}$ between the state spaces, or density matrices respectively preduals, $B(\cH_{N+1})_{*}\rightarrow B(\cH_{N})_{*}$, in the ultraweakly continuous case. Here, $B(\cH_{N})_{*}$ denotes the predual on $B(\cH_{N})$, i.e.~the trace-class operators on $\cH_{N}$.
	\item a sequence of initial (bare) states $\omega^{(N)}_{0}$ on the algebras $\mathfrak{A}_N$, possibly given by a density matrix $\rho^{(N)}_{0}\in B(\cH_{N})_{*}$, which are renormalized according to $\omega^{(N)}_{M} = \omega^{(N+M)}_{0}\circ\alpha^{N}_{N+M}$.
\end{itemize}
The scaling limit is then obtained via the Gelf'and-Naimark-Segal (GNS) construction\footnote{Informally, this means that the continuum Hilbert space and its inner product are reconstructed from the correlation functions of the basic observables $\fA_{N}$ in the scaling limit $\omega_{\infty}$.} $\{\cH_{\infty}, \pi_{\infty}, \Omega_{\infty}\}$ applied to the inductive limit $\fA_{\infty} = \varinjlim_{N}\fA_{N}$ and scaling-limit state $\omega^{(N)}_{\infty} = \lim_{M\rightarrow\infty}\omega^{(N)}_{M}$ (the existence of the latter usually requires the imposition of additional \emph{renormalization conditions} on the sequence $\omega^{(N)}_{0}$). The convergence of Hamiltonians $H^{(N)}_{0}\rightarrow H$ is considered in terms convergent operator sequences $\lim_{M\rightarrow\infty}\limsup_{N\rightarrow\infty}\|H^{(N)}_{0}-\alpha^{M}_{N}(H^{(M)}_{0})\|_{*}$ in a suitable operator topology (or their associated unitary respectively automorphism groups) \cite{DuffieldMeanFieldDynamical}. It is important to note that this convergence is a direct extension of the convergence that describes the inductive-limit algebra $\fA_{\infty}$, i.e., each element $O\in\fA_{\infty}$ is the limit of a convergent sequence $O_{N}$ in the sense that:
\begin{align*}
\lim_{M\rightarrow\infty}\limsup_{N\rightarrow\infty}\|O_{N}-\alpha^{M}_{N}(O_{M})\|_{C^{*}} & = 0.
\end{align*}

To obtain a (full) CFT via an OAR scaling limit one needs to realize the infinite-dimensional conformal group of symmetries via projective unitary representations on the limit space. This requires the additional identification of a discretized family of generators corresponding to the \emph{Virasoro algebra with central charge} $c$ \cite{DiFrancescoCFTBook},
\begin{align*}
[L_{k}, L_{k’}] & = \tfrac{L}{\pi}(k-k’)L_{k+k’} + \delta_{k+k’,0}\tfrac{c}{12}(\tfrac{L}{\pi}k)((\tfrac{L}{\pi}k)^2-1), \\ \nonumber 
[\overline{L}_{k}, \overline{L}_{k’}] & = \tfrac{L}{\pi}(k-k’)\overline{L}_{k+k’} + \delta_{k+k’,0}\tfrac{c}{12}(\tfrac{L}{\pi}k)((\tfrac{L}{\pi}k)^2-1), \\ \nonumber
[L_{k}, \overline{L}_{k’}] & = 0.
\end{align*}
Such a family is furnished by the \emph{Koo-Saleur (KS) approximants} \cite{KooRepresentationsOfThe} 
\begin{align*}
L_{k}^{(N)} &\!=\!\tfrac{1}{2}\Big(H_{k}^{(N)}\!+\!\tfrac{\pi\vep_{N}}{2L\sin(\frac{1}{2}\vep_{N}k)}\Big[H_{k}^{(N)}\!,H_0^{(N)}\Big]\Big), & \overline{L}_{k}^{(N)} \!=\!\tfrac{1}{2}\Big(H_{-k}^{(N)}\!+\!\tfrac{\pi\vep_{N}}{2L\sin(\frac{1}{2}\vep_{N}k)}\Big[H_{-k}^{(N)}\!,H_0^{(N)}\Big]\Big),
\end{align*}
where the \emph{lattice hamiltonian Fourier modes} are given by
\begin{align*}
H_k^{(N)} = \tfrac{L}{\pi}\vep_{N}\sum_{x\in \Lambda_{N}} e^{ikx} h^{(N)}_x,
\end{align*}
with $k\in\Gamma_{N}$ in the dual lattice. The inductive family of ucp maps $\alpha_{N+1}^N$ must be carefully chosen to ensure the existence of (meaningful) scaling limits (interestingly the inductive-limit algebra is somewhat insensitive to the specific realization of each map \cite{BlackadarGeneralizedInductiveLimits1}). We focus on two classes (and their combination): a real-space renormalisation group based on wavelets and their scaling functions \cite{DaubechiesTenLecturesOn} which explicitly maintains locality of the algebra and a momentum-space renormalisation group based on a sharp cutoff in momentum space. With this preliminary data in hand we turn now to a summary of the main results.

\subsection{Main results}
\label{sec:main}
Let us give a brief overview of the results presented in this paper. The main results are presented in abridged form as Theorems \ref{thm:A} to \ref{thm:C} accompanied by references to the appropriate statements in the main text. The basis of our results is a scaling limit construction for free lattice fermions using OAR. This entails that our approximation results for CFTs, i.e.~the convergence of the KS approximants and other current-type fields, are currently restricted to free-fermion CFTs, products thereof and certain CFTs embeddable into those (see below). Although, we phrase all results in terms of complex fermion algebras, fully analogous statements also hold for suitably defined self-dual fermion algebras (Majorana fermions), as explained in the main sections. Specifically, we consider scaling limits $\omega_{\infty}$ of the ground state $\omega^{(N)}_{0}$ of a quadratic lattice Dirac Hamiltonian:
\begin{align*}
H^{(N)}_{0} & = \vep_{N}^{-1}\sum_{x\in\Lambda_{N}}\left(a^{(1)\!\!\ \dagger}_{x+\vep_{N}}a^{(2)}_{x} - a^{(1)\!\!\ \dagger}_{x}a^{(2)}_{x} + \textup{h.c.} + \lambda_{N}\left(a^{(1)\!\!\ \dagger}_{x}a^{(1)}_{x} - a^{(2)\!\!\ \dagger}_{x}a^{(2)}_{x}\right)\right),
\end{align*}
in the massless case $\lambda_{N}=0$. Here, $a_{x}, a^{\dag}_{x}$ are the usual annihilation and creation operators generating a complex fermion algebra $\fA_{N}$ associated with the lattice $\Lambda_{N}$.\\

The first main result of the paper concerns the convergence of the KS approximants to the continuum Virasoro generators in projective unitary positive-energy representations with central charge $c=\tfrac{1}{2},1$ associated with free-fermion scaling limits $\omega_{\infty}$. The convergence holds on the natural domain of finite-energy vectors in Fock space (with respect to the chiral conformal Hamiltonian $L_{\pm,0}$ arising from $H^{(N)}_{0}$) \cite{CarpiOnTheUniqueness}. Here and in the following, $\pm$ denotes the chiral and anti-chiral components defined in Section \ref{sec:ffchiral}. In the main text, this result is stated as Theorem \ref{thm:KSconvc}, where all details can be found.

\begin{letterthm}[Convergence of the Koo-Saleur approximants for $\pi_{\infty}$]
\label{thm:A}
Let $\omega_{\infty}$ be a massless scaling limit of the free-fermion ground state $\omega^{(N)}_{0}$ and $\pi_{\infty}$ its GNS representation. The chiral Koo-Saleur approximants, $L^{(N)}_{k,\pm}$, converge strongly to the continuum Virasoro generators, $L_{k,\pm}$, on the dense, common core $\cD_{\textup{fin}}\subset\cH_{\infty}$ spanned by finite-energy vectors of the chiral conformal Hamiltonian $L_{\pm,0}$:
\begin{align*}
\lim_{N\rightarrow\infty}\|(:\!\pi_{\infty}(\alpha^{N}_{\infty}(L^{(N)}_{\pm,k})\!)\!: - L_{\pm,k})\phi\| & = 0,
\end{align*}
for all $\phi\in\cD_{\textup{fin}}$.
\end{letterthm}
The proof of this theorem is a consequence of the basic Lemma \ref{lem:KSconvc} and exploits the momentum-cutoff renormalization group given in Definition \ref{def:momrg} in combination with known results concerning the implementability of Bogoliubov transformations in quasi-free representations of fermion algebras (see Section \ref{sec:bogoliubov}).

A version of this result for the Fock representation with $c=0$ (not of positive energy) is given in Theorem \ref{thm:KSconv} and another one for smeared Virasoro generators in Theorem \ref{thm:KSconvsmeared}. Moreover, a similar result can be obtained for the lattice approximation of Wess-Zumino-Witten (WZW) currents (see Section \ref{sec:WZWcur} and Theorem \ref{thm:U1cursmearedconv}).

\bigskip

The second major result of our paper concerns correlation functions of fermion operators under the action of conformal transformations generated by Virasoro generators. Informally it states that the finite-scale dynamical correlation function approximate the continuum dynamical correlation functions for arbitrary collections of operators in the (chiral) fermion algebra (especially including trace-class generators of quasi-free derivations as in \cite{WitteveenQuantumCircuitApproximations}). This results is referred to as Theorem \ref{thm:corapprox} together with Corollary \ref{cor:corapprox} in the main text.
\begin{letterthm}[Convergence of fermion correlation functions]
\label{thm:B}
Let $s\in C^{\alpha}(\R)$ be a sufficiently regular, compactly supported orthonormal Daubechies scaling function, and $\pi_{\infty}$ be the scaling limit representation of the fermion algebra $\fA_{\infty}$ associated with the free-fermion scaling limit $\omega_{\infty}$. Then, for any convergent sequences $\{A_{N}\}_{N\in\N_{0}}$, $\{B_{N}\}_{N\in\N_{0}}$ with limits $A, B\in\fA_{\infty}$ and uniformly in $t\in\R$ on compact intervals, we have:
\begin{align*}
\lim_{N\rightarrow\infty}(\Omega^{(N)}_{0},\pi^{(N)}(A_{N})\sigma^{(N)}_{t}(\pi^{(N)}(B_{N}))\Omega^{(N)}_{0}) & = (\Omega_{\infty},\pi_{\infty}(A)\sigma_{t}(\pi_{\infty}(B))\Omega_{\infty})\;,
\end{align*}
where $\Omega^{(N)}_{0}$, $\pi^{(N)}$ are the GNS-vector and -representation of $\omega^{(N)}_{0}$. $\sigma^{(N)}_{t}$ and $\sigma_{t}$ are $1$-parameter (semi-)groups of Bogoliubov transformations generated by a (smeared) Koo-Saleur approximant or a (non-)abelian current and their continuum analogues respectively.
\end{letterthm}
The proof of this theorem is implied the existence of the scaling limit of ground states due to Lemma \ref{lem:stateconv} and Lemmata \ref{lem:KSconv}, \ref{lem:KSconvc} \& \ref{lem:KSconvsmeared} because of the semi-group convergence theorem \cite[Theorem 2.16, p. 504]{KatoPerturbationTheoryFor}, cf.~also \cite[Theorem 1.8, p. 141]{EngelAShortCourse}.

\bigskip

The final major result pertains to the correlation functions of the Virasoro algebra in free-fermion CFTs yielding a result comparable to that presented in \cite{ZiniConformalFieldTheories}. Similar to Theorem \ref{thm:B}, we find that the correlation functions of smeared Virasoro generators respectively their exponentials, which generate local Virasoro nets \cite{BrunettiModularStructureAnd, GabbianiOperatorAlgebrasAnd} (here with central charge $c=\tfrac{1}{2},1$), are obtained as limits of ground state correlation functions of the Koo-Saleur approximants. This results is referred to as Theorem \ref{thm:corapproxvir} in the main text.
\begin{letterthm}[Convergence of Virasoro correlation functions]
\label{thm:C}
Let $s\in C^{\alpha}(\R)$ be a sufficiently regular, compactly supported orthonormal Daubechies scaling function, and $\pi_{\infty}$ be the scaling limit representation of the fermion algebra $\fA_{\infty}$ associated with the free-fermion scaling limit $\omega_{\infty}$. Then, for any $n\in\N$ and convergent sequences of smearing functions $X_{N,p}\stackrel{N\rightarrow\infty}{\rightarrow} X_{p}$, $N\in\N_{0}$ and $p=1,...,n$ with sufficient regularity, we have:
\begin{align*}
\lim_{N\rightarrow\infty}(\Omega_{\infty},\prod_{p=1}^{n}:\!(\pi_{\infty}\circ\alpha^{N}_{\infty})(L^{(N)}_{\pm}(X_{N,p}))\!:\Omega_{\infty}) & = (\Omega_{\infty},\prod_{p=1}^{n}L_{\pm}(X_{p})\Omega_{\infty})\;,
\end{align*}
and similarly,
\begin{align*}
\lim_{N\rightarrow\infty}(\Omega_{\infty},\prod_{p=1}^{n}e^{i:(\pi_{\infty}\circ\alpha^{N}_{\infty})(L^{(N)}_{\pm}(X_{N,p})):}\Omega_{\infty}) & = (\Omega_{\infty},\prod_{p=1}^{n}e^{iL_{\pm}(X_{p})}\Omega_{\infty})\;,
\end{align*}
Moreover, the finite-scale correlation functions of the scaling limit state $\omega_{\infty}$ can be approximated in terms of the renormalized finite-scale states:
\begin{align*}
& (\Omega_{\infty},\prod_{p=1}^{n}:\!\pi_{\infty}(\alpha^{N}_{\infty}(L^{(N)}_{\pm}(X_{N,p})\!)\!)\!:\Omega_{\infty}) = \omega^{(N)}_{\infty}(\prod_{p=1}^{n}(L^{(N)}_{\pm}(X_{N,p})-\omega^{(N)}_{\infty}(L^{(N)}_{\pm}(X_{N,p})\!)\!)\!) \\
& = \lim_{M\rightarrow\infty}\omega^{(N)}_{M}(\prod_{p=1}^{n}(L^{(N)}_{\pm}(X_{N,p})-\omega^{(N)}_{\infty}(L^{(N)}_{\pm}(X_{N,p})\!)\!)\!) \\
& = \lim_{M\rightarrow\infty}\omega^{(N+M)}_{0}\!(\prod_{p=1}^{n}(\alpha^{N}_{N+M}(L^{(N)}_{\pm}\!(X_{N,p})\!)\!-\!\omega^{(N+M)}_{0}\!(\alpha^{N}_{N+M}(L^{(N)}_{\pm}\!(X_{N,p})\!)\!)\!)\!) \;,
\end{align*}
and similarly,
\begin{align*}
& (\Omega_{\infty},\prod_{p=1}^{n}e^{i:\pi_{\infty}(\alpha^{N}_{\infty}(L^{(N)}_{\pm}(X_{N,p})\!)\!):}\Omega_{\infty}) = \omega^{(N)}_{\infty}(\prod_{p=1}^{n}e^{iL^{(N)}_{\pm}(X_{N,p})-i\omega^{(N)}_{\infty}(L^{(N)}_{\pm}(X_{N,p})\!)}) \\
& = \lim_{M\rightarrow\infty}\omega^{(N)}_{M}(\prod_{p=1}^{n}e^{iL^{(N)}_{\pm}(X_{N,p})-i\omega^{(N)}_{M}(L^{(N)}_{\pm}(X_{N,p})\!)}) \\
& = \lim_{M\rightarrow\infty}\omega^{(N+M)}_{0}\!(\prod_{p=1}^{n}e^{i\alpha^{N}_{N+M}(L^{(N)}_{\pm}(X_{N,p})\!)-i\omega^{(N+M)}_{0}\!(\alpha^{N}_{N+M}(L^{(N)}_{\pm}(X_{N,p})\!)\!)}) \;,
\end{align*}
where $\omega^{(N)}_{\infty} = \omega_{\infty}\circ\alpha^{N}_{\infty}$, and $L^{(N)}_{\pm}(X_{N})\!=\!\tfrac{1}{2L_{N}}\!\sum_{k\in\Gamma_{N}}\!\hat{X}_{N|k}L^{(N)}_{\pm,k}$ (similarly for $L_{\pm}(X)$).
\end{letterthm}
The proof of this theorem is a consequence of the convergence of the scaling limit procedure for the ground states stated in Lemma \ref{lem:stateconv} in combination with the extension of Theorem \ref{thm:A} to smeared Koo-Saleur approximants and Virasoro generators given in Theorem \ref{thm:KSconvsmeared}. It equally well applies to WZW currents instead because of Theorem \ref{thm:U1cursmearedconv} and its generalization. 

\bigskip

In view of the seminal work of Koo and Saleur \cite{KooRepresentationsOfThe}, it should be noted that the convergence of the KS approximants for free-fermion CFTs was anticipated therein. Specifically, exact agreement with the Virasoro generators in expectation values with respect to lattice ground states in a formal scaling limit construction was found for central charges $c=0,\tfrac{1}{2},-2$. Thus, our results yield a rigorous justification of the computations in \cite{KooRepresentationsOfThe} for the cases $c=0,\tfrac{1}{2}$. Moreover, our construction of scaling limits $\omega_{\infty}$ using OAR provides a detailed mathematical framework for the scaling limit considered by Koo and Saleur, allowing for the approximation of Virasoro generators in an operator sense (strong operator topology) instead of expectation values (weak operator topology). This is achieved by explicitly constructing a Hilbert space of states via the GNS representation of $\omega_{\infty}$ corresponding to scaling limits of the low-lying excited states relative to the lattice ground state $\omega^{(N)}_{0}$. Interestingly, the exact formal as well as rigorous results for the approximation of the Virasoro generators by the Koo-Saleur formula essentially rely on the definition of the correct scaling limit representation by subtracting vacuum/ground-state expectation values, which corresponds to normal ordering as the Virasoro generators and KS approximants are quadratic expression in free fermions. Therefore, it appears reasonable to assume that the failure of the exact computations in \cite{KooRepresentationsOfThe} for general central charge $c\leq1$, is due, at least in part, to the insufficiency of this procedure in the general case. To allow for a more direct comparison between \cite{KooRepresentationsOfThe}, we translate our formulation in terms of fermions into the language of Temperley-Lieb algebras \cite{TemperleyRelationsBetweenThe, GrimmTheSpin12} in Section \ref{sec:tl}.

\bigskip

We point out that in addition to the treatment of currents of WZW models, our results also cover the scaling limit of the transverse-field Ising model, which is intimately related to the two-dimensional classical Ising model \cite{LiebTwoSolubleModels, SchultzTwoDimensionalIsing}, in the following sense: (1) the lattice ground states converge to their appropriate limit on the subalgebra of even observables (by the Jordan-Wigner isomorphism \cite{EvansQuantumSymmetriesOn}), (2) the Koo-Saleur approximants converge to the correct Virasoro generators in the representation associated with the scaling limit (as also observed in \cite{ZiniConformalFieldTheories}). The convergence with respect to the full observable algebra including odd observables \cite{LashkevichSectorsOfMutually, BostelmannCharacterizationOfLocal} and their conformal covariance requires additional work going beyond the scope of the present paper and will be treated in a separate publication adapting results from \cite{SatoHolonomicQuantumFieldsIV}.

\bigskip

Although, in this paper, we restricted our treatment to scaling limits of fermion systems and their quasi-free representations, or systems based thereon such as WZW models, not least to show the validity of our approach to recover conformal symmetry in the scaling limit via OAR in the clearest possible manner, the general method is not restricted to this setting. But, to handle, for example, arbitrary rational CFTs, it will be necessary to address lattice models involving anyonic chains \cite{FeiguinInteractingAnyonsIn, ZiniConformalFieldTheories}. To employ OAR to such models, we need an appropriate family of refining maps. Such a potential family of such maps is induced by the Jones-Wenzl projection \cite{WenzlOnSequencesOf, KauffmanTemperleyLiebRecoupling}. However, this it is beyond the scope of this paper and will be presented elsewhere.

Furthermore, our current treatment of WZW models only allows for a central charge in the range $r\leq c\leq D$, where $D$ is the number of copies of fermions and $r$ is the rank of Lie algebra of the model. To allow for a central charge $c<1$, a natural next step is to check the compatibility with the coset construction \cite{GoddardSymmetricSpacesSugawara, GoddardUnitaryRepresentationsOf}

In addition, it would be interesting to understand whether the analysis presented here extends to the setting of so-called ``symplectic fermions'' \cite{SaleurPolymersAndPercolation, KauschSymplecticFermions, GaberdielALocalLogarithmic} with $c=-2$ corresponding to one of the case where exact formal results were obtained by Koo and Saleur \cite{KooRepresentationsOfThe}. Such an extension appears to be feasible because the algebraic structure of the continuum chiral algebra generated by the symplectic two-component fermion can be discretized by multi-scale decomposition using the adjoint conditional expectations $\alpha^{M}_{N}$, $M>N$, associated with the refining quantum channels (cp.~Section \ref{sec:comp}). But, it needs to be clarified whether such a discretization of the continuum model can be recovered via the renormalization group flow in OAR from a suitable Hamiltonian lattice model, e.g.~the XX spin chain formulation of dense polymers \cite{SaleurPolymersAndPercolation}. This said, we leave a detailed analysis of this case for future work.

\subsection{Comparison with other approaches}
\label{sec:comp}

The strategy outlined above is by no means the only one may adopt to realize QFTs rigorously. Indeed, tremendous effort and major successes have been obtained by focussing on other approaches where the Hilbert space of the quantum field theory is realized directly in the continuum or by probabilistic Euclidean methods (see \cite{SummersAPerspectiveOn} for an overview), and most recently in the context of stochastic quantization (see, e.g., \cite{GubinelliAPDEConstruction, HairerATheoryOf}). Our approach can be viewed as closely related to the quantum mechanical constructions of Glimm-Jaffe and others \cite{GlimmJaffe1, GlimmJaffe2, GlimmQuantumFieldTheory} in contrast with Euclidean approaches. Moreover, considerable dividends are paid by the discretization approach because, if realized correctly, one obtains a sequence of quantum lattice systems which are directly amenable to \emph{quantum simulation} on quantum computers (we comment on this aspect in Remark \ref{rem:stateapprox}, see also \cite{OsborneCFTsim}). This was emphasized by Zini and Wang, who made the first concerted effort to realize CFTs via limits of Hamiltonian lattice models. Below we summarize their approach to scaling limits in the context of OAR.

\bigskip

Zini and Wang \cite{ZiniConformalFieldTheories} define a \emph{low-energy scaling limit} with the following data:
\begin{itemize}
	\item A nested sequence of (energy-bounded) subspaces $\cH^{E}_{N}\subset\cH^{E'}_{N}$, associated with energy scales $E'\geq E\geq0$,
	\item \emph{connecting unitiaries} $\phi^{E}_{N}:\cH^{E}_{N}\rightarrow\cH^{E}_{N+1}$ for sufficiently large $N$, i.e.~a stabilizing dimension $\dim\cH^{E}_{N} = d_{E}$ for all $N>>0$,
	\item and the extension property: $\phi^{E}_{N} = \phi^{E'}_{|\cH^{E}_{N}}$ for $E'\geq E$.
\end{itemize}
The scaling limit results from an inductive limit $\cH^{E}_{\infty}=\varinjlim_{N}\cH^{E}_{N}$ along the connecting unitaries, and, by the extension property, gives a coherent system of energy-bounded Hilbert spaces $\cH^{E}_{\infty}\subset\cH^{E'}_{\infty}\subset\cH_{\infty}$ for $E\leq E'$. The convergence of Hamiltonians is understood in terms of convergent operator sequences $\lim_{N\rightarrow\infty}\lim_{N'\rightarrow\infty}\|H^{(N')}_{0|E}-\phi^{E}_{N\rightarrow N'}H^{(N)}_{0|E}(\phi^{E}_{N\rightarrow N'})^{*}\|_{*}$, for some suitable semi-norm $\|\!\ .\!\ \|_{*}$, which should determine $H_{|E}$ (e.g.~by the second Trotter-Kato approximation theorem \cite{ReedMethodsOfModern1}). Here, $_{|E}$ indicates the restriction of the Hamiltonians to energies below $E$, and $\phi^{E}_{N\rightarrow N'}$ results from the iteration of the connecting unitaries between scales $N\leq N'$. Again, because of coherence in $E$, this defines $H$ acting on $\cH_{\infty}$.

Zini and Wang also define the notion of \emph{strong scaling limit}, which describes the case, when $\phi^{E}_{N\rightarrow N'}$ results from the restriction of an isometry $\phi_{N}:\cH_{N}\rightarrow\cH_{N+1}$ to the energy-bounded subspace.

For $Q_{N}(\!\ .\!\ ,\!\ .\!\ ) = \langle\!\ .\!\ ,O_{N}\!\ .\!\ \rangle_{\cH_{N}}$, the sesquilinear form associated with $O_{N}\in\fA_{N}$, convergence is defined as convergence,
\begin{align*}
\lim_{N\rightarrow\infty}\lim_{N'\rightarrow}Q_{N'}(\phi^{E}_{N\rightarrow N'}(\!\ .\!\ ),\phi^{E}_{N\rightarrow N'}(\!\ .\!\ )) & = Q^{E},
\end{align*}
on $\cH^{E}_{\infty}$ (similar to the convergence of the Hamiltonians).

Zini and Wang observe \cite[p. 898]{ZiniConformalFieldTheories} that compatible embeddings $\tau_{N}:\fA_{N}\rightarrow\fA_{N+1}$,
\begin{align*}
\phi_{N}\circ O_{N} & = \tau_{N}(O_{N}) \circ \phi_{N},
\end{align*}
 in the special case of a strong scaling limit, yield an (operator-)algebraic scaling limit of the observable algebras $\fA_{\infty} = \varinjlim_{N}\fA_{N}$ along these embeddings. They also conjecture that higher-level anyonic chains have this property. 
 
 We now explain how these scaling limits are recovered via OAR: Assuming that we have constructed the inductive limit algebra $\fA_{\infty}$, together with a scaling limit state $\omega^{(N)}_{\infty}$, we can perform the GNS construction at each scale $N$ resulting in a sequence of triples $\{\cH_{\infty}^{(N)}, \pi_{\infty}^{(N)}, \Omega_{\infty}^{(N)}\}$ together with isometries $V^{N}_{N+1}:\cH_{\infty}^{(N)}\rightarrow\cH_{\infty}^{(N+1)}$ induced by $\alpha^{N}_{N+1}$ such that $V^{N}_{N+1}\Omega_{\infty}^{(N)} = \Omega_{\infty}^{(N+1)}$ and:
\begin{align*}
V^{N}_{N+1}\circ \pi_{\infty}^{(N)}(O_{N}) & = \pi_{\infty}^{(N+1)}(\alpha^{N}_{N+1}(O_{N}))\circ V^{N}_{N+1}.
\end{align*}
Thus, we find that the compatibility condition of Zini and Wang is precisely recovered by the inductive-limit structure of the scaling-limit Hilbert space $\cH_{\infty}$ in the operator-algebraic formulation. Moreover, if the GNS representations $\{\cH_{M}^{(N)}, \pi_{M}^{(N)}, \Omega_{M}^{(N)}\}$ of the finitely renormalized states $\omega^{(N)}_{M}$ are unitarily equivalent, e.g.~assuming a Stone-von Neumann-type result at finite scales, that is, if there are (connecting) unitaries,
\begin{align*}
\phi^{N}_{M} : \cH_{M}^{(N)} & \rightarrow \cH_{M+1}^{(N)},
\end{align*}
we can recover the structure of a low-energy scaling limit (the extension property follows again from the GNS construction because $\omega^{(N)}_{M+1} = \omega^{(N+1)}_{M}\circ\alpha^{N}_{N+1}$), such that $N$ plays the role of the energy scale $E$. The overall structure is illustrated in Figure \ref{fig:Wtree}.

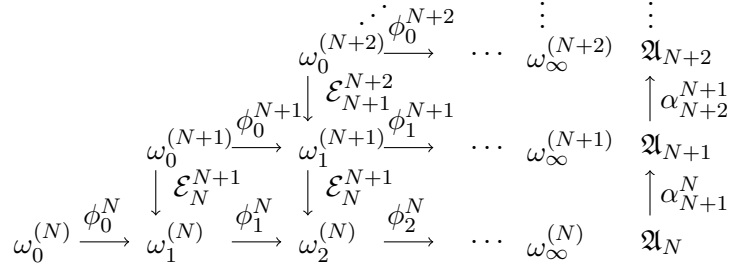
\begin{figure}[ht]
\begin{center}
\scalebox{1}{
\begin{tikzpicture}
	
	\draw (0.25,0.25) node[right]{$\omega^{(N)}_{0}$} (2,0.25) node[right]{$\omega^{(N)}_{1}$} (4,0.25) node[right]{$\omega^{(N)}_{2}$}
	(6.25,0.25) node[right]{$\dots$} (7,0.25) node[right]{$\omega^{(N)}_\infty$} (8.5,0.25) node[right]{$\mathfrak{A}_N$};
	\draw[->] (1.25,0.25) to (1.9,0.25);
	\draw (1.55, 0.25) node[above]{$\phi^{N}_{0}$};
	\draw[->] (3.25,0.25) to (3.9,0.25);
	\draw (3.55, 0.25) node[above]{$\phi^{N}_{1}$};
	\draw[->] (5.25,0.25) to (5.9,0.25);
	\draw (5.55, 0.25) node[above]{$\phi^{N}_{2}$};
	
	\draw (2,1.5) node[right]{$\omega^{(N+1)}_{0}$} (4,1.5) node[right]{$\omega^{(N+1)}_{1}$} (6.25,1.5)
	node[right]{$\dots$} (7,1.5) node[right]{$\omega^{(N+1)}_\infty$}  (8.5,1.5) node[right]{$\mathfrak{A}_{N+1}$};
	\draw[->] (3.25,1.5) to (3.9,1.5);
	\draw (3.75, 1.5) node[above]{$\phi^{N+1}_{0}$};
	\draw[->] (5.25,1.5) to (5.9,1.5);
	\draw (5.75, 1.5) node[above]{$\phi^{N+1}_{1}$};
	
	\draw (2.35,0.975) node[right]{$\mathcal{E}^{N+1}_{N}$} (4.35,0.975) node[right]{$\mathcal{E}^{N+1}_{N}$};
	\draw[->] (2.25,1.2) to (2.25,0.6);
	\draw[->] (4.25,1.2) to (4.25,0.6);
	\draw[<-] (8.75,1.2) to (8.75,0.6); 
	\draw (8.75,0.875) node[right]{$\alpha^{N}_{N+1}$};
	
	\draw (4,2.75) node[right]{$\omega^{(N+2)}_{0}$} (6.25,2.75)
	node[right]{$\dots$} (7,2.75) node[right]{$\omega^{(N+2)}_\infty$} (8.5,2.75) node[right]{$\mathfrak{A}_{N+2}$};
	\draw[->] (5.25,2.75) to (5.9,2.75);
	\draw (5.75, 2.75) node[above]{$\phi^{N+2}_{0}$};
	
	\draw (4.35,2.225) node[right]{$\mathcal{E}^{N+2}_{N+1}$};
	\draw[->] (4.25,2.45) to (4.25,1.85);
	\draw[<-] (8.75,2.45) to (8.75,1.85);
	\draw (8.75,2.125) node[right]{$\alpha^{N+1}_{N+2}$};
	
	\node[right] at (4.75,3.375) {$\iddots$};
	\node[right] at (7.15,3.375) {$\vdots$};
	\node[right] at (8.575,3.375) {$\vdots$};
\end{tikzpicture}
}
\caption{\small Enriched version of Wilson's triangle of renormalization \cite{WilsonTheRenormalizationGroupKondo}: Coarse-graining quantum channels ($\cE$'s) between quantum states, refining quantum channels between observables/fields ($\alpha$'s), and low-energy unitaries ($\phi$'s) between renormalized representations (if existent).}
\label{fig:Wtree}
\end{center}
\end{figure}
(Such structures are also achieved in the operator-algebraic renormalization of a scalar lattice field \cite{MorinelliScalingLimitsOf}.)

It is worth pointing out that such a structure is not immediately available for the operator-algebraic renormalization of (free) massless chiral lattice fermions in 1+1-dimensions considered here, where the finitely renormalized states $\omega^{(N)}_{M}$ are not pure (due to the entanglement of the chiral halves) in contrast with their scaling limit (given by the Hardy projection).

An important difference of our approach with the construction by Zini and Wang is that our renormalized Hilbert spaces $\cH_{M}^{(N)}$ and the connecting unitaries $\phi^{N}_{M}$ are not necessarily associated with energy bounds, and indeed can be far more general, e.g., in the wavelet approach formulated in \cite{StottmeisterOperatorAlgebraicRenormalization} and in Section \ref{sec:waveletrg} the parameter $N$ is associated with a spatial resolution, while energy bounds are similar to the momentum-cutoff approach, see \cite{MorinelliScalingLimitsOf} and Section \ref{sec:momcut}. This flexibility enables us to overcome a major challenge facing the construction of Zini and Wang's strong scaling limit, namely that the addition of irrelevant terms to the lattice discretizations, resulting in lattice-scale artifacts, makes it difficult if not impossible to find connecting unitaries. 

In summary, the core differences between the scaling limit construction of Zini and Wang, and the OAR, may be summarized as follows: 
\begin{itemize}
	\item In the framework of Zini and Wang, the Hilbert space $\cH_{\infty}$ of the scaling limit, in other words the sector of the observables, is fixed from the outset by the specification of the connecting unitaries ($\phi$'s) and the energy-bounded subspaces, or by the connecting isometries in the strong scaling limit. Then, the renormalization of the Hamiltonian and other observables amounts to determining all compatible operators (or sesquilinear forms), i.e.~$\fQ_{\infty}$, on $\cH_{\infty}$. Because the sector is fixed by the provision of the connecting unitaries, quantum states are in essence not renormalized.
	\item In the framework of operator-algebraic renormalization, we fix a minimal set of observables over all scales, i.e.~$\fA_{\infty}$, the consistency of which is described by the refining quantum channels ($\alpha$'s). In turn, the renormalization of states is used to find (all) sensible, compatible states on $\fA_{\infty}$ by starting from initial choices on each scale.
\end{itemize}
In this respect, we observe that the two notions of convergence of Hamiltonians (and observables) relative to the scaling limit $\cH_{\infty}$ can be compared using the GNS isometries, because
\begin{align*}
V^{N}_{N+1}\circ \pi_{\infty}^{(N)}(O_{N})(V^{N}_{N+1})^{*} & \!=\!\pi_{\infty}^{(N+1)}(\alpha^{N}_{N+1}(O_{N}))p_{\im V^{N}_{N+1}}
\end{align*}
where $p_{\im V^{N}_{N+1}}$ is projection onto the image of $V^{N}_{N+1}$.

\bigskip

A second key contribution to the mathematical literature on approximations of CFTs is found in papers building tensor-network approximations via matrix product states \cite{KoenigMatrixProductApproximations1, KoenigMatrixProductApproximations2} in the context of VOAs or wavelets and their scaling functions in a Hamiltonian lattice setup \cite{EvenblyEntanglementRenormalizationAnd, HaegemanRigorousFreeFermion, WitteveenQuantumCircuitApproximations}. Specifically, in the latter group of papers a central role is also played by an inductive system of quantum channels $\alpha^{N}_{M}$, $M>N$, or more precisely their adjoint conditional expectations \cite{EvansCompletelyPositiveQuasi},
\begin{align*}
\alpha^{M}_{N} : \fA_{M} & \longrightarrow \fA_{N}, & M & >N,
\end{align*}
These are used to define coarse-grained fermion fields as well as finite-scale approximants of (trace-class) second-quantized one-particle operators. In that respect, it should be noted that the main theorem of \cite{WitteveenQuantumCircuitApproximations} concerning the approximation of correlation functions of the continuum free fermion field only applies to the insertion of basic fermions in the sense of our Theorem \ref{thm:B}, see \cite[Theorem 4.4, p.~31]{WitteveenQuantumCircuitApproximations}). Thus, more general operators such as Virasoro generators $L_{\pm,k}$ or WZW currents are explicitly excluded. Such restriction do not apply in the OAR approach we adopt here. Another important difference in comparison with the operator-algebraic approach is the need for a continuum model to define the approximants using the conditional expectations $\alpha^{\infty}_{N}$, which is not intrinsically required by our method. It should also be noted that the approximation of correlation functions achieved in \cite{HaegemanRigorousFreeFermion, WitteveenQuantumCircuitApproximations} is deduced from an error bound of the form,
\begin{align*}
|\langle O_{1}\ldots O_{n}\rangle^{(N)}\!-\!\langle O_{1}...O_{n}\rangle_{\text{cont}}|\!\lesssim\!\cO(2^{c N})\!+\!\cO(\vep\log\vep),
\end{align*}
requiring an $\vep$-dependent choice of scaling functions to define the lattice correlation function $\langle O_{1}\dots O_{n}\rangle^{(N)}$, $N = N(\vep)$, and achieve a desired degree of accuracy prohibiting a rigorous proof of convergence in the scaling limit. Although it should be said, that $N$ is expected to decay exponentially with $\vep$, which would entail the aforesaid convergence, according to \cite{WitteveenQuantumCircuitApproximations} and the references therein. Such an adaptation of the chosen scaling functions to the accuracy goal of the approximation is not required by OAR.\\

Finally, let us comment on further results in the literature.\\

In \cite{GransSamuelssonTheActionOf1} (see also \cite{GransSamuelssonTheActionOf2}), there is very interesting work on the KS approximants in the setting of the XXZ spin chain in terms of Temperley-Lieb algebras as in \cite{KooRepresentationsOfThe}. The authors consider a certain ``weak-scaling'' procedure, which we would paraphrase as weak operator convergence of KS approximants and products thereof. Interestingly, the weak convergence is restricted to a certain class of scaling states lattice states, which we believe roughly correspond to the states identified by the GNS representation of the scaling limit $\{\fA_{\infty},\omega_{\infty}\}$ in OAR. Moreover, \cite{GransSamuelssonTheActionOf1} provides a wealth of conjectures, backed by extensive numerics, concerning the limits appearing in the ``weak-scaling'' procedure, and, thus, it would be worthwhile to investigate whether our methods allow for a proof at least further progress, potentially exploiting the connection with the formulation based on Temperley-Lieb algebras for general $c\neq0$ (see Section \ref{sec:tl}). In particular, it would be interesting whether our methods allow to lift the ``weak-scaling'' procedure to a ``strong'' version (in the sense of operator topologies).\\

As mentioned above, there is an ongoing concerted and very successful effort in the Euclidean setting using classical probabilistic methods and, specifically, the concept of \textit{discrete holomorphicity} (see \cite{HonglerConformalFieldTheory, ChelkakCorrelationsOfPrimary, ChelkakConformalInvarianceOf}) and references therein. In \cite{HonglerConformalFieldTheory} it is shown how to explicitly relate discrete holomorphic structures with the continuum Virasoro algebra, thereby establishing a correspondence between correlation functions of lattice local fields and their counterparts in the continuum. As this is analogous to various degrees to what is achieved in the work presented here and the other Hamiltonian frameworks discussed above, it would be interesting to explore potential connections. Concretely, such an analysis could start from the ideas presented \cite{FendleyIntegrabilityAndBraided} and by exploiting the natural connection between Euclidean and Hamiltonian formulation of critical systems via the transfer matrix formalism.

\subsection{Outline of the article}
\label{sec:outline}
In this paper we develop the theory of OAR and apply it to the discrete approximation of conformal symmetries in continuum CFTs. This yields an explicit family of lattice systems and bounds which may be directly exploited to carry out the quantum simulation of CFTs on quantum computers (see our companion paper \cite{OsborneCFTsim}). 

The structure of the paper is as follows: In Section \ref{sec:ffapprox}, we provide a detailed overview of known results on complex and self-dual fermion algebras required for the proofs of our results. Moreover, we introduce the necessary setup for fermions on lattices required for our scaling limit construction. In Section \ref{sec:ffscaling}, we introduce the scaling maps required for OAR: (1) the wavelet renormalization group, (2) the momentum-cutoff renormalization group. In addition, we discuss the compatibility of the wavelet renormalization group with the quasi-local structure of the fermion algebra in the scaling limit, give a basic decay estimate for wavelet approximations, and prove the convergence of the scaling limit of the ground states of the lattice Dirac Hamiltonian. In Section \ref{sec:approxconf}, we discuss Koo-Saleur approximants and their formal scaling limit yielding the Virasoro generators, before we prove their convergence in the scaling limit as alluded to in Theorem \ref{thm:A}. In Section \ref{sec:WZWcur}, we explain the modifications necessary to apply the convergence results to cover WZW currents. Finally, in Section \ref{sec:corr}, we apply the convergence results for the KS approximants to deduce the convergence of (dynamical) correlation functions of various kinds as exemplified by Theorems \ref{thm:B} \& \ref{thm:C}.

\section{Lattice fermions in 1+1-dimensions}
\label{sec:ffapprox}

\subsection{The algebra of canonical anti-commutation relations}
\label{sec:car}
Let us recall some basic structures and results concerning the canonical anti-commutation relations (CAR). For further details we refer to the general references \cite{EvansQuantumSymmetriesOn, BratteliOperatorAlgebrasAnd2}.

The complex CAR algebra $\fA_{\CAR}(\fh)$ of a complex Hilbert space $\fh$, the \textit{one-particle space}, is the universal unital $C^{*}$-algebra generated by anti-linear map,
\begin{align}
\label{eq:annihilation}
a : \fh & \longrightarrow \fA_{\CAR}(\fh),
\end{align}
referred to as \textit{annihilation operators} and subject to,
\begin{align}
\label{eq:car}
\{a(\xi),a^{\dag}(\eta)\} & = \langle\xi,\eta\rangle_{\fh}, & \{a(\xi),a(\eta)\} & = 0 = \{a^{\dag}(\xi),a^{\dag}(\eta)\}, & \xi,\eta & \in\fh,
\end{align}
where $a^{\dag}(\eta) = a(\eta)^{*}$ denotes the adjoint \textit{creation operator}. Each $a(\xi)$ is automatically bounded: $\|a(\xi)\|=\|\xi\|_{\fh}$. We will also refer to the maps $a, a^{\dag}$ as a \textit{complex fermion}. We also recall that there is a $*$-isomorphism $\fA_{\CAR}(\fh)\cong\otimes_{n=1}^{\dim\fh}M_{2}(\C)$ which we come back to in Section \ref{sec:xymodel}.

The irreducible standard \textit{Fock representation} of $\fA_{\CAR}(\fh)$ on the anti-symmetric Fock space $\fF_{\ua}(\fh) = \oplus_{n=0}^{\infty}\wedge^{n}\fh$, with vacuum vector $\Omega_{0}$, is given by:
\begin{align}
\label{eq:fockrep}
a^{\dag}(\xi)\eta_{1}\wedge\dots\wedge\eta_{n} & = \xi\wedge\eta_{1}\wedge\dots\wedge\eta_{n}, & \eta_{1}\wedge\dots\wedge\eta_{n} & \in\wedge^{n}\fh,
\end{align}
where $\eta_{1}\wedge\dots\wedge\eta_{n} = (n!)^{\frac{1}{2}}S_{-}(\eta_{1}\otimes\dots\otimes\eta_{n})$ with the projection, $S_{-}:\fh^{\otimes n}\rightarrow\wedge^{n}\fh$, on the anti-symmetric subspace. Subsequently, we call $P_{\leq n}$ the projection onto the subspace $\oplus_{m=0}^{n}\wedge^{m}\fh$, and $P_{n} = P_{\leq n}-P_{\leq n-1}$. On $\fF_{\ua}(\fh)$, we denote by $(-1)^{F}$ the \textit{parity operator}, i.e.~the unitary operator implementing the grading of $\fA_{\CAR}(\fh)$ defined by: $\alpha_{-1}(a(\xi)) = a(-\xi) = -a(\xi)$, $\xi\in\fh$. We remark that $(-1)^{F}\notin\fA_{\CAR}(\fH)$ unless $\dim\fh<\infty$.

To treat Majorana fermions, we need in addition to $\fA_{\CAR}$ the notion a \textit{self-dual CAR algebra} $\fA_{\SDC}(\fh,C)$ for some (charge) conjugation $C:\fh\rightarrow\fh$. The $C^{*}$-algebra $\fA_{\SDC}(\fH,C)$ is generated by an anti-linear map,
\begin{align}
\label{eq:majorana}
\Psi : \fh & \longrightarrow \fA_{\SDC}(\fh,C),
\end{align}
which we refer to as the \textit{Majorana fermion or operator}, subject to,
\begin{align}
\label{eq:sdc}
\{\Psi(\xi),\Psi(\eta)\} & = \langle\xi,C\eta\rangle_{\fh}, & \Psi(\xi)^{*} & = \Psi(C\xi), & \xi,\eta & \in\fh.
\end{align}
We note that complex and Majorana fermions can be related by a projection, $P:\fh\rightarrow\fh$, such that $CP=(1-P)C$ (called a \textit{basis projection}):
\begin{align}
\label{eq:carsdc}
\fA_{\SDC}(\fh,C) & \cong \fA_{\CAR}(P\fh), & a(\xi) & = \Psi(\xi), & a^{\dag}(\xi) & = \Psi(C\xi), & \xi & \in P\fh.  
\end{align}
For the description of ground states of quadratic Hamiltonians, we need the notion of a quasi-free states $\omega$ on $\fA_{\CAR}(\fh)$ respectively $\fA_{\SDC}(\fh,C)$ (gauge invariant in the first case). In both cases, the state $\omega$ is completely determined by its two-point function, i.e.~$\omega(a(\xi)a^{\dag}(\eta))$ respectively $\omega(\Psi(\xi)\Psi(\eta)^{*})$. 

In the first case, the two-point function defines a positive operator, $0\leq S\leq1$, such that:
\begin{align}
\label{eq:carqfs}
\omega_{S}(a(\xi)a^{\dag}(\eta)) & = \langle\xi,(1-S)\eta\rangle_{\fh}, & \xi,\eta & \in\fh, \\ \nonumber
\omega_{S}(a(\xi_{1})\dots a(\xi_{n})a^{\dag}(\eta_{n+1})\dots a^{\dag}(\eta_{n+m})) & = \delta_{n,m}\det((\langle\xi_{i},(1-S)\eta_{j}\rangle_{\fh})_{i,j=1}^{n}).
\end{align}
Clearly, the vacuum vector $\Omega_{0}$ corresponds to $S=0$. Moreover, $\omega_{S}$ is a pure state if and only if $S$ is a projection.

In the second case, the two-point function also defines a positive operator, $0\leq S\leq 1$, with the additional property $CS = (1-S)C$, such that:
\begin{align}
\label{eq:sdcqfs}
\omega_{S}(\Psi(\xi)\Psi(\eta)^{*}) & = \langle\xi,(1-S)\eta\rangle_{\fh}, & \xi,\eta & \in\fh, \\ \nonumber
\omega_{S}(\Psi(\xi_{1})\dots \Psi(\xi_{n})) & = \pf((\langle\xi_{i},(1-S)C\xi_{j}\rangle_{\fh})_{i,j=1}^{n}),
\end{align}
where $\pf$ denotes the Pfaffian, which we define to vanish for $n$ odd. As before, the state $\omega_{S}$ is pure if and only if $S$ is a (basis) projection.

In both cases, the Gelfand-Naimark-Segal (GNS) representation of pure states, $S^{2}=S$, can be realized on anti-symmetric Fock space by:
\begin{align}
\label{eq:repqfpcar}
\pi_{S}(a(\xi)) & = a((1-S)\xi) + a^{\dag}(JS\xi), & \xi & \in\fh,
\end{align}
for some conjugation $J:\fh\rightarrow\fh$ with $JS = SJ$, respectively,
\begin{align}
\label{eq:repqfpsdc}
\pi_{S}(\Psi(\xi)) & = a((1-S)\xi) + a^{\dag}(CS\xi), & \xi & \in\fh.
\end{align}
Clearly, the notation is consistent with $\pi_{0}$ being the Fock representation \eqref{eq:fockrep}. In the non-pure situation, an analogous realization on the doubled Fock space, $\fF_{a}(\fh^{\oplus 2})=\fF_{a}(\fh)^{\otimes 2}$, via the $\Z_{2}$-twisted tensor product, $\fA_{\CAR}(\fh)^{\otimes_{\Z_{2}}2} = \fA_{\CAR}(\fh^{\oplus 2})$ respectively $\fA_{\SDC}(\fh,C)^{\otimes_{\Z_{2}}2} = \fA_{\SDC}(\fh^{\oplus 2},C\oplus(-C))$, is possible because,
\begin{align}
\label{eq:doubleproj}
P_{S} & = \begin{pmatrix} S & S^{\frac{1}{2}}(1-S)^{\frac{1}{2}} \\ S^{\frac{1}{2}}(1-S)^{\frac{1}{2}} & 1-S \end{pmatrix},
\end{align}
is a (basis) projection \cite{ArakiOnQuasiFree, LundbergQuasiFreeSecond}.

\subsubsection{Implementation of Bogoliubov transformations}
\label{sec:bogoliubov}
Central to our analysis of the recovery of conformal symmetry via lattice approximations is the question of implementability of \textit{Bogoliubov or quasi-free transformations} of the complex and Majorana fermion algebras, $\fA_{\CAR}(\fh)$ respectively $\fA_{\SDC}(\fh, C)$, in quasi-free representations \eqref{eq:repqfpcar} \& \eqref{eq:repqfpsdc}. To this end, we briefly recall some well-known formulas and results and refer to \cite{ArakiOnQuasiFree, LundbergQuasiFreeSecond, RuijsenaarsOnBogolyubovTransformations2, FredenhagenImplementationOfAutomorphisms, CareyOnFermionGauge, EvansQuantumSymmetriesOn} for further details.

In its basic form a Bogoliubov transformation $\alpha_{T}$ is densely defined morphism of the CAR associated with an invertible, bounded operator, $T\in B(\fh)$, on the one-particle space:
\begin{align}
\label{eq:bogmor}
\alpha_{T}(a^{\dag}(\xi)) & = a^{\dag}(T\xi), & \alpha_{T}(a(\xi)) & = a(T^{-1\!\ *}\xi).
\end{align}
Similarly a bounded operator $G\in B(\fh)$ yields a densely defined derivation $\delta_{G}$ by:
\begin{align}
\label{eq:bogder}
\delta_{G}(a^{\dag}(\xi)) & = a^{\dag}(G\xi), & \delta_{G}(a(\xi)) & = -a(G^{*}\xi).
\end{align}
In the Fock representation \eqref{eq:fockrep}, Bogoliubov transformations and their derivations can be implemented as (unbounded) operators on the dense subspace $\cD(\fh) =\!\ _{\alg}\!\oplus_{n=0}^{\infty}\wedge^{n}\fh$ of vectors with finite particle number by multiplicative, $F_{0}$, and additive, $dF_{0}$, second quantization:
\begin{align}
\label{eq:bogmultcar}
\alpha_{T}(a^{\dag}(\xi)) & = \Ad_{F_{0}(G)}(a^{\dag}(\xi)), & \delta_{G}(a^{\dag}(\xi)) & = \ad_{dF_{0}(G)}(a^{\dag}(\xi)),
\end{align}
where
\begin{align}
\label{eq:secondquant}
F_{0}(T)\eta_{1}\wedge\dots\wedge\eta_{n} & = T\eta_{1}\wedge\dots\wedge T\eta_{n}, & F_{0}(T)\Omega_{0} & = \Omega_{0}, \\ \nonumber
dF_{0}(G)\eta_{1}\wedge\dots\wedge\eta_{n} & = \sum_{m=1}^{n}\eta_{1}\wedge\dots\wedge G\eta_{m}\wedge\dots\wedge\eta_{n}, & dF_{0}(G)\Omega_{0} & = 0.
\end{align}
$F_{0}$ and $dF_{0}$ are related by $F_{0} = \exp dF_{0}$ via the identification $T = \exp(G)$, and we have the obvious bound:
\begin{align}
\label{eq:secondquantbound}
\|dF_{0}(G)P_{\leq n}\| & \leq n \|G\|.
\end{align}
Moreover, we can expand $F_{0}$ and $dF_{0}$ in terms of annihilation and creation operators because:
\begin{align}
\label{eq:secondquantaraki}
dF_{0}(G) & = \sum_{i\in\cI}a^{\dag}(G\xi_{i})a(\xi_{i}) = a^{\dag}Aa,
\end{align}
for some orthonormal basis $\{\xi_{i}\}_{i\in\cI}$ of $\fh$, which satisfies:
\begin{align}
\label{eq:secondquantcomm}
[dF_{0}(G_{1}),dF_{0}(G_{2})] & = dF_{0}([G_{1},G_{2}]).
\end{align}
We note that the parity operator is given by $F_{0}(-1) = (-1)^{F}$. 

In the self-dual setting, $\fA_{\SDC}(\fh,C)$, the same ideas apply with further constraints on the admissible one-particle operators. Specifically, with respect to the conjugation $C$, we require:
\begin{align}
\label{eq:sdcconstraints}
CT^{-1\!\ *}C & = T, & CG^{*}C & = -G.
\end{align}
If we write $G = iH$ instead, we will also have $CH^{*}C=-H$. We denote the analogues of $F_{0}$ and $dF_{0}$ by $Q_{0}$ respectively $dQ_{0}$ in this case:
\begin{align}
\label{eq:secondquantarakisdc}
dQ_{0}(G) & = \sum_{i\in\cI}\tfrac{1}{2}\Psi(G\xi_{i})^{*}\Psi(\xi_{i}) = \tfrac{1}{2}\Psi^{*}A\Psi.
\end{align}

Starting from the basic case, it is possible to extend the notion of Bogoliubov transformations and their derivations as well as multiplicative and additive second quantization to larger classes of operators on $\fh$, based on the bound:
\begin{align}
\label{eq:secondquantboundvect}
\|dF_{0}(G)a^{\dag}(\eta_{1})\dots a^{\dag}(\eta_{n})\Omega_{0}\| & \leq \sum_{m=1}^{n}\|\eta_{1}\|\dots\|G\eta_{m}\|\dots\|\eta_{n}\|.
\end{align}
Important examples, that we will make use of, comprise contractions $\|T\|\leq1$ as well as unbounded (essentially) self-adjoint operators $(H,\cD(H))$, such that $G = iH$, or even closed unbounded operators in the case of the generators of the Virasoro algebra.

As pointed out above, we are particularly interested in the implementability of Bogoliubov transformations in the context of quasi-free representations $\pi_{S}$. In this context, a key formula is the additive \textit{normal-ordered second quantization} relative to a pure-state representation $\pi_{S}$, $S^{2}=S$, in analogy with \eqref{eq:secondquantaraki}:
\begin{align}
\label{eq:secondquantarakino}
dF_{S}(G) & = \sum_{i\in\cI}:\pi_{S}(a^{\dag}(G\xi_{i})a(\xi_{i})): \\ \nonumber
 & = \sum_{i\in\cI}\Big(a^{\dag}(S_{+}G\xi_{i})a(S_{+}\xi_{i})) - a^{\dag}(JS_{-}\xi_{i})a(JS_{-}G\xi_{i})) \\[-0.35cm] \nonumber
 &\hspace{1.5cm} + a^{\dag}(S_{+}G\xi_{i})a^{\dag}(JS_{-}\xi_{i}))+a(JS_{-}G\xi_{i}))a(S_{+}\xi_{i}))\Big) \\ \nonumber
 & = a^{\dag}G_{++}a - a^{\dag}G^{t}_{--}a + a^{\dag}G_{+-}a^{\dag}+aG_{-+}a,
\end{align}
with $G^{t} = JG^{*}J$, and similarly,
\begin{align}
\label{eq:secondquantarakinosdc}
dQ_{S}(G) & = \sum_{i\in\cI}:\pi_{S}(\tfrac{1}{2}\Psi(G\xi_{i})^{*}\Psi(\xi_{i})): \\ \nonumber
 & = \tfrac{1}{2}\big(a^{\dag}G_{++}a - a^{\dag}G^{T}_{--}a + a^{\dag}G_{+-}a^{\dag}+aG_{-+}a\big) \\ \nonumber
 & = a^{\dag}G_{++}a + \tfrac{1}{2}\big(a^{\dag}G_{+-}a^{\dag}+aG_{-+}a\big),
\end{align}
with $G^{T} = CG^{*}C$. Here, $S_{-} = S$, $S_{+}=1-S$, and we use block-matrix notation,
\begin{align}
\label{eq:blockmatrix}
G & = S_{+}GS_{+} + S_{+}GS_{-} + S_{-}GS_{+} + S_{-}GS_{-} = \begin{pmatrix} G_{++} & G_{+-} \\ G_{-+} & G_{--} \end{pmatrix},
\end{align}
as well as the standard symbol $:\!\ \cdot\!\ :$ for normal ordering with respect to the Fock vacuum $\Omega_{0}$\footnote{i.e.~creation operators to left and annihilation operators to the right with the appropriate signs due to the CAR.}. These formulas show that $dF_{S}(G)$ and $dQ_{S}(G) $ can be defined as operators on Fock space if the off-diagonal parts $G_{+-}$ and $G_{-+}$ are in the Hilbert-Schmidt class because:
\begin{align}
\label{eq:hilbertschmidt}
\|a^{\dag}G_{+-}a^{\dag}P_{\leq n}\| & \leq (n+2)\|G_{+-}\|_{2}, & \|aG_{+-}aP_{\leq n}\| & \leq n \|G_{+-}\|_{2}.
\end{align}
Notably, the formulas \eqref{eq:secondquantarakino} \& \eqref{eq:secondquantarakinosdc} can be written as:
\begin{align}
\label{eq;vacexpren}
dF_{S}(G) & = \pi_{S}(dF_{0}(G)) - \Tr_{\fh}(G_{--}) = \pi_{S}(dF_{0}(G)) - \omega_{S}(dF_{0}(G)), \\ \nonumber
dQ_{S}(G) & = \pi_{S}(dQ_{0}(G)) - \tfrac{1}{2}\Tr_{\fh}(G_{--}) = \pi_{S}(dQ_{0}(G)) - \omega_{S}(dQ_{0}(G)),
\end{align}
when $G_{--}$ is in the trace class. It is a consequence of these formulas that the commutator identity \eqref{eq:secondquantcomm} acquires an additional central term, called the \textit{Schwinger cocycle} \cite{LundbergQuasiFreeSecond}:
\begin{align}
\label{eq:schwinger}
c_{S}(G_{1},G_{2}) & = \Tr_{\fh}([G_{1},G_{2}]_{--}) = \Tr_{\fh}((G_{1})_{-+}(G_{2})_{+-}-(G_{2})_{-+}(G_{1})_{+-}),
\end{align}
such that
\begin{align}
\label{secondquantcommno}
[dF_{S}(G_{1}),dF_{S}(G_{2})] & = dF_{S}([G_{1},G_{2}]) + c_{S}(G_{1},G_{2}), \\ \nonumber
[dQ_{S}(G_{1}),dQ_{S}(G_{2})] & = dQ_{S}([G_{1},G_{2}]) + \tfrac{1}{2}c_{S}(G_{1},G_{2}).
\end{align}
It precisely the Schwinger cocycle $c_{S}$ that leads to the central charge in the commutation relations of the Virasoro generators. For operators $G_{1}$, $G_{2}$ with off-diagonal parts in the Hilbert-Schmidt class the Schwinger cocycle satisfies the obvious bound:
\begin{align}
\label{eq:schwingerbound}
|c_{S}(G_{1},G_{2})| & \leq \big(\|(G_{1})_{-+}\|_{2}\|(G_{2})_{+-}\|_{2}+\|(G_{2})_{-+}\|_{2}\|(G_{1})_{+-}\|_{2}\big).
\end{align}

\subsection{Lattice fermions at different scales}
\label{sec:latfermi}
We describe spatial fermions\footnote{Or time-zero fermions as opposed to spacetime fermions.} with $s$ components at a (dyadic) resolution $\vep_{N}=2^{-N}\vep_{0}$, $N\in\N_{0}$, for some basic length scale $\vep_{0}$ in a $d$-dimensional hypercubic volume, $\T^{d}_{L} = (S^{1}_{L})^{\times d} = \R^{d}/2L\Z^{d}$, of size $\vol \T^{d}_{L} = (2L)^{d}$, $L>0$, by the one-particle space,
\begin{align}
\label{eq:1psp}
\fh_{N,{\color{blue}\pm}} & = \ltwo(\Lambda_{N})_{{\color{blue}\pm}}\otimes\C^{s}, & \langle\xi,\eta\rangle_{N} & = \sum_{x\in\Lambda_{N}}\langle\xi_{x},\eta_{x}\rangle_{\C^{s}}, 
\end{align}
associated with the lattice,
\begin{align}
\label{eq:latsp}
\Lambda_{N} & = \vep_{N}\{-L_{N},\dots,L_{N}-1\}^{d}\subset\T^{d}_{L}, & \vep_{N}L_{N} & = L.
\end{align}
The subscript ${\color{blue}\pm}$ refers to the unitary action of the translations $\tau^{(N)}_{{\color{blue}\pm}}:\Lambda_{N}\act\fh_{N,{\color{blue}\pm}}$ given by:
\begin{align}
\label{eq:translations}
\tau^{(N)}_{{\color{blue}\pm} |x}(\xi)_{y} & = ({\color{blue}\pm} 1)^{\big[\frac{y-x}{2L}\big]}\xi_{y-x}, & \xi & \in\fh_{N, {\color{blue}\pm}},
\end{align}
where we allow for the occurrence of a non-trivial phase which we also refer to as \textit{boundary condition}\footnote{We restrict attention boundary conditions associated with phases ${\color{blue}\pm} 1$, i.e. unitary representations of the double cover of $\T^{d}_{L}$, because we are interested in translation-invariant quadratic Hamiltonians in the following, but more general phases given by other coverings of $\T^{d}_{L}$ are conceivable.}, because it encodes whether $\xi\in\fh_{N, {\color{blue}\pm}}$ can be considered as periodic (${\color{blue}+} $) or anti-periodic (${\color{blue}-}$) on $\Lambda_{N}$. These two types of boundary conditions are also referred to as \textit{Ramond sector} (${\color{blue}+}$) and \textit{Neveu-Schwarz sector} (${\color{blue}-}$) \cite{DiFrancescoCFTBook, EvansQuantumSymmetriesOn}.

We denote the corresponding algebras of lattice fermions, complex and Majorana, by:
\begin{align}
\label{eq:latfermi}
\fA_{N, {\color{blue}\pm}} & = \fA_{\CAR}(\fh_{N, {\color{blue}\pm}}), & \fB_{N, {\color{blue}\pm}} & = \fA_{\SDC}(\fh_{N, {\color{blue}\pm}},C), 
\end{align}
where $C:\fh_{N, {\color{blue}\pm}}\rightarrow\fh_{N, {\color{blue}\pm}}$ is the charge conjugation that we only specify explicitly in the next subsection. We use the notation,
\begin{align}
\label{eq:rsfermi}
a^{(j)}_{x} & = a(\delta^{(N)}_{x,j}), & \Psi_{x} & = \Psi(\delta^{(N)}_{x,j}),
\end{align}
for the special set of generators of $\fA_{N, {\color{blue}\pm}}$ and $\fB_{N, {\color{blue}\pm}}$ associated with the standard basis of $\fh_{N, {\color{blue}\pm}}$:
\begin{align}
\label{eq:basisrs}
\delta^{(N)}_{x,j} & = \delta^{(N)}_{x}\otimes(0,\dots,\!\!\!\!\!\!\underbrace{1}_{j\textup{th component}}\!\!\!\!\!\!,\dots,0) & x & \in\Lambda_{N}, j=1,...,s.
\end{align}

The lattice Fourier transform,
\begin{align}
\label{eq:lattF}
\hat{\xi}^{(j)}(k) = \sum_{x\in\Lambda_{N}}e^{-ikx}\xi^{(j)}(x) & = \langle e_{k,j},\xi\rangle_{N}, & \xi = (\xi^{(1)},\dots,\xi^{(s)})&\in\ltwo(\Lambda_{N})_{{\color{blue}\pm}}\otimes\C^{s},
\end{align}
provides a unitary identification $\fh_{N, {\color{blue}\pm}}\cong\ltwo(\Gamma_{N,{\color{blue}\pm}},(2L_{N})^{-d})\otimes\C^{s}$ where
\begin{align}
\label{eq:latmom}
\Gamma_{N, {\color{blue}+}} & = \tfrac{\pi}{L}\{-L_{N},...,0,...,L_{N}-1\}^{d}, & \Gamma_{N, {\color{blue}-}} & = \tfrac{\pi}{L}\{-L_{N}+\tfrac{1}{2},...,0,...,L_{N}-\tfrac{1}{2}\}^{d},
\end{align}
is the dual momentum-space lattice $\Lambda_{N}$, and
\begin{align}
\label{eq:basismom}
e_{k,j} & = e^{ik(\!\ .\!\ )}\otimes(0,\dots,\!\!\!\!\!\!\underbrace{1}_{j\textup{th component}}\!\!\!\!\!\!,\dots,0), & k & \in\Gamma_{N, {\color{blue}\pm}}, j=1,\dots,s,
\end{align}
is the plane-wave basis of $\fh_{N, {\color{blue}\pm}}$. 

We extend the Fourier transform and its inverse to $\fA_{N, {\color{blue}\pm}}$ and $\fB_{N, {\color{blue}\pm}}$ by linearity:
\begin{align}
\label{eq:Ffields}
\hat{a}^{(j)}_{k} & = a(e_{k,j}) = \sum_{x\in\Lambda_{N}}e^{-ikx}\phi^{(j)}_{x}, & a^{(j)}_{x} & = \hat{a}(e_{-x,j}) = \tfrac{1}{2L_{N}}\sum_{k\in\Gamma_{N}}e^{ikx}\hat{a}^{(j)}_{k},  \\ \nonumber
\hat{\Psi}^{(j)}_{k} & = \Psi(e_{k,j}) = \sum_{x\in\Lambda_{N}}e^{-ikx}\Psi^{(j)}_{x}, & \Psi^{(j)}_{x} & = \hat{\Psi}(e_{-x,j}) = \tfrac{1}{2L_{N}}\sum_{k\in\Gamma_{N}}e^{ikx}\hat{\Psi}^{(j)}_{k},
\end{align}
for $x\in\Lambda_{N}$, $k\in\Gamma_{N}$, $j =1,\dots,s$. Here, we also introduced the dual plane-wave basis $\{e_{x,j}\ |\ x\in\Lambda_{N}, j=1,\dots,s\}$ by analogy with \eqref{eq:basismom}. Thus, the lattice Fourier transform provides us with $*$-isomorphims,
\begin{align}
\label{eq:Ffermi}
\fA_{N.\pm} &\!\cong\!\fA_{\CAR}(\ltwo(\Gamma_{N,{\color{blue}\pm}},(2L_{N})^{-d})\otimes\C^{s}), & \fB_{N,{\color{blue}\pm}} &\!\cong\!\fA_{\SDC}(\ltwo(\Gamma_{N,{\color{blue}\pm}},(2L_{N})^{-d})\otimes\C^{s},C),
\end{align}
which can be expressed by: $\hat{a}(\hat{\xi}) = a(\xi)$ respectively $\hat{\Psi}(\hat{\xi}) = \Psi(\xi)$ for $\xi\in\fh_{N,{\color{blue}\pm}}$.

In view of the construction of the scaling limit in Section \ref{sec:ffscaling}, we also allow for $N=\infty$, where we put:
\begin{align}
\label{eq:1pcont}
\fh_{\infty,{\color{blue}\pm}} & = L^{2}(\T^{d}_{L})_{{\color{blue}\pm}}\otimes\C^{s}\cong\ltwo(\Gamma_{\infty,{\color{blue}\pm}},(2L)^{-d})\otimes\C^{s}, & \Gamma_{\infty,{\color{blue}+}} & = \tfrac{\pi}{L}\Z^{d}, & \Gamma_{\infty,{\color{blue}-}} & = \tfrac{\pi}{L}(\Z+\tfrac{1}{2})^{d},
\end{align}
with the continuum Fourier transform,
\begin{align}
\label{eq:contF}
\hat{\xi}^{(j)}_{k} & = \cF_{L}[\xi^{(j)}]_{k} = \int_{\T^{d}_{L}}dx\!\ e^{-ikx}\xi^{(j)}_{x} = \langle e_{k,j}, f\rangle_{L^{2}(\T^{d}_{L})}, & k & \in\Gamma_{\infty,{\color{blue}\pm}},
\end{align}
for periodic and anti-periodic functions in $L^{2}(\T^{d}_{L})_{{\color{blue}\pm}}$. As above, we have:
\begin{align}
\label{eq:contfermi}
\fA_{\infty,{\color{blue}\pm}} & = \fA_{\CAR}(\fh_{\infty,{\color{blue}\pm}}), & \fB_{\infty,{\color{blue}\pm}} & = \fA_{\SDC}(\fh_{\infty,{\color{blue}\pm}},C). 
\end{align}

\subsection{The free lattice Dirac Hamiltonian in 1+1-dimensions}
\label{sec:latdirac}
Since our principal interest lies with models that exhibit local conformal symmetries in the scaling limit, we restrict from this point on to one spatial dimension ($d=1$) and two-component fermions ($s=1$). Specifically, we undertake a detailed analysis of a discretized version of the free Dirac Hamiltonian on the lattice $\Lambda_{N}$ as various models of conformal field theory can be realized in terms of free fermions in the continuum \cite{DiFrancescoCFTBook, EvansQuantumSymmetriesOn, XuSomeResultsOn, WassermannOperatorAlgebrasAnd}. In the staggered lattice approximation\footnote{We use the forward and backward difference to approximate the derivative of the two components respectively. This avoids the notorious fermion doubling in lattice approximations of continuum fermion models.} \cite{SusskindLatticeFermions} the massive Dirac-Hamiltonian as an element of $\fA_{N,{\color{blue}\pm}}$ is given by:
\begin{align}
\label{eq:stagH}
H^{(N)}_{0} & = \vep_{N}^{-1}\sum_{x\in\Lambda_{N}}\left(a^{(1)\!\!\ \dagger}_{x+\vep_{N}}a^{(2)}_{x} - a^{(1)\!\!\ \dagger}_{x}a^{(2)}_{x} + \textup{h.c.} + \lambda_{N}\left(a^{(1)\!\!\ \dagger}_{x}a^{(1)}_{x} - a^{(2)\!\!\ \dagger}_{x}a^{(2)}_{x}\right)\right).
\end{align}
Here, $\lambda_{N}\geq0$ is a dimensionless ``lattice mass'' parameter, and ``h.c.'' denotes the hermitean conjugate of the preceding terms. Invoking the lattice Fourier transform, the Hamiltonian \eqref{eq:stagH} can be written in $2\times 2$-matrix notation as a second quantized operator:
\begin{align}
\label{eq:FstagH}
H^{(N)}_{0} & = \tfrac{1}{2L}\sum_{k\in\Gamma_{N,{\color{blue}\pm}}}\begin{pmatrix} \hat{a}^{(1)}_{k} \\ \hat{a}^{(2)}_{k} \end{pmatrix}^{\!\!*} \underbrace{\begin{pmatrix} \lambda_{N} & e^{-ik\vep_{N}}-1 \\ e^{ik\vep_{N}}-1 & -\lambda_{N} \end{pmatrix}}_{= n_{\lambda_{N}}(k)\cdot\sigma = h^{(N)}_{0}(k)} \begin{pmatrix} \hat{a}^{(1)}_{k} \\ \hat{a}^{(2)}_{k} \end{pmatrix} = \vep_{N}^{-1}dF_{0}(h^{(N)}_{0}),
\end{align}
where $n_{\lambda_{N}}(k) = (\cos(\vep_{N}k)-1)e_{x} + \sin(\vep_{N}k)e_{y} + \lambda_{N}e_{z}\in\R^{3}$, and $\sigma$ is the vector of Pauli matrices:
\begin{align}
\label{eq:pauli}
\sigma_{x} & = \begin{pmatrix} 0 & 1 \\ 1 & 0 \end{pmatrix}, & \sigma_{y} & = \begin{pmatrix} 0 & -i \\ i & 0 \end{pmatrix}, & \sigma_{z} & = \begin{pmatrix} 1 & 0 \\ 0 & -1 \end{pmatrix}.
\end{align}
The spectral decomposition of the \textit{one-particle Hamiltonian} $h^{(N)}_{0}(k)$,
\begin{align}
\label{eq:1pHdiag}
h^{(N)}_{0}(k) & = \omega_{\lambda_{N}}(k)\Big(P^{(+)}_{\lambda_{N}}(k)-P^{(-)}_{\lambda_{N}}(k)\Big),
\end{align}
is given in terms of the spectral projections\footnote{Apart from $k=0$ and $\lambda_{N}=0$ where $h^{(N)}_{0}(0)=0$ which is consistent with $P^{(\pm)}_{0}(0) = \tfrac{1}{2}\1_{2}$.},
\begin{align}
\label{eq:1pproj}
P^{(\pm)}_{\lambda_{N}}(k) & = \tfrac{1}{2\omega_{\lambda_{N}}(k)}\big(\omega_{\lambda_{N}}(k)\1_{2}\pm h^{(N)}_{0}(k)\big)
\end{align}
with the lattice dispersion relation:
\begin{align}
\label{eq:lattdisp}
\omega_{\lambda_{N}}(k) & = \left(\lambda_{N}^{2}+4\sin(\tfrac{1}{2}\vep_{N}k)^{2}\right)^{\frac{1}{2}}.
\end{align}
At $k=0$, we have $h^{(N)}_{0}(0)=\lambda_{N}\sigma_{z}$, and we diagonalize $h^{(N)}_{0}(k)$ for all $k\in\Gamma_{N,{\color{blue}\pm}}$ by:
\begin{align}
\label{eq:stagHbog}
\gamma_{\lambda_{N}}(k) & = \Big(\tfrac{1}{2\omega_{\lambda_{N}}(k)}\Big)^{\!\frac{1}{2}}\begin{pmatrix} (\omega_{\lambda_{N}}(k)+\lambda_{N})^{\frac{1}{2}} & i\sign(k)(\omega_{\lambda_{N}}(k)-\lambda_{N})^{\frac{1}{2}}e^{-\frac{i}{2}\vep_{N}k} \\ i\frac{2\sin(\frac{1}{2}\vep_{N}k)}{(\omega_{\lambda_{N}}(k)+\lambda_{N})^{\frac{1}{2}}}e^{\frac{i}{2}\vep_{N}k} & \sign(k)\frac{2\sin(\frac{1}{2}\vep_{N}k)}{(\omega_{\lambda_{N}}(k)-\lambda_{N})^{\frac{1}{2}}} \end{pmatrix} \\ \nonumber
& = \cosh(\vartheta_{k})^{-\frac{1}{2}}\begin{pmatrix} \cosh(\frac{1}{2}\vartheta_{k}) & i\sinh(\frac{1}{2}\vartheta_{k})e^{-\frac{i}{2}\vep_{N}k} \\ i\sinh(\frac{1}{2}\vartheta_{k})e^{\frac{i}{2}\vep_{N}k} & \cosh(\frac{1}{2}\vartheta_{k}) \end{pmatrix},
\end{align}
i.e.~$\gamma_{\lambda_{N}}(k)^{*}h^{(N)}_{0}(k)\gamma_{\lambda_{N}}(k) = \omega_{\lambda_{N}}(k)\sigma_{z}$. In the second line, we introduced lattice rapidities, $2\sin(\tfrac{1}{2}\vep_{N}k)=\lambda_{N}\sinh(\vartheta_{k})$, to exemplify the unitarity of $\gamma_{\lambda_{N}}(k)$ for $\lambda_{N}\neq0$. Thus, the Bogoliubov transformation,
\begin{align}
\label{eq:lattB}
(\hat{c}^{(1)}_{k}, \hat{c}^{(2)}_{k})^{t} & = \gamma_{\lambda_{N}}(k)^{*}(\hat{a}^{(1)}_{k}, \hat{a}^{(2)}_{k})^{t},
\end{align}
diagonalizes the Hamiltonian $H^{(N)}_{0}$:
\begin{align}
\label{eq:diagH}
H^{(N)}_{0} & = \tfrac{1}{2L}\sum_{k\in\Gamma_{N,{\color{blue}\pm}}}\omega_{\lambda_{N}}(k)(\hat{c}^{(1)\!\!\ \dagger}_{k}\hat{c}^{(1)}_{k}-\hat{c}^{(2)\!\!\ \dagger}_{k}\hat{c}^{(2)}_{k}).
\end{align}
The many-body ground state of $H^{(N)}_{0}$, i.e. the lattice vacuum $\Omega^{(N)}_{0,{\color{blue}\pm}}$ at scale $N$, corresponds to the half-filled Fock vacuum (\textit{Fermi sea}):
\begin{align}
\label{eq:Fsea}
\hat{c}^{(1)}_{k}\Omega^{(N)}_{0,{\color{blue}\pm}} & = 0, & \hat{c}^{(2)\!\!\ \dagger}_{k}\Omega^{(N)}_{0,{\color{blue}\pm}} & = 0, & \forall k&\in\Gamma_{N,{\color{blue}\pm}}.
\end{align}
The unitary that implements the Bogoliubov transform \eqref{eq:lattB} and relates the lattice vacuum $\Omega^{(N)}_{0,{\color{blue}\pm}}$ to the Fock vacuum of $\fA_{N,{\color{blue}\pm}}$ is given by:
\begin{align}
\label{eq:lattBimp}
U_{\lambda_{N}}	& \!=\!\prod_{k\in\Gamma_{N,{\color{blue}\pm}}}\!\!\!\!\exp\!\Big(\!\tfrac{1}{2L_{N}}\!\big(\nu_{k}\hat{a}^{(2)}_{k}\hat{a} ^{(1)\!\!\ \dag}_{k}\!\!+\!\mu_{k}\hat{a}^{(1)}_{k}\hat{a}^{(2)\!\!\ \dag}_{k}\big)\!\!\Big)\!=\!\exp\!\Big(\!\tfrac{1}{2L_{N}}\!\!\!\!\sum_{k\in\Gamma_{N,{\color{blue}\pm}}}\!\!\!\!\nu_{k}\hat{a}^{(2)}_{k}\hat{a}^{(1)\!\!\ \dag}_{k}\!\!+\!\mu_{k}\hat{a}^{(1)}_{k}\hat{a}^{(2)\!\!\ \dag}_{k}\Big)
\end{align}
where
\begin{align}
\label{eq:lattBpar}
(\nu_{k}\mu_{k})^{\frac{1}{2}} & = \cosh^{-1} \Big(\tfrac{\cosh(\frac{1}{2}\vartheta_{k})}{\cosh(\vartheta_{k})^{\frac{1}{2}}}\Big),\ \big(\tfrac{\nu_{k}}{\mu_{k}}\big)^{\frac{1}{2}} = -ie^{-\frac{i}{2}\vep_{N}k}.
\end{align}
It is evident from the expression for $U_{\lambda_{N}}$ that the lattice vacuum is of even parity: $(-1)^{F}\Omega^{(N)}_{0,{\color{blue}\pm}}\!=\!\Omega^{(N)}_{0,{\color{blue}\pm}}$. Alternatively, we can invoke another Bogoliubov transformation, commonly used in quantum field theory,
\begin{align}
\label{eq:qftB}
\tilde{c}^{(1)}_{k} & = \hat{c}^{(1)}_{k}, & \tilde{c}^{(2)}_{k} & = \hat{c}^{(2)\!\!\ \dagger}_{k}, & \forall k&\in\Gamma_{N,{\color{blue}\pm}},
\end{align}
to obtain a description of the Fermi sea as a standard Fock vacuum:
\begin{align}
\label{eq:qftvac}
\tilde{c}^{(1)}_{k}\Omega^{(N)}_{0,{\color{blue}\pm}} & = 0, & \tilde{c}^{(2)}_{k}\Omega^{(N)}_{0,{\color{blue}\pm}} & = 0, & \forall k&\in\Gamma_{N,{\color{blue}\pm}}.
\end{align}
The Bogoliubov transformation \eqref{eq:qftB} exemplifies the need for an additive energy renormalization of the Hamiltonian $H^{(N)}_{0}$ in the scaling limit $N\rightarrow\infty$:
\begin{align}
\label{eq:qftH}
H^{(N)}_{0} & = \tfrac{1}{2L}\sum_{k\in\Gamma_{N,{\color{blue}\pm}}}\omega_{\lambda_{N}}(k)(\tilde{c}^{(1)\!\!\ \dagger}_{k}\tilde{c}^{(1)}_{k}+\tilde{c}^{(2)\!\!\ \dagger}_{k}\tilde{c}^{(2)}_{k}) - \vep_{N}^{-1}\sum_{k\in\Gamma_{N,{\color{blue}\pm}}}\omega_{\lambda_{N}}(k).
\end{align}
The lattice vacuum $\Omega^{(N)}_{0,{\color{blue}\pm}}$ determines a quasi-free state on the fermion algebra $\fA_{N,{\color{blue}\pm}}$:
\begin{align}
\label{eq:Fgroundstate}
\omega^{(N)}_{0,{\color{blue}\pm}}(a(\xi)a^{\dag}(\eta)\!) &\!=\!(\Omega^{(N)}_{0,{\color{blue}\pm}}\!,\hat{a}(\hat{\xi})\hat{a}^{\dag}(\hat{\eta})\Omega^{(N)}_{0,{\color{blue}\pm}})\!=\!\tfrac{1}{2 L_{N}}\!\!\!\!\sum_{k\in\Gamma_{N,{\color{blue}\pm}}}\!\!\!\langle\hat{\xi}_{k},\!P^{(+)}_{\lambda_{N}}\!(k)\hat{\eta}_{k}\rangle_{\C^{2}}\!=\!\langle\xi,\!(1\!-\!S^{(N)}_{0})\eta\rangle_{N},
\end{align}
where the one-particle operator is specified in momentum-space by $S^{(N)}_{0}(k) = P^{(-)}_{\lambda_{N}}(k)$.

\bigskip

From this expression it is easy to obtain the massless lattice vacuum in the limit $\lambda_{N}\rightarrow0+$, although the Ramond sector acquires a degeneracy in the zero mode ($k=0$) because $h^{(N)}_{0}(0) = 0$:
\begin{align}
\label{eq:Fmassless}
P^{(+)}_{\lambda_{N}=0}(k) & = \tfrac{1}{2}(\mathds{1}_{2} + \sign(k)(-\sin(\tfrac{1}{2}\vep_{N}k)\sigma_{x} + \cos(\tfrac{1}{2}\vep_{N}k)\sigma_{y})) \\ \nonumber
 & = \tfrac{1}{2}\begin{pmatrix} 1 & -i\sign(k)e^{-\frac{i}{2}k\vep_{N}} \\ i\sign(k)e^{\frac{i}{2}k\vep_{N}} & 1 \end{pmatrix},
\end{align}
with the convention $\sign(0) = 0$. This way, $P^{(+)}_{\lambda_{N}=0}(0) = \tfrac{1}{2}\mathds{1}_{2}$ reflects a uniform probability distribution of the quasi-free state $\omega^{(N)}_{0,{\color{blue}+}}$ at $k=0$ compatible with the chiral decomposition of $\fh_{N,{\color{blue}\pm}}$ (see Section \ref{sec:ffchiral}). But, the state is not pure because $P^{(+)}_{\lambda_{N}=0}$ is not a projection\footnote{Another possibility is to enforce purity by putting $\sign(0) = 1$ to have $P^{(+)}_{\lambda_{N}=0}(0) = \tfrac{1}{2}(\mathds{1}_{2}+\sigma_{y})$ which is also compatible with the chiral decomposition.}.

\subsection{Chiral decomposition}
\label{sec:ffchiral}
The chiral parts of the lattice fermion algebras result from applying the chiral projectors $p_{\pm} = \tfrac{1}{2}(\mathds{1}_{2}\pm\sigma_{y})$\footnote{In our parametrization of the Dirac Hamiltonian $\gamma_{5} = \sigma_{y}$.} to the one-particle space $\fh_{N,{\color{blue}\pm}}$:
\begin{align}
\label{eq:chiral1p}
\fh_{N,{\color{blue}\pm}} & = \ltwo(\Lambda_{N})_{{\color{blue}\pm}}\otimes\CC^{2}= \ltwo(\Lambda_{N})_{{\color{blue}\pm}}\otimes p_{+}\CC^{2}\oplus\ltwo(\Lambda_{N})_{{\color{blue}\pm}}\otimes p_{-}\CC^{2} = \fh^{(+)}_{N,\pm}\oplus\fh^{(-)}_{N,\pm}.
\end{align}
Using the eigenvectors $e_{\pm} = 2^{-\frac{1}{2}}(e_{1}\pm i e_{2})$ of $p_{\pm}$, we can introduce the complex chiral fermions, $\fA^{(\pm)}_{N,{\color{blue}\pm}} = \fA_{\CAR}(\fh^{(\pm)}_{N,{\color{blue}\pm}})$, by
\begin{align}
\label{eq:chiralF}
\psi_{\pm}(\xi) = a(\xi e_{\pm}) & = 2^{-\frac{1}{2}}(a^{(1)}(\xi) \mp i a^{(2)}(\xi)) = \langle e_{\pm}, (a^{(1)}(\xi), a^{(2)}(\xi))^{t}\rangle.
\end{align}
At this point, we also fix the charge conjugation $C$ (for all $N\leq\infty$) as:
\begin{align}
\label{eq:majconj}
C\xi & = \sigma_{z}\overline{\xi}, & \xi & \in\fh_{N,{\color{blue}\pm}},
\end{align}
which acts as complex conjugation with respect to the chiral decomposition:
\begin{align}
\label{eq:chiralconj}
C\xi & = \overline{\xi_{+}}e_{+}+\overline{\xi_{-}}e_{-},
\end{align}
for $\xi_{\pm} = \langle e_{\pm},\xi\rangle$ and, thus, $a(C\xi) = \psi_{+}(\overline{\xi_{+}})+\psi_{-}(\overline{\xi_{-}})$. This should be contrasted with the standard conjugation, $\xi\mapsto\overline{\xi}$, on $\fh_{N}$ that satisfies $a(\overline{\xi}) = \psi_{+}(\overline{\xi_{-}})+\psi_{-}(\overline{\xi_{+}})$.

In contrast with infinite-volume continuum theory, the massless ground-state projector $P^{(+)}_{\lambda_{N}=0}$ does not factorize w.r.t.~the chiral parts, in other words the massless lattice vacuum is entangled relative to the chiral parts:
\begin{align}
\label{eq:chiralentanglement}
p_{\pm}P^{(+)}_{\lambda_{N}=0}(k)p_{\pm} & = \tfrac{1}{2}(1\pm\sign(k)\cos(\tfrac{1}{2}\vep_{N}k))p_{\pm}, \\ \nonumber
p_{\pm}P^{(+)}_{\lambda_{N}=0}(k)p_{\mp} & = \pm\tfrac{i}{2}\sign(k)\sin(\tfrac{1}{2}\vep_{N}k)\underbrace{\tfrac{1}{2}(\sigma_{z}\pm i\sigma_{x})}_{=n_{\pm}}.
\end{align}
In terms of these, the Hamiltonian \eqref{eq:FstagH} takes the form:
\begin{align}
\label{eq:FstagHchiral}
H^{(N)}_{0} & = \tfrac{1}{2L}\!\!\sum_{k\in\Gamma_{N,{\color{blue}\pm}}}\!\!\begin{pmatrix} \hat{\psi}_{+|k} \\ \hat{\psi}_{-|k} \end{pmatrix}^{\dagger}\!\! \begin{pmatrix} \sin(\vep_{N}k) & \!\!\!\!-i(\cos(\vep_{N}k)-1) + \lambda_{N} \\ i(\cos(\vep_{N}k)-1) + \lambda_{N} & -\sin(\vep_{N}k) \end{pmatrix}\!\! \begin{pmatrix} \hat{\psi}_{+|k} \\ \hat{\psi}_{-|k} \end{pmatrix},
\end{align}
which, in the massless case ($\lambda_{N}=0$), is compatible \eqref{eq:sdcconstraints} with the charge conjugation $C$ which allows for the identification of the self-dual components, i.e.~the Majorana fermions $\fB^{(\pm)}_{N,{\color{blue}\pm}}=\fA_{\SDC}(\fh^{(\pm)}_{N,{\color{blue}\pm}},C)$, of the complex massless chiral fields (see Section \ref{sec:sdc}).

\bigskip

Because of the entanglement of the massless lattice vacuum $\omega^{(N)}_{0,{\color{blue}\pm}}$ relative to the chiral decomposition, its restriction to chiral components defines non-pure quasi-free states ($p_{\pm}P^{(+)}_{\lambda_{N}=0}(k)p_{\pm}$ is positive but not a projection):
\begin{align}
\label{eq:chiral2p}
\omega^{(N)}_{0,{\color{blue}\pm}}(\psi_{\pm}(\xi)\psi_{\pm}^{\dagger}(\eta)) & = \tfrac{1}{2 L_{N}}\sum_{k\in\Gamma_{N,{\color{blue}\pm}}}\tfrac{1}{2}(1\pm\sign(k)\cos(\tfrac{1}{2}\vep_{N}k))\overline{\hat{\xi}(k)}\hat{\zeta}(k),
\end{align}
for $\xi,\eta\in\fh^{(\pm)}_{N,{\color{blue}\pm}}\cong\ltwo(\Gamma_{N,{\color{blue}\pm}},(2 L_{N})^{-1})$.

\subsection{The Majorana mass term}
\label{sec:majmass}
As stated above, the Hamiltonian \eqref{eq:stagH} is a discretized version of the massive Dirac Hamiltonian using the algebra $\fA_{N,{\color{blue}\pm}}$, which is only compatible with the Majorana fermions $\fB_{N,{\color{blue}\pm}}$ for vanishing lattice mass $\lambda_{N}=0$. To allow for a massive Hamiltonian in terms of the Majorana fermions $\fB_{N,{\color{blue}\pm}}$, the mass term (proportional to $\lambda_{N}$) needs to be compatible with self-duality \eqref{eq:sdcconstraints}. This is achieved by replacing the Dirac mass term in \eqref{eq:stagH} with:
\begin{align}
\label{eq:majmass}
H^{(N)}_{\textup{mass}} & = \vep_{N}^{-1}\lambda_{N}\sum_{x\in\Lambda_{N}}\!\!\!\left(a^{(2)\!\!\ \dagger}_{x}a^{(1)}_{x} + a^{(1)\!\!\ \dagger}_{x}a^{(2)}_{x}\!\right) = \vep_{N}^{-1}\lambda_{N}\sum_{x\in\Lambda_{N}}\!(-i)\!\left(\psi^{\dagger}_{+|x}\psi_{-|x} - \psi^{\dagger}_{-|x}\psi_{+|x}\right) \\ \nonumber
 & = \tfrac{1}{2L}\!\!\sum_{k\in\Gamma_{N,{\color{blue}\pm}}}\!\!\begin{pmatrix} \hat{\psi}_{+|k} \\ \hat{\psi}_{-|k} \end{pmatrix}^{\dagger}\!\! \begin{pmatrix} 0 & \!\!\!\!-i\lambda_{N} \\ i \lambda_{N} & 0 \end{pmatrix}\!\! \begin{pmatrix} \hat{\psi}_{+|k} \\ \hat{\psi}_{-|k} \end{pmatrix},
\end{align}
which results in the modified spectral projections,
\begin{align}
\label{eq:majproj}
P^{(\pm)}_{\lambda_{N}}(k) = \tfrac{1}{2\omega_{\lambda_{N}}(k)}(\omega_{\lambda_{N}}(k)\mathds{1}_{2} \pm n_{\lambda_{N}}(k)\cdot\sigma),
\end{align}
with $n_{\lambda_{N}}(k) = (\cos(\vep_{N}k)-1)e_{x} + \sin(\vep_{N}k)e_{y} + \lambda_{N}e_{x}$, and the modified dispersion relation,
\begin{align}
\label{eq:majdis}
\omega_{\lambda_{N}}(k) & = \left(\lambda_{N}^{2}+4(1-\lambda_{N})\sin(\tfrac{1}{2}\vep_{N}k)^{2}\right)^{\frac{1}{2}}.
\end{align}
The modified spectral projections satisfy $CP^{(+)}_{\lambda_{N}} = P^{(-)}_{\lambda_{N}}C$, consistent with \eqref{eq:sdcqfs}, and, thus, \eqref{eq:Fgroundstate} defines a quasi-free state, also denoted $\omega^{(N)}_{0,{\color{blue}\pm}}$ by a slight abuse of notation, on $\fB_{N,{\color{blue}\pm}}$.

\subsection{The self-dual components of the lattice Majorana Hamiltonian}
\label{sec:sdc}
The Hamiltonian with the Majorana mass term,
\begin{align}
\label{eq:stagHmaj}
H^{(N)}_{0} & = \vep_{N}^{-1}\sum_{x\in\Lambda_{N}}\left(a^{(1)\!\!\ \dagger}_{x+\vep_{N}}a^{(2)}_{x} - a^{(1)\!\!\ \dagger}_{x}a^{(2)}_{x} + \textup{h.c.} + \lambda_{N}\left(a^{(2)\!\!\ \dagger}_{x}a^{(1)}_{x} + a^{(1)\!\!\ \dagger}_{x}a^{(2)}_{x}\right)\right) \\ \nonumber
 & = \tfrac{1}{2L}\!\!\sum_{k\in\Gamma_{N,{\color{blue}\pm}}}\!\!\begin{pmatrix} \hat{\psi}_{+|k} \\ \hat{\psi}_{-|k} \end{pmatrix}^{\dagger}\!\! \underbrace{\begin{pmatrix} \sin(\vep_{N}k) & \!\!\!\!-i((\cos(\vep_{N}k)-1) + \lambda_{N}) \\ i((\cos(\vep_{N}k)-1) + \lambda_{N}) & -\sin(\vep_{N}k) \end{pmatrix}}_{=h^{(N)}_{\sd}(k)}\!\! \begin{pmatrix} \hat{\psi}_{+|k} \\ \hat{\psi}_{-|k} \end{pmatrix} \\ \nonumber
 & = \vep_{N}^{-1}dF_{0}(h^{(N)}_{\sd}),
\end{align}
is compatible with the charge conjugation \eqref{eq:majconj}. Moreover, the projections on positive and negative momenta (with an appropriate modification in the Ramond sector),
\begin{align}
\label{eq:pmom}
(P^{\pm}\hat{\xi})_{k} & = \left\{\begin{matrix} \theta(\pm k)\hat{\xi}_{k} & k\neq0,-\tfrac{\pi}{\vep_{N}} \\ \Re(e^{\mp\frac{i}{4}\pi}\hat{\xi}_{k})e^{\pm\frac{i}{4}\pi} & k=0,-\tfrac{\pi}{\vep_{N}} \end{matrix} \right., & \xi & \in\fh^{(\pm)}_{N,{\color{blue}\pm}},
\end{align}
are basis projections for $C$ (see \eqref{eq:sdc} and below):
\begin{align}
\label{eq:sdcbasis}
(CP^{\pm}\hat{\xi})(k) & = ((1-P^{\pm})C\hat{\xi})(k).
\end{align}
Therefore, it is possible to formulate $H^{(N)}_{0}$ in terms of four Majorana fermions, $\Psi_{\pm},:\fh^{(\pm)}_{N,{\color{blue}\pm}}\rightarrow\fB^{(\pm)}_{N,{\color{blue}\pm}}$ and $ \tilde{\Psi}_{\pm}:\fh^{(\pm)}_{N,{\color{blue}\pm}}\rightarrow\fB^{(\pm)}_{N,{\color{blue}\pm}}$:
\begin{align}
\label{eq:sdcfermi}
\Psi_{\pm}(\xi) & = \psi_{\pm}(P^{+}\xi) + \psi^{\dag}_{\pm}(CP^{-}\xi), & \tilde{\Psi}_{\pm}(\xi) & = \psi_{\pm}(P^{-}\xi) + \psi^{\dag}_{\pm}(CP^{+}\xi),
\end{align}
for $\xi\in\fh^{(\pm)}_{N,{\color{blue}\pm}}$, which correspond to the restrictions of the complex chiral fermions $\psi_{\pm}$ to positive and negative momenta. Then $H^{(N)}_{0}$ is given by self-dual second quantization:
\begin{align}
\label{eq:stagHmajSQ}
H^{(N)}_{0} & = \vep_{N}^{-1}\big(dQ_{0}(h^{(N)}_{\sd}) + d\tilde{Q}_{0}(h^{(N)}_{\sd})\big), & [dQ_{0}(h^{(N)}_{\sd}),d\tilde{Q}_{0}(h^{(N)}_{\sd})] & = 0,
\end{align}
In view of the diagonalization \eqref{eq:diagH} of $H^{(N)}_{0}$, it is worth noting that the Majorana fermions \eqref{eq:sdcfermi} can be interpreted as the self-dual components of $c^{(1)}, c^{(2)}$ in the (formal) scaling limit $N\rightarrow\infty$ of the massless case $\lambda_{N}=0$:
\begin{align}
\label{eq:sdcfermidiag}
\hat{\Psi}_{\sigma|k} &\!=\!\left\{\begin{matrix} \theta(+k)\hat{c}^{(1)}_{k}+\theta(-k)\hat{c}^{(1)\!\!\ \dag}_{-k} & \sigma\!=\!+ \\[0.1cm] \theta(+k)\hat{c}^{(2)}_{k}+\theta(-k)\hat{c}^{(2)\!\!\ \dag}_{-k} & \sigma\!=\!- \end{matrix}\right., & \hat{\tilde{\Psi}}_{\sigma|k} &\!=\!\left\{\begin{matrix} \theta(-k)\hat{c}^{(2)}_{k}+\theta(+k)\hat{c}^{(2)\!\!\ \dag}_{-k} & \sigma\!=\!+ \\[0.1cm] \theta(-k)\hat{c}^{(1)}_{k}+\theta(+k)\hat{c}^{(1)\!\!\ \dag}_{-k} & \sigma\!=\!- \end{matrix}\right.,
\end{align}
for $k\in\Gamma_{\infty,{\color{blue}\pm}}$. Furthermore, it is evident from \eqref{eq:stagHmajSQ} that $H^{(N)}_{0}$ should decouple into four massless Majorana fermions in this limit.

\subsection{Equivalence with the transverse XY model}
\label{sec:xymodel}
In view of our companion article \cite{OsborneCFTsim} that discusses our results in relation to the quantum simulation of CFTs, we provide some details on how to map the fermion algebra $\fA_{N,{\color{blue}\pm}}$ together with the Dirac Hamiltonian $H^{(N)}_{0}$ to a \textit{Pauli algebra}, $\fP_{N+1,{\color{orange}\pm}}=\otimes_{x\in\Lambda_{N+1}}M_{2}(\C)$, in which each local matrix factor resembles a (logical) qubit, using a Jordan-Wigner isomorphism. Under this mapping the Dirac Hamiltonian becomes that of an XY model with a transverse field.

To begin with, we identify the algebra $\fA_{N,{\color{blue}\pm}}$ of two-component fermion with a one-component fermion algebra, $\fA_{\SDC}(\ltwo(\Lambda_{N+1})_{{\color{blue}\pm}})$, on the doubled lattice $\Lambda_{N+1}$:
\begin{align}
\label{eq:stagfermion}
b_{x} & = a^{(1)}_{x}, & b_{x+\vep_{N+1}} & = a^{(2)\dag}_{x}, & x & \in\Lambda_{N},
\end{align}
which exploits the symmetry of $H^{(N)}_{0}$ given by: $(a^{(1)}_{x},a^{(2)}_{x})\mapsto(a^{(2)\dag}_{x}, a^{(1)\dag}_{x+\vep_{N}})$ for all $x\in\Lambda_{N}$. In terms of $b, b^{\dag}$, the Hamiltonian \eqref{eq:stagH} is mapped to:
\begin{align}
\label{eq:stagH1c}
H^{(N)}_{0} & = \vep_{N}^{-1}\hspace{-0.25cm}\sum_{x\in\Lambda_{N+1}}\hspace{-0.25cm}\big(b_{x}b_{x+\vep_{N+1}}-b^{\dag}_{x}b^{\dag}_{x+\vep_{N+1}}+\lambda_{N}(b^{\dag}_{x}b_{x}-\tfrac{1}{2}\mathds{1}_{2})\big),
\end{align}
with periodic or anti-periodic boundary conditions in accordance with $(a^{(1)},a^{(2)})$. The translation invariance of $H^{(N)}_{0}$ on $\Lambda_{N+1}$ subsumes the translation invariance on $\Lambda_{N}$ and special field-exchange symmetry from above. In the momentum space representation, the relation \eqref{eq:stagfermion} between the single-component fermion $a$ and two-component fermion $(a^{(1)},a^{(2)})$ is given by:
\begin{align}
\label{eq:stagfermionfourier}
\hat{a}^{(1)}_{k} & = \tfrac{1}{2}(\hat{b}_{k}+\hat{b}_{k+\frac{\pi}{\varepsilon_{N+1}}}), & \hat{a}^{(2)}_{k} & = \tfrac{1}{2}e^{i\varepsilon_{N+1}k}(\hat{b}^{\dag}_{-k}-\hat{b}^{\dag}_{-k+\frac{\pi}{\varepsilon_{N+1}}}),
\end{align}
where the two-component fermion can be considered to be periodically extended from $\Gamma_{N,{\color{blue}\pm}}$ to $\Gamma_{N+1,{\color{blue}\pm}}$.

\bigskip

A follow-up \textit{Jordan-Wigner isomorphism},
\begin{align}
\label{eq:JWtrafo}
b_{x} & = \Big(\prod_{\substack{y\in\Lambda_{N+1} \\ -L\leq y<x}}\sigma^{(1)}_{y}\Big)\tfrac{1}{2}(\sigma^{(3)}_{x}+i\sigma^{(2)}_{x}), & x & \in\Lambda_{N+1},
\end{align}
yields a soluble transverse XY model (cp.~\cite{LiebTwoSolubleModels}):
\begin{align}
\label{eq:XYmodel}
H^{(N)}_{0} & = \tfrac{1}{2}\vep_{N}^{-1} \hspace{-0.25cm}\sum_{x\in\Lambda_{N+1}} \hspace{-0.25cm}\big(\sigma^{(3)}_{x}\sigma^{(3)}_{x+\vep_{N+1}}-\sigma^{(2)}_{x}\sigma^{(2)}_{x+\vep_{N+1}}+\lambda_{N}\sigma^{(1)}_{x}\big) \\ \nonumber
& \hspace{0.25cm} -\tfrac{1}{2}\vep_{N}^{-1}(\texto{\pm}1\textr{\pm}(-1)^{F})(\sigma^{(3)}_{-L}\sigma^{(3)}_{L-\vep_{N+1}}-\sigma^{(2)}_{-L}\sigma^{(2)}_{L-\vep_{N+1}}),
\end{align}
where ($\textr{\pm}$) indicates the fermion boundary conditions while ($\texto{\pm}$) labels the complementing spin boundary conditions, such that boundary term in the second line vanishes on the lattice vacuum $\Omega^{(N)}_{0,{\color{blue}\pm}}$.

We can rephrase the diagonalization of $H^{(N)}_{0}$ as given in \eqref{eq:diagH} in terms of the Fourier transform of the one-component fermion $b, b^{\dag}$:
\begin{align}
\label{eq:FstagH1c}
H^{(N)}_{0} & = \tfrac{1}{4L}\hspace{-0.25cm}\sum_{k\in\Gamma_{N+1,-,>0}}\begin{pmatrix} \hat{b}_{k} \\ \hat{b}^{\dag}_{-k} \end{pmatrix}^{\dagger} \begin{pmatrix} \lambda_{N} & -i2\sin(\vep_{N+1}k) \\ i2\sin(\vep_{N+1}k) & -\lambda_{N} \end{pmatrix} \begin{pmatrix} \hat{b}_{k} \\ \hat{b}^{\dag}_{-k} \end{pmatrix}.
\end{align}
The diagonalizing Bogoliubov transformation is completely analogous to \eqref{eq:stagHbog} for $k>0$:
\begin{align}
\label{eq:stagHbog1c}
\gamma_{\lambda_{N}}(k) & = \Big(\tfrac{1}{2\omega_{\lambda_{N}}(k)}\Big)^{-\frac{1}{2}}\begin{pmatrix} (\omega_{\lambda_{N}}(k)+\lambda_{N})^{\frac{1}{2}} & i\sign(k)(\omega_{\lambda_{N}}(k)-\lambda_{N})^{\frac{1}{2}} \\ i\frac{2\sin(\vep_{N+1}k)}{(\omega_{\lambda_{N}}(k)+\lambda_{N})^{\frac{1}{2}}} & \sign(k)\frac{2\sin(\vep_{N+1}k)}{(\omega_{\lambda_{N}}(k)-\lambda_{N})^{\frac{1}{2}}} \end{pmatrix},
\end{align}
which defines the diagonalizing fermion, $(\hat{c}_{k}, \hat{c}^{\dag}_{-k})^{t}=\gamma_{\lambda_{N}}(k)^{\dagger}(\hat{b}_{k}, \hat{b}^{\dag}_{-k})^{t}$,
leading to:
\begin{align}
\label{eq:diagH1c}
H^{(N)}_{0} & = \tfrac{1}{4L}\hspace{-0.25cm}\sum_{k\in\Gamma_{N+1,{\color{blue}\pm}}}\hspace{-0.25cm}\omega_{\lambda_{N}}(k)\hat{c}^{\dag}_{k}\hat{c}_{k} - \tfrac{1}{4}\vep_{N+1}^{-1}\hspace{-0.25cm}\sum_{k\in\Gamma_{N+1,{\color{blue}\pm}}}\hspace{-0.25cm}\omega_{\lambda_{N}}(k).
\end{align}
If we normalize the fermion, $\tilde{c}_{k} = (2L_{N+1})^{-\frac{1}{2}}\hat{c}_{k}$, and subject it to another inverse Jordan-Wigner transform similar to \eqref{eq:JWtrafo}, we can write
\begin{align}
\label{eq:diagH1cnorm}
H^{(N)}_{0} & = \tfrac{1}{2}\vep_{N}^{-1}\sum_{k\in\Gamma_{N+1,{\color{blue}\pm}}}\omega_{\lambda_{N}}(k)\tilde{\sigma}^{(1)}_{k}.
\end{align}
Therefore, it is possible to express the lattice vacuum as a qubit product state: $\Omega^{(N)}_{0,{\color{blue}\pm}}=\otimes_{k\in\Gamma_{N+1,{\color{blue}\pm}}}|\!\leftarrow\rangle$, which allows for the initialization of a quantum simulation of \eqref{eq:XYmodel} in its many-body ground state following \cite{VerstraeteQuantumCircuitsFor}.

\section{Scaling limits of lattice fermions}
\label{sec:ffscaling}
We now explain how to implement the scaling limit, $N\rightarrow\infty$, rigorously by operator-algebraic renormalization \cite{StottmeisterOperatorAlgebraicRenormalization, BrothierAnOperatorAlgebraic} which in turn allows us to give precise statements about the convergence of lattice quantities, e.g.~correlation functions with respect to the lattice vacuum $\omega^{(N)}_{0,{\color{blue}\pm}}$ or finite-scale Bogoliubov transformations and their implementers.
 
Following the recipe in \cite{MorinelliScalingLimitsOf}, we define the scaling limit of lattice fermions, complex and Majorana, by a renormalization group procedure in terms of an inductive system of unital, injective $*$-morphisms,
\begin{align}
\label{eq:rgF}
\alpha^{N}_{N+1} : \fA_{N,{\color{blue}\pm}} & \longrightarrow \fA_{N+1,\pm}, & \alpha^{N+1}_{N+2}\circ\alpha^{N}_{N+1} & = \alpha^{N}_{N+2} , & N & \in\N_{0},
\end{align}
and similarly for $\fB_{N,{\color{blue}\pm}}$. In the following we state the formulas for $\fA_{N,{\color{blue}\pm}}$ only, but they equally apply to $\fB_{N,{\color{blue}\pm}}$.

In present setting, we define the \textit{renormalization group} $\{\alpha^{N}_{N+1}\}_{N\in\N_{0}}$ as Bogoliubov transformations associated with an inductive system of isometries between one-particle spaces,
\begin{align}
\label{eq:rgH}
R^{N}_{N+1} : \fh_{N,{\color{blue}\pm}} & \longrightarrow \fh_{N+1,\pm}, & R^{N+1}_{N+2}\circ R^{N}_{N+1} & = R^{N}_{N+2} & N & \in\N_{0},
\end{align}
with the additional requirement, $R^{N}_{N+1}C = CR^{N}_{N+1}$, in the self-dual case.

Because of the general properties of CAR algebras \cite{BratteliOperatorAlgebrasAnd2, EvansQuantumSymmetriesOn}, we obtain the inductive-limit objects:
\begin{align}
\label{eq:indlimF}
\fA_{\infty,{\color{blue}\pm}} & = \fA_{\CAR}(\fh_{\infty,{\color{blue}\pm}}) = \varinjlim_{N}\fA_{N,{\color{blue}\pm}}, & \fh_{\infty,{\color{blue}\pm}} & = \varinjlim_{N}\fh_{N,{\color{blue}\pm}}.
\end{align}
At the level of states, $\omega^{(N)}\in\fA^{*}_{N,\pm}$, the \textit{renormalization group flow} is defined by:
\begin{align}
\label{eq:rgflow}
\omega^{(N)}_{M} & = \omega^{(N+M)}\circ\alpha^{N}_{N+M}, & N,M\in\NN_{0}.
\end{align}
A scaling limit of a family of lattice states $\{\omega^{(N)}\}_{N\in\N_{0}}$ is defined as:
\begin{align}
\label{eq:sl}
\omega^{(N)}_{\infty} & = \lim_{M\rightarrow\infty}\omega^{(N)}_{M}, & N & \in\N_{0},
\end{align}
which is automatically projectively consistent, $\omega^{N+1}_{\infty}\circ\alpha^{N}_{N+1} = \omega^{N}_{\infty}$, by continuity \cite{MorinelliScalingLimitsOf, BrothierAnOperatorAlgebraic} if it exists as a weak$^{*}$-limit and, thus, defines a state $\omega^{(\infty)}_{\infty}$ on $\fA_{\infty,{\color{blue}\pm}}$. In this sense, the inductive-limit objects \eqref{eq:indlimF} serve as carrier spaces of the scaling limit of a family $\{\omega^{(N)}\}_{N\in\N_{0}}$.

The observation that the inductive-limit objects \eqref{eq:indlimF} can be characterized as the sets of $\alpha$- respectively $R$-convergent sequences \cite{DuffieldMeanFieldDynamical},
\begin{align}
\label{eq:convseq}
\lim_{M\rightarrow\infty}\limsup_{N\rightarrow\infty}\|O_{N}-\alpha^{M}_{N}(O_{M})\| & = 0, & \lim_{M\rightarrow\infty}\limsup_{N\rightarrow\infty}\|\xi_{N}-R^{M}_{N}(\xi_{M})\|_{N} & = 0,
\end{align}
for $O_{N}\in\fA_{N,{\color{blue}\pm}}$ respectively $\xi_{N}\in\fh_{N,{\color{blue}\pm}}$, motivates the following definition of the scaling limit of sequences of composite operators with limits possibly not in $\fA_{\infty,{\color{blue}\pm}}$\footnote{In \cite{DuffieldMeanFieldDynamical} the concept of convergent sequences is discussed for generalized inductive systems of Banach spaces in the context of mean-field limits.}.
\begin{defn}[Convergent sequences of operators]
\label{def:convseq}
Let $\{p_{N}\}_{N\in\N_{0}}$ be a sequence of semi-norms, $p_{N}:\fA_{N,{\color{blue}\pm}}\rightarrow\R_{\geq0}$. A sequence of operators $\{O_{N}\}_{N\in\N_{0}}$ with $O_{N}\in\fA_{N,{\color{blue}\pm}}$ is called $\alpha$-convergent with respect to the family $\{p_{N}\}_{N\in\N_{0}}$ if:
\begin{align*}
\lim_{M\rightarrow\infty}\limsup_{N\rightarrow\infty}p_{N}(O_{N}-\alpha^{M}_{N}(O_{M})) & = 0.
\end{align*}
Similarly, let $\{q_{N}\}_{N\in\N_{0}}$ be a sequence of semi-norms, $q_{N}:B(\fh_{N,{\color{blue}\pm}})\rightarrow\R_{\geq0}$, we call a sequence of operators $\{o_{N}\}_{N\in\N_{0}}$, with $o_{N}\in B(\fh_{N,{\color{blue}\pm}})$, $R$-convergent with respect to the family $\{q_{N}\}_{N\in\N_{0}}$ if:
\begin{align*}
\lim_{M\rightarrow\infty}\limsup_{N\rightarrow\infty}q_{N}(o_{N}-R^{M}_{N}o_{M}R^{M}_{N}{}^{*}) & = 0.
\end{align*}
\end{defn}
Of course, $R$-convergence for one-particle operators is related to $\alpha$-convergence by the observation:
\begin{align}
\label{eq:Ralpha}
\alpha^{N}_{M}(dF_{0}(o_{N})) & = dF_{0}(R^{N}_{M}o_{N}R^{N}_{M}{}^{*}).
\end{align}
In Section \ref{sec:approxconf}, we invoke families of semi-norms that are induced by semi-norms $p_{\infty}$ and $q_{\infty}$ on $B(\fF_{\ua}(\fh_{\infty,{\color{blue}\pm}}))$ respectively $B(\fh_{\infty,{\color{blue}\pm}})$:
\begin{align}
\label{eq:snind}
p_{N} & = p_{\infty}\circ\alpha^{N}_{\infty}, & q_{N} & = q_{\infty}\circ(R^{N}_{\infty}(\!\ \cdot\!\ )R^{N}_{\infty}{}^{*}).
\end{align}
Specific examples are strong operator semi-norms (for $M\in\N_{0}$),
\begin{align}
\label{eq:snex}
p_{\infty}(O) & = \|Oa^{\dag}(R^{M}_{\infty}(\xi_{1}))...a^{\dag}(R^{M}_{\infty}(\xi_{n}))\Omega_{0}\|, & \xi_{j} & \in\fh_{M,{\color{blue}\pm}}, j=1,...,n, \\ \nonumber
q_{\infty}(o) & = \|oR^{M}_{\infty}(\xi)\|_{\infty}, & \xi & \in\fh_{M,{\color{blue}\pm}}, 
\end{align}
for operators $O$ on $\fF_{\ua}(\fh_{\infty,{\color{blue}\pm}})$ or operators $o$ on the one-particle space $\fh_{N,{\color{blue}\pm}}$. This way, we can study the convergence of sequences of operators at finite scales such as the lattice Dirac Hamiltonian $H^{(N)}_{0}$ and its one-particle version $h^{(N)}_{0}$, which should have limits  that are unbounded operators in the scaling limit.

\subsection{Wavelet renormalization group}
\label{sec:waveletrg}
In the following, we make use of a specific realization of the renormalization group using the theory of wavelets by adapting the constructions in \cite{MorinelliScalingLimitsOf} to the fermionic setting.

To this end, we consider a compactly supported orthonormal \textit{scaling function}, $s\in L^{2}(\R)$, $\|s\|_{L^{1}}=1$, associated with a wavelet basis of $L^{2}(\R)$ that satisfies the \textit{scaling equation} \cite{DaubechiesTenLecturesOn, MeyerWaveletsAndOperators}:
\begin{align}
\label{eq:scalingeq}
s(x) & = \sum_{n\in\Z}h_{n} 2^{\frac{1}{2}}s(2x-n), & x & \in\R,
\end{align}
where the coefficients $\{h_{n}\}_{n\in\Z}$ are called a \textit{low-pass filter}\footnote{The compact support of $s$ enforces a that the low-pass filter has only finitely many non-vanishing elements.}. Since our finite-scale fermion algebras, $\fA_{N,{\color{blue}\pm}}$ and $\fB_{N,{\color{blue}\pm}}$, describe quantum systems in the finite volume $S^{1}_{L}$, we denote by,
\begin{align}
\label{eq:rescaledsf}
s^{(\vep_{N})}(x) & = \sum_{m\in\Z}\vep_{N}^{-\frac{1}{2}}s(\vep_{N}^{-1}(x+2Lm)),
\end{align}
the rescaled and $2L$-periodized scaling function, which yield a wavelet basis of $L^{2}(S^{1}_{L})$ and satisfy in analogy with \eqref{eq:scalingeq}:
\begin{align}
\label{eq:scalingeqre}
s^{(\vep_{N})}(x) & = \sum_{n\in\vep_{N+1}^{-1}\Lambda_{N+1}}h_{n}s^{(\vep_{N+1})}(x-\vep_{N+1}n), & x & \in S^{1}_{L}.
\end{align}
The wavelet renormalization group is now defined as follows:

\begin{defn}[Wavelet renormalization group]
\label{def:waveletrg}
Given a compactly supported orthonormal scaling function $s\in L^{2}(\R)$ with low-pass filter $\{h_{n}\}_{n\in\Z}$, let,
\begin{align}
\label{eq:waveletrg1p}
R^{N}_{N+1}(\xi^{(j)}) & = \sum_{x\in\Lambda_{N}}\xi^{(j)}_{x}\sum_{n\in\vep_{N+1}^{-1}\Lambda_{N+1}}h_{n}\delta^{(N+1)}_{x+\vep_{N+1}n},
\end{align}
and
\begin{align}
\label{eq:waveletrg1pinf}
R^{N}_{\infty}(\xi^{(j)}) & = \sum_{x\in\Lambda_{N}}\xi^{(j)}_{x}\underbrace{s^{(\vep_{N})}(\cdot\!\ -x)}_{:=s^{(\vep_{N})}_{x}} = s^{(\vep_{N})}\ast_{\Lambda_{N}}\xi^{(j)},
\end{align}
for $\xi=(\xi^{(1)},\xi^{(2)})\in\fh_{N,{\color{blue}\pm}}$, where we use the notation $\ast_{\Lambda_{N}}$ to denote the convolution with respect to the lattice $\Lambda_{N}$. The \emph{wavelet renormalization group} is densely defnined by:
\begin{align}
\label{eq:waveletrg}
\alpha^{N}_{N+1}(a(\xi)) & = a(R^{N}_{N+1}(\xi)), & \alpha^{N}_{\infty}(a(\xi)) & = a(R^{N}_{\infty}(\xi)), & \xi & \in\fh_{N,{\color{blue}\pm}}.
\end{align}
\end{defn}

We summarize some important properties of the maps $R^{N}_{N+1}$ that follow immediately from the properties of a (compactly supported) scaling function and its associated wavelet basis, cf.~\cite{MeyerWaveletsAndOperators, DaubechiesTenLecturesOn, MorinelliScalingLimitsOf}:
\begin{prop}
\label{prop:waveletrg}
The one-particle map of the wavelet renormalization group has the following properties: $R^{N}_{N+1}:\fh_{N,{\color{blue}\pm}}\rightarrow\fh_{N+1,\pm}$ is
\begin{itemize}
	\item[1.] well-defined,
	\item[2.] $C$-compatible,
	\item[3.] isometric, i.e.~ $R^{N}_{N+1}{}^{\!\!*}\circ R^{N}_{N+1} = \1_{\fh_{N}}$, 
	\item[4.] asymptotically compatible, i.e.~$R^{N+1}_{\infty}\circ R^{N}_{N+1} = R^{N}_{\infty}$.
\end{itemize}
Moreover, we have $\fh_{\infty,{\color{blue}\pm}}\cong L^{2}(S^{1}_{L})_{{\color{blue}\pm}}\otimes\C^{2}$, and $R^{N}_{\infty}{}^{*}\circ R^{N}_{\infty} = \1_{\fh_{N}}$.
\end{prop}
For our purposes it is convenient to explicitly derive the form of the wavelet renormalization groups in the momentum-space representation via the Fourier transform, cf. \eqref{eq:lattF} \& \eqref{eq:contF}:
\begin{align}
\label{eq:waveletrg1pF}
R^{N}_{N+1}(\hat{\xi})_{k} & = 2^{\frac{1}{2}}m_{0}(\vep_{N+1}k)(\hat{\xi}_{\per})_{k}, & R^{N}_{\infty}(\hat{\xi})_{k} & = \vep_{N}^{\frac{1}{2}}\hat{s}(\vep_{N}k)(\hat{\xi}_{\per})_{k}.
\end{align}
where $m_{0}(\vep_{N+1}k) = 2^{-\frac{1}{2}}\sum_{n\in\Z}h_{n}e^{-ik\vep_{N+1}n}$ (considered as periodic on $\Gamma_{N+1,{\color{blue}\pm}}$), and $\hat{\xi}_{\per}$ denotes the periodic extension of $\hat{\xi}$ from $\Gamma_{N,{\color{blue}\pm}}$ to $\Gamma_{N+1,{\color{blue}\pm}}$ respectively $\Gamma_{\infty,{\color{blue}\pm}}$. The asymptotic compatibility and various convergence results, as we show in Section \ref{sec:approxconf}, turn out to be a consequence of (applying the Fourier transform to \eqref{eq:scalingeq}):
\begin{align}
\label{eq:infprod}
\hat{s}(\vep_{N}\!\ .\!\ ) & = m_{0}(\vep_{N+1}\!\ .\!\ )\hat{s}(\vep_{N+1}\!\ .\!\ ) = \lim_{M\rightarrow\infty}\prod_{n=1}^{M}m_{0}(\vep_{N+n}\!\ .\!\ ),
\end{align}
where the infinite product converges absolutely and uniformly on compact sets \cite{DaubechiesTenLecturesOn}. But, the convergence also holds in arbitrary Sobolev-type subspaces of $L^{2}(S^{1}_{L})$ provided the scaling function is regular enough, $s\in C^{\alpha}(\R)$, $\alpha>0$, e.g.~those of the Daubechies family, see Lemma \ref{lem:convergence} below. 

As we are not only interested in the scaling limit of the fermions, $\fA_{N,{\color{blue}\pm}}$ and $\fB_{N,{\color{blue}\pm}}$, but also that of implementers of Bogoliubov transformations and their generators, such as $dF_{S}(G)$ and $dQ_{S}(Q)$ given in \eqref{eq:secondquantarakino} \& \eqref{eq:secondquantarakinosdc}, we introduce a modified version of the maps $R^{N}_{N+1}$ adapted to the fact that the generators are quadratic in the fermions (i.e.~\textit{currents}). To this end, consider the spaces,
\begin{align}
\label{eq:vectcurrents}
\fl(\Lambda_{N},\C^{D}) & = \{X:\Lambda_{N}\rightarrow\C^{D}\} \cong \fl(\Gamma_{N,{\color{blue}+}},\C^{D}),
\end{align}
of $\C^{D}$-valued differential loops (periodic boundary conditions) together with the obvious extension of the (lattice) Fourier transform.
\begin{defn}[Wavelet renormalization for differential loops]
\label{def:waveletrgcur}
Given a compactly supported orthonormal scaling function $s\in L^{2}(\R)$ as in Definition \ref{def:waveletrg}. The wavelet renormalization group for $\C^{D}$-valued loops is defined by:
\begin{align}
\label{eq:waveletrgcur}
S^{N}_{N+1}(\hat{X}^{(j)})_{k} & = 2 m_{0}(\vep_{N+1}k)(\hat{X}^{(j)}_{\per})_{k}, & S^{N}_{\infty}(\hat{X}^{(j)})_{k} & = \vep_{N}\hat{s}(\vep_{N}k)(\hat{X}^{(j)}_{\per})_{k},
\end{align}
for $X\in\fl(\Lambda_{N},\C^{D})$.
\end{defn}
In view of Proposition \ref{prop:waveletrg}, we note that the wavelet renormalization group for currents also satisfies asymptotic compatibility:
\begin{align}
\label{eq:waveletrgcurcomp}
S^{N+1}_{\infty}\circ S^{N}_{N+1} & = S^{N}_{\infty},
\end{align}
but we refrain from formulating an associated (topological) inductive limit and simply note that $\im S^{N}_{\infty}\subset \fl^{\alpha}(S^{1}_{L},\C^{D})$ is a finite-dimensional subspace of differential $\C^{D}$-valued $C^{\alpha}$-loops associated with the span of $s^{(\vep_{N})}$ and its $\Lambda_{N}$-translates. 

Clearly, the difference between $R^{N}_{\infty}$ and $S^{N}_{\infty}$ is the geometrical scaling factor $\vep_{N}^{\frac{1}{2}}$ respectively $\vep_{N}$ which reflects the fact that fermions are half densities while generators are densities.

\subsubsection{Quasi-local structure}
\label{sec:antiloc}
By analogy with case of lattice scalar fields, the wavelet renormalization group is compatible with the real-space anti-local structure of the fermions on the lattice and in the continuum. The simple reason for this compatibility is the finite length of the low-pass filter $\{h_{n}\}_{n\in\Z}$, which entails that according to \eqref{eq:waveletrg1p} and \eqref{eq:waveletrg} the fermion and Majorana algebras of a single lattice site at a given scale, $\fA_{N,{\color{blue}\pm}}(x)$ and $\fB_{N,{\color{blue}\pm}}(x)$, $x\in\Lambda_{N}$, generated by $a_{x}, a^{\dag}_{x}$ respectively $\Psi_{x}$, are mapped into localized algebras at the successive scale, $\fA_{N+1,\pm}(I_{x})$ and $\fB_{N+1,\pm}(I_{x})$, generated by $a_{y}, a^{\dag}_{y}$ respectively $\Psi_{y}$ for $y\in I_{x}\cap\Lambda_{N+1}$, where $I_{x}\subset S^{1}_{L}$ is an interval determined by the length of the low-pass filter. 

This observation allows for the following definition.
\begin{defn}
\label{def:antiloc}
Let $I\subset S^{1}_{L}$ be an open interval, and denote by $\Lambda_{N}\cap I$ denote the subset of those $x$ of intersection as subsets of $S^{1}_{L}$ such that $x+\supp(s^{(\vep_{N})})$ does not intersect the boundary $\partial I$. Then, we define the \emph{local one-particle Hilbert spaces} as:
\begin{align*}
\fh_{N,{\color{blue}\pm}}(I) & = \ltwo(\Lambda_{N}\cap I)_{{\color{blue}\pm}}\otimes\C^{2}\subset\fh_{N,{\color{blue}\pm}},
\end{align*}
where the inclusion results from extension by zero. The \emph{twisted-local fermion and Majorana algebras} are, thus, defined as:
\begin{align*}
\fA_{N,{\color{blue}\pm}}(I) & = \fA_{\CAR}(\fh_{N,{\color{blue}\pm}}(I)), & \fB_{N,{\color{blue}\pm}}(I) & = \fB_{\CAR}(\fh_{N,{\color{blue}\pm}}(I),C),
\end{align*}
which are considered to be subalgebras of $\fA_{N,{\color{blue}\pm}}$ and $\fB_{N,{\color{blue}\pm}}$.
\end{defn}
As an immediate consequence of this definition we have the following properties:
\begin{prop}
\label{prop:antilocal}
The local one-particle Hilbert spaces and twisted-local fermion and Majorana algebras form inductive systems with respect to the wavelet renormalization group. Specifically, we have:
\begin{align*}
\fh_{\infty,{\color{blue}\pm}}(I) & = \varinjlim_{N}\fh_{N,{\color{blue}\pm}}(I), & \fA_{\infty,{\color{blue}\pm}}(I) & = \varinjlim_{N}\fA_{N,{\color{blue}\pm}}(I), & \fB_{\infty,{\color{blue}\pm}}(I) & = \varinjlim_{N}\fB_{N,{\color{blue}\pm}}(I).
\end{align*}
Moreover, we recover the (twisted) quasi-local structure of $\fA_{\infty,{\color{blue}\pm}}$ and $\fB_{\infty,{\color{blue}\pm}}$:
\begin{align*}
& \fA_{\infty,{\color{blue}\pm}}(I)\subset\fA_{\infty,{\color{blue}\pm}}(I'), &  & \fB_{\infty,{\color{blue}\pm}}(I)\subset\fB_{\infty,{\color{blue}\pm}}(I'), & I\subset I', \\
& [\fA_{\infty,{\color{blue}\pm}}(I),\fA_{\infty,{\color{blue}\pm}}(I')] = \{0\}, &  & [\fB_{\infty,{\color{blue}\pm}}(I),\fB_{\infty,{\color{blue}\pm}}(I')] = \{0\}, & I\cap I' & = \emptyset, \\
& \fA_{\infty,{\color{blue}\pm}} = \overline{\bigcup_{I\subset S^{1}_{L}}\fA_{\infty,{\color{blue}\pm}}(I)}, &  & \fB_{\infty,{\color{blue}\pm}} = \overline{\bigcup_{I\subset S^{1}_{L}}\fB_{\infty,{\color{blue}\pm}}(I)},
\end{align*}
where $[\!\ ,\!\ ]$ in the second line is the $\Z_{2}$-graded or twisted commutator with respect to the grading given by the parity operator $(-1)^{F}$ \cite{BratteliOperatorAlgebrasAnd1}.
\end{prop}

\subsubsection{A decay estimate}
\label{sec:convergence}
We provide decay estimate on the finite products , $\prod_{n=1}^{M}m_{0}(\vep_{N+n}k)$, that directly implies \eqref{eq:infprod} in Sobolev-type norms for compactly supported Daubechies scaling functions by dominated convergence. The proof adapts a multi-scale decomposition strategy in \cite{DaubechiesTenLecturesOn} used to obtain regularity estimates for such scaling functions to finite products. A similar albeit less general statement can be found in a recent article by one of the authors \cite{MorinelliScalingLimitsOf}.

\begin{lemma}
\label{lem:convergence}
Let $\phi =\!\ _{K}\phi$ be a compactly supported Daubechies scaling function with $K\geq2$, i.e.:
\begin{align}
\label{eq:m0factor}
_{K}m_{0}(l) & = \left(\tfrac{1+e^{-il}}{2}\right)^{K} {}_{K}\cL(l).
\end{align}
where ${}_{K}\cL$ is a certain trigonometric polynomial. Then, for $j\in\N$, and $j^{-1}M\in\N$:
\begin{align}
\label{eq:convergence}
\left|\chi_{[-2^{M}\pi,2^{M}\pi)}(l)\!\prod_{n=1}^{M}\!{}_{K}m_{0}(2^{-n}l)\right| & \!=\!\chi_{[-2^{M}\pi,2^{M}\pi)}(l)\max\{e^{C}\!,\!\pi^{-\cK_{j}}\}\!\!\left(\tfrac{|\sin(\frac{1}{2}l)|}{|l|}\right)^{\!K}\!(1+|l|)^{\cK_{j}},
\end{align}
where $\cK_{j} = j^{-1}\log_{2}(q_{j})$, $q_{j}=\sup_{l\in\R}\prod_{n=0}^{j-1}{}_{K}\cL(2^{-n}l)$, and $|{}_{K}\cL(l)|\leq 1+C|l|$.
\end{lemma}
\begin{proof}
Clearly, $_{K}\cL(0)=1$ and, thus, $|{}_{K}\cL(l)|\leq 1+C|l|$. Put, ${}_{K}\cL_{j}(l) \!=\! \prod_{n=0}^{j-1}{}_{K}\cL(2^{-n}l)$ such that $\prod_{m=0}^{j^{-1}M-1}{}_{K}\cL_{j}(2^{-(mj+1)}l) = \prod_{n=1}^{M}{}_{K}\cL(2^{-n}l)$ and $q_{j} = \sup_{l\in\R}|{}_{K}\cL_{j}(l)|$. Now:
\begin{itemize}
	\item[1.] Assume $|l|\leq1$:
	\begin{align*}
	\prod_{m=0}^{j^{-1}M-1}|{}_{K}\cL_{j}(2^{-mj+1}l)| & \leq e^{C(1-2^{-m})|l|}\leq e^{C}.
	\end{align*}
	\item[2.] Assume $2^{j(M'-1)}<|l|\leq2^{jM'}$ for some $M'\leq j^{-1}M$:
	\begin{align*}
	\prod_{m=0}^{j^{-1}M-1}|{}_{K}\cL_{j}(2^{-(mj+1)}l)| & = \prod_{m=0}^{M'-1}|{}_{K}\cL_{j}(2^{-(mj+1)}l)|\prod_{n=M'}^{j^{-1}M-1}|{}_{K}\cL_{j}(2^{-(nj+1)}l)| \\
	& = \prod_{m=0}^{M'-1}|{}_{K}\cL_{j}(2^{-(mj+1)}l)|\prod_{n=0}^{j^{-1}M-M'-1}|{}_{K}\cL_{j}(2^{-(nj+1)}\underbrace{2^{-jM'}l}_{|\!\ \cdot\!\ |\leq1})| \\
	& \leq e^{C} q_{j}^{M'} = e^{C}2^{jM'\cK_{j}} \leq e^{C}(1+|l|)^{\cK_{j}}.
	\end{align*}
	\item[3.] Assume $2^{M}<|l|\leq2^{M}\pi$:
	\begin{align*}
	\prod_{m=0}^{j^{-1}M-1}|{}_{K}\cL_{j}(2^{-(mj+1)}l)| & \leq q_{j}^{j^{-1}M} = 2^{M\cK_{j}}\leq \pi^{-\cK_{j}}(1+|l|)^{\cK_{j}}.
	\end{align*}
\end{itemize}
Thus, we have $\prod_{n=1}^{M}|{}_{K}\cL(2^{-n}l)| = \prod_{m=0}^{j^{-1}M-1}|{}_{K}\cL_{j}(2^{-(mj+1)}l)| \leq \max\{e^{C},\pi^{-\cK_{j}}\}(1+|l|)^{\cK_{j}}$ for $|l|\leq 2^{M}\pi$. In combination with the formula $\prod_{n=1}^{M}\Big|\tfrac{1-e^{-i2^{-n}l}}{2}\Big| = \prod_{n=1}^{M}|\cos(\tfrac{1}{2}2^{-n}l)| = \Big|\tfrac{2\sin(\frac{1}{2}l)}{2^{M+1}\sin(2^{-(M+1)l})}\Big|$ the result follows.
\end{proof}
It is known that $\cK_{2} \leq (K-1)(2-\tfrac{3}{4}\log_{2}(3)$ \cite{DaubechiesTenLecturesOn} implying:
\begin{align}
\label{eq:decay}
-K+\cK_{2} & \leq -(1+(K-1)\underbrace{(\tfrac{3}{4}\log_{2}(3)-1)}_{\sim 0.1887}) \stackrel{K\rightarrow\infty}{\longrightarrow} -\infty.
\end{align}
Thus, the decay in \eqref{eq:convergence} can be made arbitrarily fast in a polynomial sense by increasing $K$, and we note that $K_{j_{2}}\leq\cK_{j_{1}}+C\tfrac{j_{1}}{j_{2}}$ for $j_{1}<j_{2}$

For convenience, we state the following regularity estimate for the scaling function from \cite[Lemma 7.1.2]{DaubechiesTenLecturesOn}:
\begin{lemma}
\label{lem:decay}
Let $\phi =\!\ _{K}\phi$ be a compactly supported Daubechies scaling function with $K\geq2$, then it satisfies the following decay estimate:
\begin{align}
\label{eq:decaysf}
|\hat{s}(l)| & \leq C'(1+|l|)^{-K+\cK}
\end{align}
for $\cK = \inf_{j\in\N}\cK_{j}$. which implies $s\in C^{\alpha}(\R)$ if $\cK < K-(1+\alpha)$.
\end{lemma}

\subsection{Momentum-cutoff renormalization group}
\label{sec:momcut}
A second renormalization group that we make heavy use of in Section \ref{sec:approxconf} implements a sharp cutoff in momentum space \cite{MorinelliScalingLimitsOf}: 
\begin{defn}[Momentum-cutoff renormalization group]
\label{def:momrg}
Let $\chi_{\Gamma_{N,{\color{blue}\pm}}}$ be the characteristic function of $\Gamma_{N,{\color{blue}\pm}}\subset\Gamma_{N+1,{\color{blue}\pm}}\subset\Gamma_{\infty,{\color{blue}\pm}}$, and let:
\begin{align}
\label{eq:momrg1p}
R^{N}_{N+1}(\hat{\xi})_{k} & = 2^{\frac{1}{2}}\chi_{\Gamma_{N}}(k)(\hat{\xi}_{\per})_{k}, & R^{N}_{\infty}(\hat{\xi})_{k} & = \vep_{N}^{\frac{1}{2}}\chi_{\Gamma_{N}}(k)(\hat{\xi}_{\per})_{k},
\end{align}
for $\xi=(\xi^{(1)},\xi^{(2)})\in\fh_{N,{\color{blue}\pm}}$. The \emph{momentum-cutoff renormalization group} is densely defined by:
\begin{align}
\label{eq:momrg}
\alpha^{N}_{N+1}(\hat{a}(\hat{\xi})) & = \hat{a}(R^{N}_{N+1}(\hat{\xi})), & \alpha^{N}_{\infty}(\hat{a}(\hat{\xi})) & = \hat{a}(R^{N}_{\infty}(\hat{\xi})).
\end{align}
\end{defn}
The momentum-cutoff renormalization group enjoys the same properties as the wavelet renormalization group, cp.~Proposition \ref{prop:waveletrg}.
\begin{prop}
\label{prop:momrg}
The one-particle map of the momentum-cutoff renormalization group has the following properties: $R^{N}_{N+1}:\fh_{N,{\color{blue}\pm}}\rightarrow\fh_{N+1,{\color{blue}\pm}}$ is
\begin{itemize}
	\item[1.] well-defined,
	\item[2.] $C$-compatible,
	\item[3.] isometric, i.e.~ $R^{N}_{N+1}{}^{\!\!*}\circ R^{N}_{N+1} = \1_{\fh_{N}}$, 
	\item[4.] asymptotically compatible, i.e.~$R^{N+1}_{\infty}\circ R^{N}_{N+1} = R^{N}_{\infty}$.
\end{itemize}
Moreover, the inductive-limit Hilbert space agrees with that of the wavelet renormalization group, $\fh_{\infty,{\color{blue}\pm}}\cong L^{2}(S^{1}_{L})_{{\color{blue}\pm}}\otimes\C^{2}$, and $R^{N}_{\infty}{}^{*}\circ R^{N}_{\infty} = \1_{\fh_{N}}$.
\end{prop}
Let us also state the real-space form of the asymptotic one-particle maps of the momentum-cutoff renormalization group:
\begin{align}
\label{eq:momrg1preal}
R^{N}_{\infty}(\xi) & = \tfrac{1}{2L}\vep_{N}^{\frac{1}{2}}\sum_{x\in\Lambda_{N}}\xi_{x}\!\ e^{i\frac{\pi}{2L}(\!\ .\!\ -x)}\tfrac{\sin(\pi\vep_{N}^{-1}(\!\ .\!\ -x))}{\sin(\frac{\pi}{2L}(\!\ .\!\ -x)},
\end{align}
for $\xi\in\fh_{N,{\color{blue}\pm}}$. This makes it evident that images, $\alpha^{N}_{\infty}(a(\xi))$, of lattice-localized localized operators, $a(\xi)$, $\xi\in\fh_{N,{\color{blue}\pm}}(I)$, are not strictly localized in $I\subset S^{1}_{L}$, in contrast with the wavelet scaling maps.

\subsection{The scaling limit of lattice vacua}
\label{sec:latvaclim}
We are now in a position to determine the scaling limit of the family of lattice vacua $\{\omega^{(N)}_{0,{\color{blue}\pm}}\}_{N\in\N_{0}}$ according to \eqref{eq:rgflow}. Using either the wavelet or the momentum-cutoff renormalization group, we find:
\begin{align}
\label{eq:renvac}
\omega^{(N)}_{M,{\color{blue}\pm}}(a(\xi)a^{\dag}(\eta)) & = \langle R^{N}_{M}(\xi),(1-S^{(N+M)}_{0})R^{N}_{M}(\eta)\rangle_{N+M} \\ \nonumber
& = \tfrac{1}{2L_{N+M}}\sum_{k\in\Gamma_{N+M,{\color{blue}\pm}}}\langle R^{N}_{M}(\hat{\xi})_{k},P^{(+)}_{\lambda_{N+M}}(k)R^{N}_{M}(\hat{\eta})_{k}\rangle_{\C^{2}},
\end{align}
such that the scaling-limit states $\omega^{(N)}_{\infty,{\color{blue}\pm}}$, $N\in\N_{0}$, is well-defined if we impose the \textit{renormalization condition},
\begin{align}
\label{eq:rgcond}
\lim_{N\rightarrow\infty}\vep_{N}^{-1}\lambda_{N} & = m,
\end{align}
for some $m\geq0$ (the \textit{physical mass} of the continuum fermion):
\begin{align}
\label{eq:scalinglim}
\omega^{(N)}_{\infty,{\color{blue}\pm}}(a(\xi)a^{\dag}(\eta)) & = \tfrac{1}{2L}\sum_{k\in\Gamma_{\infty,{\color{blue}\pm}}}\langle R^{N}_{\infty}(\hat{\xi})_{k},P^{(+)}_{m} R^{N}_{\infty}(\hat{\eta})_{k}\rangle_{\C^{2}},
\end{align}
where $P_{m}(k) = (2\omega_{m}(k))^{-1}(\omega_{m}(k)\mathds{1}_{2}+n_{m}(k)\cdot\sigma)$ with $\omega_{m}(k)=(m^{2}+k^{2})^{\frac{1}{2}}$ and $n_{m}(k) = k e_{y} + m e_{z}$.
If we use the Majorana mass term \eqref{eq:majmass} instead of the Dirac mass term, $P^{(+)}_{m}$ instead takes the form:
\begin{align}
\label{eq:majprojlim}
P^{(+)}_{m}(k) = (2\omega_{m}(k))^{-1}(\omega_{m}(k)\mathds{1}_{2} + n_{m}(k)\cdot\sigma),
\end{align}
with  $n_{m}(k) = k e_{y} + m e_{x}$, which satisfies: $CP^{(+)}_{m} = (\mathds{1}-P^{(+)}_{m})C$ for all $m\geq0$. In this case, the scaling limit is also well-defined as a state on $\fB_{\infty,{\color{blue}\pm}}$:
\begin{align}
\label{eq:scalinglimsdc}
\omega^{(N)}_{\infty,{\color{blue}\pm}}(\Psi(\xi)\Psi(\eta)^{*}) & = \tfrac{1}{2L}\sum_{k\in\Gamma_{\infty,{\color{blue}\pm}}}\langle R^{N}_{\infty}(\hat{\xi})_{k},P^{(+)}_{m} R^{N}_{\infty}(\hat{\eta})_{k}\rangle_{\C^{2}}.
\end{align}

The one-particle operator $P^{(+)}_{0}$ of the massless scaling-limit state is directly related to the chiral projections $p_{\pm}$, cp.~\eqref{eq:chiral1p}:
\begin{align}
\label{eq:masslesslimproj}
P^{(+)}_{0}(k) & = \tfrac{1}{2}(\mathds{1}_{2}+\sign(k)\sigma_{y}) = p_{\sign(k)},
\end{align}
which results in the following expression of the scaling limit on the chiral algebra $\fA^{(\pm)}_{N,{\color{blue}\pm}}$ (and again similar for $\fB^{(\pm)}_{N,{\color{blue}\pm}}$),
\begin{align}
\label{eq:chiral2plimit}
\omega^{(N)}_{\infty,{\color{blue}\pm}}(\psi_{\pm}(\xi)\psi_{\pm}^{\dagger}(\eta)) & = \tfrac{1}{2L}\sum_{k\in\Gamma_{\infty,{\color{blue}\pm}}}\tfrac{1}{2}(1\pm\sign(k))\overline{R^{N}_{\infty}(\xi)_{k}}R^{N}_{\infty}(\eta)_{k}, \\ \nonumber
\omega^{(N)}_{\infty,{\color{blue}\pm}}(\psi_{\pm}(\xi)\psi_{\mp}^{\dagger}(\eta)) & = 0,
\end{align}
such that the chiral parts decouple. As expected, we observe that the the scaling limit of the massless lattice splits into independent components relative to the chiral decomposition \eqref{eq:chiral1p}, and \eqref{eq:chiral2plimit} is related to the projection, cf. \eqref{eq:pmom},
\begin{align}
\label{eq:hardyproj}
P^{+} & : L^{2}(S^{1}_{L}) \longrightarrow H^{2}(S^{1}_{L}),
\end{align}
onto Hardy space and its complement $P^{-} = 1-P^{+}$. Inverting the Fourier transform in \eqref{eq:chiral2plimit}, we obtain a real-space expression that is seen to approximate the standard infinite-volume two-point function in the limit $L\rightarrow\infty$ \cite{AbdallaNonPerturbativeMethods}:
\begin{align}
\label{eq:chiral2preal}
\omega^{(N)}_{\infty,{\color{blue}\pm}}(\psi_{\pm}(\xi)\psi_{\pm}^{\dagger}(\zeta)) & = \mp\int_{S^{1}_{L}}dx\!\ \int_{S^{1}_{L}}dx’\!\ (2L)^{-1}\tfrac{\overline{\xi\ast_{\Lambda_{N}}s^{(\vep_{N})}(x)}\zeta\ast_{\Lambda_{N}}s^{(\vep_{N})}(x’)}{e^{i\frac{\pi}{L}(x-x’\pm i0^{+})}-1} \\ \nonumber
& \rightarrow \pm\tfrac{i}{2\pi}\int_{\RR}dx\!\ \int_{\RR}dx’\!\ \frac{\overline{\xi\ast_{\Lambda_{N}}s^{(\vep_{N})}(x)}\zeta\ast_{\Lambda_{N}}s^{(\vep_{N})}(x’)}{x-x’\pm i0^{+}}.
\end{align}

Now that we have established sensible candidates, \eqref{eq:scalinglim} \& \eqref{eq:scalinglimsdc}, for the scaling limits of the families of lattice vacua, we establish the convergence to these limits for the wavelet renormalization group using the results of Section \ref{sec:convergence}.

\begin{lemma}
\label{lem:stateconv}
Let $\{\omega^{(N)}\}_{N\in\N_{0}}$ be a family of quasi-free lattice states, each defined by an operator $0\leq S^{(N)}\leq1$ on $\fh_{N,{\color{blue}\pm}}$ ($CS^{(N)} = (1-S^{(N)})C$ in the self-dual case). Assume that for all $N\in\N_{0}$ the kernel of $S^{(N)}$ is diagonal in momentum space and converges point-wise to the kernel of an operator $0\leq S\leq 1$ ($CS = (1-S)C$),
\begin{align}
\label{eq:kernelpointlim}
\lim_{N\rightarrow\infty}S^{(N)}(k)_{ij} & = S(k)_{ij}, & k\in\Gamma_{N,{\color{blue}\pm}}, i,j=1,2,
\end{align}
where $\langle e_{k,i}, S^{(N)}e_{l,j}\rangle_{N} = 2L_{N}\delta_{k,l}S^{(N)}(k)_{ij}$ and $\langle e_{k,i}, S e_{l,j}\rangle_{\infty}= 2L\delta_{k,l}S(k)_{ij}$. Then, the renormalization group flow \eqref{eq:rgflow} of $\omega^{(N)}$ converges to $\omega^{(N)}_{S} = \omega_{S}\circ\alpha^{N}_{\infty}$ :
\begin{align}
\label{eq:2pconvnorm}
\lim_{M\rightarrow\infty}\|\omega^{(N)}_{M}-\omega^{(N)}_{S}\| & = 0.
\end{align}
\end{lemma} 
\begin{proof}
It is sufficient to prove the statement for the two-point functions because the states are quasi-free and the algebras $\fA_{N,{\color{blue}\pm}}$ and $\fB_{N,{\color{blue}\pm}}$ are generated (in norm) the by algebraic span of annihilation and creation operators respectively the Majorana operators. We state the proof only for the complex fermions $\fA_{N,{\color{blue}\pm}}$ as it is completely analogous for the Majorana fermions $\fB_{N,{\color{blue}\pm}}$. By the Cauchy-Schwarz inequality and the properties of the renormalization group we have:
\begin{align*}
|(\omega^{(N)}_{M}-\omega^{(N)}_{S})(a(\xi)a^{\dag}(\eta)) & = \langle R^{N}_{\infty}(\xi), (R^{N+M}_{\infty}\!S^{(N+M)}\!R^{N}_{N+M}-SR^{N}_{\infty})(\eta)\rangle_{\infty} \\
& \leq \|\xi\|_{N}\|(R^{N+M}_{\infty}\!S^{(N+M)}\!R^{N}_{N+M}-SR^{N}_{\infty})(\eta)\|_{\infty} \\
& \leq \|\xi\|_{N} \tfrac{1}{2L_{N}}\!\!\!\sum_{l\in\Gamma_{N,{\color{blue}\pm}}}\sum_{j=1,2}|\hat{\eta}^{(j)}_{l}| \|(R^{N+M}_{\infty}\!S^{(N+M)}\!R^{N}_{N+M}\!-\!SR^{N}_{\infty})(e_{l,j})\|_{\infty} \\
& \leq \|\xi\|_{N} \|\eta\|_{N} \Big(\tfrac{1}{2L_{N}}\!\!\!\sum_{l\in\Gamma_{N,{\color{blue}\pm}}}\sum_{j=1,2}\!\!\|(R^{N+M}_{\infty}\!S^{(N+M)}\!R^{N}_{N+M}\!-\!SR^{N}_{\infty})(e_{l,j})\|_{\infty}^{2}\Big)^{\!\frac{1}{2}}.
\end{align*}
Next, we evaluate the last factor in the last line explicitly:
\begin{align*}
 & \tfrac{1}{2L_{N}}\!\!\!\sum_{l\in\Gamma_{N,{\color{blue}\pm}}}\sum_{j=1,2}\!\!\|(R^{N+M}_{\infty}\!S^{(N+M)}\!R^{N}_{N+M}\!-\!SR^{N}_{\infty})(e_{l,j})\|_{\infty}^{2} \\
 & = \tfrac{1}{2L_{N}}\!\!\!\sum_{l\in\Gamma_{N,{\color{blue}\pm}}}\sum_{j=1,2}\tfrac{1}{2L}\!\!\!\sum_{k\in\Gamma_{\infty,{\color{blue}\pm}}}\sum_{i=1,2}\!\!|\langle e_{k,i},(R^{N+M}_{\infty}\!S^{(N+M)}\!R^{N}_{N+M}\!-\!SR^{N}_{\infty})(e_{l,j})\rangle_{\infty}|^{2} \\
 & = \sum_{l\in\Gamma_{N,{\color{blue}\pm}}}\sum_{k\in\Gamma_{\infty,{\color{blue}\pm}}}\sum_{i,j=1,2}|\hat{s}(\vep_{N}k)|^{2}|(S^{(N+M)}_{\per})(k)_{ij}-S(k)_{ij}|^{2}\delta_{0,\frac{\pi}{L}(k-l)\!\!\!\mod 2L_{N}},
\end{align*}
where the last line follows from the scaling equation in momentum space \eqref{eq:infprod}, the inner products,
\begin{align}
\label{eq:waveletrginner}
\langle e_{k,i},R^{N}_{M}(e_{l,j})\rangle_{M} & = \delta_{ij}2^{\frac{M-N}{2}}\Big(\prod_{m=1}^{M-N}m_{0}(\vep_{N+m}k)\Big)\sum_{x\in\Lambda_{N}}e^{-i(k-l)x} \\ \nonumber
& = 2L_{N}\delta_{ij}2^{\frac{M-N}{2}}\Big(\prod_{m=1}^{M-N}m_{0}(\vep_{N+m}k)\Big)\delta_{0,\frac{L}{\pi}(k-l)\!\!\!\!\mod 2L_{N}}, \\ \nonumber
\langle e_{k,i},R^{N}_{\infty}(e_{l,j})\rangle_{\infty} & = \delta_{ij}\vep_{N}^{\frac{1}{2}}\hat{s}(\vep_{N}k)\sum_{x\in\Lambda_{N}}e^{-i(k-l)x} \\ \nonumber
& = 2L_{N}\delta_{ij}\vep_{N}^{\frac{1}{2}}\hat{s}(\vep_{N}k)\delta_{0,\frac{L}{\pi}(k-l)\!\!\!\!\mod 2L_{N}},
\end{align}
for all $k\in\Gamma_{\infty,{\color{blue}\pm}}$ and $l\in\Gamma_{N,{\color{blue}\pm}}$, and the periodic extension of $S^{(N+M)}(k)_{ij}$ to $k\in\Gamma_{\infty,{\color{blue}\pm}}$.

As a consequence of the decay estimate on $\hat{s}$ in Lemma \ref{lem:decay}, Lebesgue's dominated convergence theorem, and the fact that the states are quasi-free \eqref{eq:carqfs}, we know that $\omega^{(N)}$ converges to $\omega^{(N)}_{S}$ in the weak$^{*}$ sense. Since the algebra $\fA_{N,{\color{blue}\pm}}$ is finite dimensional, we know that this is equivalent to strong convergence, hence the result follows.
\end{proof}

\begin{cor}
\label{cor:stateconv}
Assuming the renormalization condition \eqref{eq:rgcond}, the renormalization group flow of the family of lattice vacua $\omega^{(N)}_{0,{\color{blue}\pm}}$ converges to $\omega^{(N)}_{\infty,{\color{blue}\pm}}$ (in norm) at any scale $N$.
\end{cor}
\begin{proof}
The renormalization condition implies the point-wise convergence of the kernel $S^{(N)}_{0}(k) = \1-P^{(+)}_{\lambda_{N}}(k) \rightarrow \1 - P^{(+)}_{m}$ for $N\rightarrow\infty$.
\end{proof}

\begin{rem}
\label{rem:stateconv}
The statement of Lemma \ref{lem:stateconv} is remains valid if we replace the wavelet renormalization group by the momentum-cutoff renormalization group. To see, this we observe that:
\begin{align*}
|\langle e_{k,i},(R^{N+M}_{\infty}\!S^{(N+M)}\!R^{N}_{N+M}\!-\!SR^{N}_{\infty})(e_{l,j})\rangle_{\infty}|^{2} & = 2L_{N}\vep_{N}^{\frac{1}{2}}\delta_{k,l}|(S^{(N+M)})(l)_{ij}-S(l)_{ij}|^{2},
\end{align*}
for all $k\in\Gamma_{\infty,{\color{blue}\pm}}$ and $l\in\Gamma_{N,{\color{blue}\pm}}$ because:
\begin{align}
\label{eq:momrginner}
\langle e_{k,i},R^{N}_{M}(e_{l,j})\rangle_{M} & = 2L_{N}\delta_{ij}2^{\frac{1}{2}(M-N)}\delta_{k,l}, \\ \nonumber
\langle e_{k,i},R^{N}_{\infty}(e_{l,j})\rangle_{\infty} & = 2L_{N}\delta_{ij}\vep_{N}^{\frac{1}{2}}\delta_{k,l}.
\end{align}
\end{rem}

\section{Approximation of conformal symmetries}
\label{sec:approxconf}
It follows from \eqref{eq:chiral2plimit} that the chiral subalgebras of the fermion and Majorana algebras, $\fA^{(\pm)}_{\infty,{\color{blue}\pm}}$ and $\fB^{(\pm)}_{\infty,{\color{blue}\pm}}$, together with the scaling limit $\omega_{{\color{blue}\pm}} = \omega^{(\infty)}_{\infty,{\color{blue}\pm}}$ are conformal field theories with central charge $c=1$ respectively $c=\tfrac{1}{2}$ \cite{DiFrancescoCFTBook, EvansQuantumSymmetriesOn}. But, we would like to understand how the $\Diff_{+}(S^{1}_{L})$-covariance arises in its differential form, i.e.~the \textit{Virasoro algebra}, in the scaling limit from the lattice data, $\fA_{N,{\color{blue}\pm}}$, $\fB_{N,{\color{blue}\pm}}$, and $\omega^{(N)}_{0,{\color{blue}\pm}}$.
To this end, we combine the method of operator-algebraic renormalization and its specific realizations introduced in Section \ref{sec:ffscaling} with the well-known results on Bogoliubov transformations and their implementability collected in Section \ref{sec:bogoliubov} to analyze the approximation of conformal symmetries by the so-called \textit{Koo-Saleur formula} \cite{KooRepresentationsOfThe}, see also  \cite{MilstedExtractionOfConformal, ZiniConformalFieldTheories}. 

\subsection{The Koo-Saleur approximants}
\label{sec:KSapprox}
To motivate the Koo-Saleur formula for the lattice analogues of the Virasoro generators, we first introduce some terminology, following \cite{MilstedExtractionOfConformal}. The discussion stays rather formal and assumes that we are given a conformal field theory defined on the circle $S^{1}_{L}$ with Hamiltonian,
\begin{align}
\label{eq:cftH}
H = \int_{S^{1}_{L}} dx\, h(x),
\end{align}
acting on some Hilbert space $\cH$, where $h$ is the Hamiltonian density (a local quantum field on $\cH$). The Hamiltonian density is itself given by the chiral and anti-chiral energy-momentum tensors, $T$ and $\overline{T}$, via
\begin{align}
\label{eq:cftHd}
h(x) = \tfrac{1}{2\pi}\big(T(x)+\overline{T}(x)\big).
\end{align}
The Fourier modes of the energy-momentum tensors,
\begin{align}
\label{eq:vira}
L_k &\!=\!\tfrac{2L}{(2\pi)^2}\int_{S^{1}_{L}}\!dx\, e^{ikx} T(x)\!+\!\tfrac{c}{24}\delta_{k,0}, & \overline{L}_k &\!=\!\tfrac{2L}{(2\pi)^2}\int_{S^{1}_{L}}\!dx\, e^{-ikx} \overline{T}(x)\!+\!\tfrac{c}{24}\delta_{k,0},
\end{align}
for $k\in\Gamma_{\infty,{\color{blue}+}}$, form a representation of the Virasoro algebra with \emph{central charge} $c$ \cite{DiFrancescoCFTBook},
\begin{align}
\label{eq:viralg}
[L_{k}, L_{k’}] & = \tfrac{L}{\pi}(k-k’)L_{k+k’} + \delta_{k+k’,0}\tfrac{c}{12}(\tfrac{L}{\pi}k)((\tfrac{L}{\pi}k)^2-1), \\ \nonumber 
[\overline{L}_{k}, \overline{L}_{k’}] & = \tfrac{L}{\pi}(k-k’)\overline{L}_{k+k’} + \delta_{k+k’,0}\tfrac{c}{12}(\tfrac{L}{\pi}k)((\tfrac{L}{\pi}k)^2-1), \\ \nonumber
[L_{k}, \overline{L}_{k’}] & = 0,
\end{align}
on $\cH$. Thus, the Fourier modes $H_k$ of the Hamiltonian density $h$ correspond to linear combinations of the Virasoro generators:
\begin{align}
\label{eq:fourierH}
H_k & = \tfrac{L}{\pi}\int_{S^{1}_{L}}\!dx\, e^{ikx}h(x) = L_k + \overline{L}_{-k} - \tfrac{c}{12}\delta_{k,0},
\end{align}
with $H_0 = \frac{L}{\pi}H$. In addition to the Virasoro generators, it is useful to consider their smeared versions,
\begin{align}
\label{eq:virasmeared}
L(X) & = \tfrac{1}{2L}\sum_{k\in\Gamma_{\infty,{\color{blue}+}}}\hat{X}_{k}L_{k}, & \overline{L}(X) & = \tfrac{1}{2L}\sum_{k\in\Gamma_{\infty,{\color{blue}+}}}\hat{X}_{k}\overline{L}_{k},
\end{align}
for a sufficiently regular differential loop $X\in\fl^{\alpha}(S^{1}_{L},\C)$ (cf.~Definition \ref{def:waveletrgcur} and below, see, for example, \cite{GoodmanProjectiveUnitaryPositive, CarpiOnTheUniqueness} precise definitions of smeared Virasoro generators).

Now, the Koo-Saleur proposal is to introduce, by analogy, the lattice Fourier modes of a lattice Hamiltonian, for example, $H^{(N)}_{0}$:
\begin{align}
\label{eq:hdfourier}
H_k^{(N)} = \tfrac{L}{\pi}\vep_{N}\sum_{x\in \Lambda_{N}} e^{ikx} h^{(N)}_x,
\end{align}
where $h^{(N)}:\Lambda_{N}\rightarrow\fA_{N,{\color{blue}\pm}}$ or $\fB_{N,{\color{blue}\pm}}$ is a suitable lattice Hamiltonian density. Here, we slightly abuse our notation for the sake of simplicity and denote $H^{(N)}_{k=0}=\tfrac{L}{\pi}H^{(N)}_{0}$ by $H^{(N)}_{0}$ as well. This leads to the following lattice analogues of $L_k$ and $\overline{L}_{k}$:
\begin{defn}[Koo-Saleur approximants]
\label{def:KSapprox}
The \emph{Koo-Saleur (KS) approximants} of a lattice Hamiltonian $H^{(N)}$ with density $h^{(N)}:\Lambda_{N}\rightarrow\fA_{N,{\color{blue}\pm}}$ or $\fB_{N,{\color{blue}\pm}}$ are given by\footnote{The prefactor $\sim k^{-1}$ used in \cite{MilstedExtractionOfConformal} has no intrinsic meaning on the lattice at finite scale and is replaced by $\tfrac{\pi}{2L_{N}}|\sin(\tfrac{1}{2}\vep_{N}k)|^{-1}\sim\tfrac{\pi}{L}k^{-1}$ for small $\vep_{N}$ which also avoids Fermion doublers at the boundary of the Brillouin zone.}:
\begin{align}
\label{eq:latviracpt}
L_{k}^{(N)} & = \tfrac{1}{2}\Big(H_{k}^{(N)}+\tfrac{\pi\vep_{N}}{2L\sin(\frac{1}{2}\vep_{N}k)}\Big[H_{k}^{(N)}\!,H_0^{(N)}\Big]\Big), \\ \nonumber
\overline{L}_{k}^{(N)} & = \tfrac{1}{2}\Big(H_{-k}^{(N)}+\tfrac{\pi\vep_{N}}{2L\sin(\frac{1}{2}\vep_{N}k)}\Big[H_{-k}^{(N)}\!,H_0^{(N)}\Big]\Big),
\end{align}
for $k\in\Gamma_{N,{\color{blue}+}}$. Given a differential loop $X\in\fl(\Lambda_{N},\C)$, we define the smeared KS approximants by:
\begin{align}
\label{eq:latvirasmeared}
L^{(N)}(X) & = \tfrac{1}{2L_{N}}\sum_{k\in\Gamma_{N,{\color{blue}+}}}\hat{X}_{k}L_{k}^{(N)}, & \overline{L}^{(N)}(X) & = \tfrac{1}{2L_{N}}\sum_{k\in\Gamma_{N,{\color{blue}+}}}\hat{X}_{k}\overline{L}_{k}^{(N)}.
\end{align}
\end{defn}
An alternative formula for the Koo-Saleur approximants is obtained by replacing $2\sin(\tfrac{1}{2}\vep_{N}k)$ with the dispersion relation $\omega_{\lambda_{N}=0}(k)$ of the free massless lattice Dirac Hamiltonian $H^{(N)}_{0}$ because $\tfrac{\pi}{2L_{N}}|\sin(\tfrac{1}{2}\vep_{N}k)|^{-1}=\tfrac{\pi}{L_{N}}\omega_{\lambda_{N}=0}(k)^{-1}$:
\begin{align}
\label{eq:latvira}
(1\!+\!\sign(k))L_{k}^{(N)}\!+\!(1\!-\!\sign(k))\overline{L}_{-k}^{(N)} & = H_{k}^{(N)} + \tfrac{\pi}{L_{N} \omega_{\lambda_{N}=0}(k)}\Big[H_{k}^{(N)}\!,H_0^{(N)}\Big], \\ \nonumber
(1\!+\!\sign(k))L_{-k}^{(N)}\!+\!(1\!-\!\sign(k))\overline{L}_{k}^{(N)} & = H_{-k}^{(N)} - \tfrac{\pi}{L_{N} \omega_{\lambda_{N}=0}(k)}\Big[H_{-k}^{(N)}\!,H_0^{(N)}\Big],
\end{align}
In the following we show that the (smeared) KS approximants, $L^{(N)}(X)$ and $\overline{L}^{(N)}(X)$, associated with the free massless lattice Dirac Hamiltonian \eqref{eq:stagH} converge to the (smeared) Virasoro generators, $L(S^{N}_{\infty}(X))$ and $\overline{L}(S^{N}_{\infty}(X))$ in a sense we make precise below.

As a direct consequence, we find that the (unitary\footnote{Note that $L^{(N)}_{k}{}^{*} = L^{(N)}_{-k}$ and $\overline{L}^{(N)}_{k}{}^{*} = \overline{L}^{(N)}_{-k}$.}) exponentials,
\begin{align}
\label{eq:latviraexp}
U^{(N)}(X) & = \exp(iL^{(N)}(X)), & \overline{U}^{(N)}(X) & = \exp(iL^{(N)}(X)),
\end{align}
as well as the (automorphic) dynamics,
\begin{align}
\label{eq:latviradyn}
\sigma^{(N)}_{tX} & =  \Ad_{U^{(N)}(tX)}, & \overline{\sigma}^{(N)}_{tX} & =  \Ad_{\overline{U}^{(N)}(tX)},
\end{align}
for $X\in\fl(\Lambda_{N},\R)$, converge to implementers of localized diffeomorphisms in $\Diff_{+}(S^{1}_{L})$ and their automorphic action on $\fA_{\infty,{\color{blue}\pm}}$ respectively $\fB_{\infty,{\color{blue}\pm}}$.

It is important to note that we have two degrees of freedom in the correspondence between lattice and continuum models: it may be necessary to shift the ground-state energy and scale the energy axis in order to match the two in the scaling limit. That is, we should allow for hamiltonian the following adjustment of the lattice Hamiltonian
\begin{align}
\label{eq:scaleshift}
H^{(N)} & \mapsto a_{N} H^{(N)} + b_{N} \mathds{1},
\end{align}  
where $a_{N}, b_{N}\in\R$ are model-dependent constants.

\subsection{The case of massless free lattice fermions}
\label{sec:ffKS}
We determine the KS approximants for the complex fermion algebra $\fA_{N,{\color{blue}\pm}}$ of the free massless ($\lambda_{N}=0$) Dirac Hamiltonian \eqref{eq:stagH} according to Definition \ref{def:KSapprox} using the following lattice Hamiltonian density\footnote{Note that the expression for $h^{(N)}_{x}$ is chosen to be symmetric with respect to the basic lattice translation $x\mapsto x+\vep_{N}$. A different choice, for example, directly using the summand in \eqref{eq:stagH} corresponds to partial integration in the scaling limit and, thus, a modification by a total divergence.}:
\begin{align}
\label{eq:lathdfermisym}
h^{(N)}_{x} & = \epsilon_{N}^{-2}\tfrac{1}{2}\Big(a^{(1)\!\!\ \dagger}_{x+\vep_{N}}a^{(2)}_{x} - a^{(1)\!\!\ \dagger}_{x}a^{(2)}_{x} +a^{(2)\!\!\ \dagger}_{x-\vep_{N}}a^{(1)}_{x} - a^{(2)\!\!\ \dagger}_{x}a^{(1)}_{x} + \textup{h.c.}\Big)\in\fA_{N,{\color{blue}\pm}}, & x\in\Lambda_{N}
\end{align}
Although we evaluate \eqref{eq:latviracpt} only for $\lambda_{N}=0$, we note that \eqref{eq:stagHmaj}) shows that the lattice approximation effectively generates a contribution that behaves like Majorana mass at finite scales. Similar to  \eqref{eq:FstagH} and the Fourier modes of the lattice Hamiltonian density are given by additive second quantization:
\begin{align}
\label{eq:Flathdfermisym}
H^{(N)}_{k} & = \tfrac{L}{\pi}\vep_{N}^{-1}dF_{0}(h^{(N)}_{k}) = \tfrac{1}{2\pi}\sum_{l\in\Gamma_{N,{\color{blue}\pm}}}\sum_{i,j=1,2}\hat{a}^{(i)\dag}_{l+k}h^{(N)}_{k}(l)_{ij}\hat{a}^{(j)}_{l}, \\ \nonumber
h^{(N)}_{k}(l) & = \begin{pmatrix}\hspace{-1cm} 0 & \hspace{-1cm} e^{-i\vep_{N}(l+\tfrac{k}{2})}\cos(\tfrac{1}{2}\vep_{N}k)-1 \\ e^{i\vep_{N}(l+\tfrac{k}{2})}\cos(\tfrac{1}{2}\vep_{N}k)-1 & 0 \end{pmatrix}.
\end{align}
Now, because of \eqref{eq:secondquantcomm}, the KS approximants are also given by additive second quantization:
\begin{align}
\label{eq:latviraSQ}
L_{k}^{(N)} & = \tfrac{1}{2}\tfrac{L}{\pi}\vep_{N}^{-1}dF_{0}\Big(h_{k}^{(N)}+(2\sin(\tfrac{1}{2}\vep_{N}k))^{-1}\Big[h_{k}^{(N)}\!,h^{(N)}_{0}\Big]\Big), \\ \nonumber
\overline{L}_{k}^{(N)} & = \tfrac{1}{2}\tfrac{L}{\pi}\vep_{N}^{-1}dF_{0}\Big(h_{-k}^{(N)}+(2\sin(\tfrac{1}{2}\vep_{N}k))^{-1}\Big[h_{-k}^{(N)}\!,h^{(N)}_{0}\Big]\Big),
\end{align}
Explicitly, we find for the KS approximants of the complex fermion algebra $\fA_{N,{\color{blue}\pm}}$:
\begin{align}
\label{eq:latviraD}
L^{(N)}_{k} & = \tfrac{L}{\pi}dF_{0}(\ell^{(N)}_{k}) = \tfrac{\vep_{N}}{2\pi}\sum_{l\in\Gamma_{N,{\color{blue}\pm}}}\sum_{i,j=1,2}\!\!\hat{a}^{(i)\dag}_{l+k}\ell^{(N)}_{k}(l)_{ij}\hat{a}^{(j)}_{l},\\ \nonumber
\ell^{(N)}_{k}(l) & = \tfrac{1}{2}\vep_{N}^{-1}\!\!\begin{pmatrix}\hspace{-1cm} -\sin(\vep_{N}(l+\tfrac{k}{2})) & \hspace{-1cm} e^{-i\vep_{N}(l+\tfrac{k}{2})}\cos(\tfrac{1}{2}\vep_{N}k)-1 \\ e^{i\vep_{N}(l+\tfrac{k}{2})}\cos(\tfrac{1}{2}\vep_{N}k)-1 & -\sin(\vep_{N}(l+\tfrac{k}{2})) \end{pmatrix}, \\[0.25cm] \nonumber
\overline{L}^{(N)}_{k} & = \tfrac{L}{\pi}dF_{0}(\overline{\ell}^{(N)}_{k}) = \tfrac{\vep_{N}}{2\pi}\sum_{l\in\Gamma_{N,{\color{blue}\pm}}}\sum_{i,j=1,2}\!\!\hat{a}^{(i)\dag}_{l-k}\overline{\ell}^{(N)}_{k}(l)_{ij}\hat{a}^{(j)}_{l},\\ \nonumber
\overline{\ell}^{(N)}_{k}(l) & = \tfrac{1}{2}\vep_{N}^{-1}\!\!\begin{pmatrix}\hspace{-1cm} \sin(\vep_{N}(l-\tfrac{k}{2})) & \hspace{-1cm} e^{-i\vep_{N}(l-\tfrac{k}{2})}\cos(\tfrac{1}{2}\vep_{N}k)-1 \\ e^{i\vep_{N}(l-\tfrac{k}{2})}\cos(\tfrac{1}{2}\vep_{N}k)-1 & \sin(\vep_{N}(l-\tfrac{k}{2})) \end{pmatrix}.
\end{align}
We find the explicit expressions for the KS approximants in terms of the chiral decomposition \eqref{eq:chiral1p} by applying the chiral projectors $p_{\pm}$ to the one-particle operators $\ell^{(N)}_{k}$, $\overline{\ell}^{(N)}_{k}$:
\begin{align}
\label{eq:latviraDchi}
\ell^{(N)}_{k}(l) & = \vep_{N}^{-1}\!\!\begin{pmatrix}\hspace{-0.25cm} -\sin(\tfrac{1}{4}\vep_{N}k)^{2}\sin(\vep_{N}(l+\tfrac{k}{2})\!) & \hspace{-0.25cm} -i(\sin(\tfrac{1}{2}\vep_{N}l)^{2}+\sin(\tfrac{1}{2}\vep_{N}(l+k)\!)^{2}) \\[0.25cm] i(\sin(\tfrac{1}{2}\vep_{N}l)^{2}+\sin(\tfrac{1}{2}\vep_{N}(l+k)\!)^{2}) & -\cos(\tfrac{1}{4}\vep_{N}k)^{2}\sin(\vep_{N}(l+\tfrac{k}{2})\!) \end{pmatrix},
\end{align}
\begin{align}
\label{eq:latviraDantichi}
\overline{\ell}^{(N)}_{k}(l) & = \vep_{N}^{-1}\!\!\begin{pmatrix}\hspace{-0.25cm} \cos(\tfrac{1}{4}\vep_{N}k)^{2}\sin(\vep_{N}(l-\tfrac{k}{2})\!) & \hspace{-0.25cm} i(\sin(\tfrac{1}{2}\vep_{N}l)^{2}+\sin(\tfrac{1}{2}\vep_{N}(l-k)\!)^{2}) \\[0.25cm] -i(\sin(\tfrac{1}{2}\vep_{N}l)^{2}+\sin(\tfrac{1}{2}\vep_{N}(l-k)\!)^{2}) & \sin(\tfrac{1}{4}\vep_{N}k)^{2}\sin(\vep_{N}(l-\tfrac{k}{2})\!) \end{pmatrix}.
\end{align}
From the latter expressions we can read of the KS approximants of the chiral and anti-chiral subalgebras \eqref{eq:chiralF}, i.e.~$p_{-}\ell^{(N)}_{k}(l)p_{-} = \ell^{(N)}_{-,k}$ respectively $p_{+}\overline{\ell}^{(N)}_{k}(l)p_{+} = \ell^{(N)}_{+,k}$:
\begin{align}
\label{eq:latvirasubchi}
L^{(N)}_{-,k} & = \tfrac{L}{\pi}dF_{0}(\ell^{(N)}_{-,k}) = \tfrac{-1}{2\pi}\sum_{l\in\Gamma_{N,{\color{blue}\pm}}}\cos(\tfrac{1}{4}\vep_{N}k)^{2}\sin(\vep_{N}(l+\tfrac{k}{2}))\hat{\psi}_{-|l+k}^{\dag}\hat{\psi}_{-|l}, \\ \nonumber
L^{(N)}_{+,k} & = \tfrac{L}{\pi}dF_{0}(\ell^{(N)}_{+,k}) = \tfrac{1}{2\pi}\sum_{l\in\Gamma_{N,{\color{blue}\pm}}}\cos(\tfrac{1}{4}\vep_{N}k)^{2}\sin(\vep_{N}(l-\tfrac{k}{2}))\hat{\psi}_{+|l-k}^{\dag}\hat{\psi}_{+|l}.
\end{align}
A direct computation verifies the validity of the lattice analogue of the Virasoro algebra \eqref{eq:viralg} at $c=0$ for $L^{(N)}_{k}$ and $\overline{L}^{(N)}_{k}$ at first order in $\vep_{N}$, with the following explicit expressions for the chiral and anti-chiral subalgebras:
\begin{align}
\label{eq:latviraCRchi}
\left[L^{(N)}_{-,k},L^{(N)}_{-,k’}\right] & = \tfrac{L_{N}}{\pi}2\sin(\tfrac{1}{2}\vep_{N}(k-k’))\frac{\cos(\frac{1}{4}\vep_{N}k)^{2}\cos(\frac{1}{4}\vep_{N}k’)^{2}}{\cos(\frac{1}{4}\vep_{N}(k+k’))^{2}}\left(L^{(N)}_{-,k+k’} + R^{(N)}_{-,k+k’}\right) \\ \nonumber
\left[L^{(N)}_{+,k},L^{(N)}_{+,k’}\right] & = \tfrac{L_{N}}{\pi}2\sin(\tfrac{1}{2}\vep_{N}(k-k’))\frac{\cos(\frac{1}{4}\vep_{N}k)^{2}\cos(\frac{1}{4}\vep_{N}k’)^{2}}{\cos(\frac{1}{4}\vep_{N}(k+k’))^{2}}\left(L^{(N)}_{+,k+k’} + R^{(N)}_{+,k+k’}\right),
\end{align}
where the remainders $R^{(N)}_{k+k’,-}$, $R^{(N)}_{k+k’,+}$ are in $\cO(\vep_{N}^{2})$ and given by:
\begin{align}
\label{eq:latviraR}
R^{(N)}_{-,k+k’} &\!=\!\!\tfrac{-1}{\pi}\!\cos(\tfrac{1}{4}\vep_{N}(k\!+\!k’)\!)^{2}\!\!\!\sum_{l\in\Gamma_{N,{\color{blue}\pm}}}\!\!\!\sin(\vep_{N}(l\!+\tfrac{k+k’}{2}\!)\sin(\tfrac{1}{2}\vep_{N}(l\!+\tfrac{k+k’}{2}\!)\hat{\psi}_{-|l+k+k’}^{\dag}\hat{\psi}_{-|l}, \\ \nonumber
R^{(N)}_{+,k+k’} &\!=\!\!\tfrac{1}{\pi}\!\cos(\tfrac{1}{4}\vep_{N}(k\!+\!k’)\!)^{2}\!\!\!\sum_{l\in\Gamma_{N,{\color{blue}\pm}}}\!\!\!\sin(\vep_{N}(l\!-\tfrac{k+k’}{2}\!)\!)\sin(\tfrac{1}{2}\vep_{N}(l\!-\tfrac{k+k’}{2}\!)\!)\hat{\psi}_{+|l-(k+k’)}^{\dag}\hat{\psi}_{+|l},
\end{align}

A similar computation starting from \eqref{eq:stagHmajSQ} yields the KS approximants for the two copies of the Majorana algebra $\fB_{N,{\color{blue}\pm}}$ making up $\fA_{N,{\color{blue}\pm}}$, which we also denote by $L^{(N)}_{k}$ and $\overline{L}^{(N)}_{k}$ by a slight abuse of notation:
\begin{align}
\label{eq:latviramaj}
L^{(N)}_{k} & = \tfrac{L}{\pi}dQ_{0}(\ell^{(N)}_{k}) = \tfrac{1}{2\pi}\sum_{l\in\Gamma_{N,{\color{blue}\pm}}}\sum_{\sigma,\sigma'=\pm}\!\!\tfrac{1}{2}\hat{\Psi}_{\sigma|-(l+k)}\ell^{(N)}_{k}(l)_{\sigma\sigma'}\hat{\Psi}_{\sigma'|l},\\ \nonumber
\ell^{(N)}_{k}(l) & = \tfrac{\vep_{N}^{-1}}{2}\!\!\left(\begin{matrix} \sin(\vep_{N}(l\!+\!k)\!)(1\!-\!\cos(\vep_{N}(l\!+\!\tfrac{k}{2})\!)\!)\!-\!\sin(\vep_{N}(l\!+\!\tfrac{k}{2}))(1\!-\!\cos(\vep_{N}(l\!+\!\tfrac{k}{2})\!)\cos(\vep_{N}\tfrac{k}{2})\!) &  \\[-0.1cm]  & \!\!\!\dots \\[-0.1cm] -i(\sin(\tfrac{1}{2}\vep_{N}l)^{2}\!+\!\sin(\tfrac{1}{2}\vep_{N}(l\!+\!k)\!)^{2})(1\!-\!\cos(\vep_{N}(l\!+\!\tfrac{k}{2})\!)\!) &  \end{matrix}\right. \\[0.1cm] \nonumber
& \left.\begin{matrix}  & i(\sin(\tfrac{1}{2}\vep_{N}l)^{2}\!+\!\sin(\tfrac{1}{2}\vep_{N}(l\!+\!k)\!)^{2})(1\!+\!\cos(\vep_{N}(l\!+\!\tfrac{k}{2})\!)\!) \\[-0.1cm] \dots &  \\[-0.1cm]  & -\sin(\vep_{N}(l\!+\!k)\!)(1\!+\!\cos(\vep_{N}(l\!+\!\tfrac{k}{2})\!)\!)\!-\!\sin(\vep_{N}(l\!+\!\tfrac{k}{2})\!)(1\!-\!\cos(\vep_{N}(l\!+\!\tfrac{k}{2})\!)\cos(\vep_{N}\tfrac{k}{2})\!)  \end{matrix}\right),
\end{align}
\begin{align}
\label{eq:latviramajanti}
\overline{L}^{(N)}_{k} & = \tfrac{L}{\pi}dQ_{0}(\overline{\ell}^{(N)}_{k}) = \tfrac{1}{2\pi}\sum_{l\in\Gamma_{N,{\color{blue}\pm}}}\sum_{\sigma,\sigma'=\pm}\!\!\tfrac{1}{2}\hat{\Psi}_{\sigma|-(l-k)}\overline{\ell}^{(N)}_{k}(l)_{\sigma\sigma'}\hat{\Psi}_{\sigma'|l},\\ \nonumber
\overline{\ell}^{(N)}_{k}(l) & = \tfrac{\vep_{N}^{-1}}{2}\!\!\left(\begin{matrix} \sin(\vep_{N}(l\!-\!k)\!)(1\!+\!\cos(\vep_{N}(l\!-\!\tfrac{k}{2})\!)\!)\!+\!\sin(\vep_{N}(l\!-\!\tfrac{k}{2})\!)(1\!-\!\cos(\vep_{N}(l\!-\!\tfrac{k}{2})\!)\cos(\vep_{N}\tfrac{k}{2})\!) &  \\[-0.1cm]  & \!\!\!\dots \\[-0.1cm] -i(\sin(\tfrac{1}{2}\vep_{N}l)^{2}\!+\!\sin(\tfrac{1}{2}\vep_{N}(l\!-\!k)\!)^{2})(1\!+\!\cos(\vep_{N}(l\!-\!\tfrac{k}{2})\!)\!) &  \end{matrix}\right. \\[0.1cm] \nonumber
& \left.\begin{matrix}  & i(\sin(\tfrac{1}{2}\vep_{N}l)^{2}\!+\!\sin(\tfrac{1}{2}\vep_{N}(l\!-\!k)\!)^{2})(1\!-\!\cos(\vep_{N}(l\!-\!\tfrac{k}{2})\!)\!) \\[-0.1cm] \dots &  \\[-0.1cm]  & -\sin(\vep_{N}(l\!-\!k)\!)(1\!-\!\cos(\vep_{N}(l\!-\!\tfrac{k}{2})\!)\!)\!+\!\sin(\vep_{N}(l\!-\!\tfrac{k}{2})\!)(1\!-\!\cos(\vep_{N}(l\!-\!\tfrac{k}{2})\!)\cos(\vep_{N}\tfrac{k}{2})\!)  \end{matrix}\right),
\end{align}
with identical expressions for $\tilde{\Psi}_{\pm}$. The analogues of the restrictions \eqref{eq:latvirasubchi} of $L^{(N)}_{k}$ and $\overline{L}^{(N)}_{k}$ to the chiral and anti-chiral subalgebras are (in agreement with \eqref{eq:latvirasubchi}):
\begin{align}
\label{eq:latvirasubchimaj}
L^{(N)}_{-,k} & = \tfrac{L}{\pi}dQ_{0}(\ell^{(N)}_{-,k}) \\ \nonumber
& = \!\tfrac{-1}{4\pi}\!\!\!\sum_{l\in\Gamma_{N,{\color{blue}\pm}}}\!\!\!\!\big(\!\sin(\vep_{N}(l\!+\!k)\!)\cos(\tfrac{1}{2}\vep_{N}(l\!+\!\tfrac{k}{2})\!)^{2} \\[-0.5cm] \nonumber
&\hspace{1.75cm}+\!\tfrac{1}{2}\sin(\vep_{N}(l\!+\!\tfrac{k}{2})\!)(1\!-\!\cos(\vep_{N}(l\!+\!\tfrac{k}{2})\!)\cos(\vep_{N}\tfrac{k}{2})\!)\!\big)\!\hat{\Psi}_{-|-(l+k)}\hat{\Psi}_{-|l} \\ \nonumber
& = \tfrac{-1}{4\pi}\sum_{l\in\Gamma_{N,{\color{blue}\pm}}}\cos(\tfrac{1}{4}\vep_{N}k)^{2}\sin(\vep_{N}(l+\tfrac{k}{2}))\hat{\Psi}_{-|-(l+k)}\hat{\Psi}_{-|l}, \\[0.1cm] \nonumber
L^{(N)}_{+,k} & = \tfrac{L}{\pi}dQ_{0}(\ell^{(N)}_{+,k}) \\ \nonumber
& = \!\tfrac{1}{4\pi}\!\!\!\sum_{l\in\Gamma_{N,{\color{blue}\pm}}}\!\!\!\!\big(\!\sin(\vep_{N}(l\!-\!k)\!)\cos(\tfrac{1}{2}\vep_{N}(l\!-\!\tfrac{k}{2})\!)^{2} \\[-0.5cm] \nonumber
&\hspace{1.75cm}+\!\tfrac{1}{2}\sin(\vep_{N}(l\!-\!\tfrac{k}{2})\!)(1\!-\!\cos(\vep_{N}(l\!-\!\tfrac{k}{2})\!)\cos(\vep_{N}\tfrac{k}{2})\!)\!\big)\hat{\Psi}_{+|-(l-k)}\hat{\Psi}_{+|l} \\ \nonumber
& = \tfrac{1}{4\pi}\sum_{l\in\Gamma_{N,{\color{blue}\pm}}}\cos(\tfrac{1}{4}\vep_{N}k)^{2}\sin(\vep_{N}(l-\tfrac{k}{2}))\hat{\Psi}_{+|-(l-k)}\hat{\Psi}_{+|l},
\end{align}
where we used the identity,
\begin{align}
\label{eq:latviramajantisym}
\sum_{l\in\Gamma_{N,{\color{blue}\pm}}}\tfrac{1}{2}(f_{\pm}(l,k)+f_{\pm}(-(l\pm k),k))\hat{\Psi}_{\pm|-(l\mp k)}\hat{\Psi}_{\pm|l} & = L_{N}\delta_{k,0}\sum_{l\in\Gamma_{N,{\color{blue}\pm}}}f_{\pm}(l,0),
\end{align}
 which is a consequence of the self-dual CAR, to arrive at each of the last lines.

\subsubsection{Relation with Temperley-Lieb algebras}
\label{sec:tl}
In the original work of Koo-Saleur \cite{KooRepresentationsOfThe} an essential step consists in relating the Hamiltonian lattice density $h^{(N)}_{x}$ to generators $e_{x}$ of a Temperley-Lieb algebra $\TL_{\delta}$ with loop parameter $\delta$. Let us briefly explain how such a relation with Temperley-Lieb algebras \cite{TemperleyRelationsBetweenThe, GrimmTheSpin12} arises in our setting, specifically in the case $c=\tfrac{1}{2}$ which is intimately connected with the transverse-field Ising model \cite{SchultzTwoDimensionalIsing}.
Invoking the self-dual formulation involving the Majorana mass term in Section \ref{sec:sdc}, we can rewrite the Hamilton \eqref{eq:stagHmaj} in terms of a pair of anti-commuting complex fermions $c^{(1)}$, $c^{(2)}$ in momentum space,
\begin{align}
\hat{c}^{(1)}_{k} & = \theta(k)\hat{\psi}_{+|k}+\theta(-k)\hat{\psi}^{\dag}_{-|-k}, & \hat{c}^{(2)}_{k} & = \theta(k)\hat{\psi}_{+|-k}+\theta(-k)\hat{\psi}^{\dag}_{-|k}
\end{align}
which yields:
\begin{align}
\label{eq:isingham1}
H^{(N)}_{0} & \!=\!\tfrac{1}{2L}\!\!\!\!\!\sum_{k\in\Gamma_{\textb{-},>0}} \!\!\!\begin{pmatrix} \hat{c}^{(1)}_{k} \\ \hat{c}^{(1)\!\ \dag}_{-k} \end{pmatrix}^{\!\!\dagger}\!\!\!\begin{pmatrix} \sin(\vep_{N}k) & \!\!\!\!-i((\cos(\vep_{N}k)-1) + \lambda_{N}) \\ i((\cos(\vep_{N}k)-1) + \lambda_{N}) & -\sin(\vep_{N}k) \end{pmatrix}\!\! \begin{pmatrix} \hat{c}^{(1)}_{k} \\ \hat{c}^{(1)\!\ \dag}_{-k} \end{pmatrix} \\ \nonumber
& \hspace{0.25cm}+\!\tfrac{1}{2L}\!\!\!\!\!\sum_{k\in\Gamma_{\textb{-},>0}} \!\!\!\begin{pmatrix} \hat{c}^{(2)}_{k} \\ \hat{c}^{(2)\!\ \dag}_{-k} \end{pmatrix}^{\!\!\dagger}\!\!\!\begin{pmatrix} -\sin(\vep_{N}k) & \!\!\!\!-i((\cos(\vep_{N}k)-1) + \lambda_{N}) \\ i((\cos(\vep_{N}k)-1) + \lambda_{N}) & \sin(\vep_{N}k) \end{pmatrix}\!\! \begin{pmatrix} \hat{c}^{(2)}_{k} \\ \hat{c}^{(2)\!\ \dag}_{-k} \end{pmatrix} .
\end{align}
Here, we restricted ourselves to anti-periodic boundary conditions (\textb{-}) for simplicity. Next, we apply additional Bogoliubov transformations,
\begin{align}
\begin{pmatrix} \hat{z}^{(1)}_{k} \\ \hat{z}^{(1)\!\ \dag}_{-k} \end{pmatrix} & = e^{-i\frac{\pi}{4}\sigma_{x}} \begin{pmatrix} \hat{c}^{(1)}_{k} \\ \hat{c}^{(1)\!\ \dag}_{-k} \end{pmatrix}, & \begin{pmatrix} \hat{z}^{(2)}_{k} \\ \hat{z}^{(2)\!\ \dag}_{-k} \end{pmatrix} & = i\sigma_{z}e^{i\frac{\pi}{4}\sigma_{x}} \begin{pmatrix} \hat{c}^{(2)}_{k} \\ \hat{c}^{(2)\!\ \dag}_{-k} \end{pmatrix}
\end{align}
to arrive at two copies of the Hamiltonian of transverse-field Ising model with mixed-sector boundary conditions \cite{SchuetzDualityTwistedBoundary}:
\begin{align}
\label{eq:isingham2}
H^{(N)}_{0} & \!=\!\tfrac{1}{2L}\!\!\!\!\!\sum_{k\in\Gamma_{\textb{-},>0}} \!\!\!\begin{pmatrix} \hat{z}^{(1)}_{k} \\ \hat{z}^{(1)\!\ \dag}_{-k} \end{pmatrix}^{\!\!\dagger} h^{(N)}_{\TIM}(k) \begin{pmatrix} \hat{z}^{(1)}_{k} \\ \hat{z}^{(1)\!\ \dag}_{-k} \end{pmatrix}+\!\tfrac{1}{2L}\!\!\!\!\!\sum_{k\in\Gamma_{\textb{-},>0}} \!\!\!\begin{pmatrix} \hat{z}^{(2)}_{k} \\ \hat{z}^{(2)\!\ \dag}_{-k} \end{pmatrix}^{\!\!\dagger} h^{(N)}_{\TIM}(k) \begin{pmatrix} \hat{z}^{(2)}_{k} \\ \hat{z}^{(2)\!\ \dag}_{-k} \end{pmatrix},
\end{align}
where $h^{(N)}_{\TIM}(k) = -\sin(\vep_{N}k)\sigma_{y}+((\cos(\vep_{N}k)-1)+\lambda_{N})\sigma_{z}$. Now, we are in a position to explicitly spell out the relation with two commuting` copies of the Temperley-Lieb algebra $\TL_{\delta=\sqrt{2}}$ (on $2^{N+2}$ strands). This is achieved by introducing the Majorana fermions on the refined lattice $\Lambda_{N+1}$,
\begin{align}
\label{eq:majising}
\psi^{(j)}_{x} & = z^{(j)}_{x}+z^{(j)\!\ \dag}_{x}, & \psi^{(j)}_{x+\vep_{N+1}} & = (-i)(z^{(j)}_{x}-z^{(j)\!\ \dag}_{x}), & j & = 1,2, & x & \in\Lambda_{N},
\end{align}
and defining the Temperley-Lieb generators:
\begin{align}
\label{eq:tlgen}
e^{(j)}_{x} & = 2^{-\frac{1}{2}}(1+i\psi^{(j)}_{x+\vep_{N+1}}\psi^{(j)}_{x}), & j & = 1,2, & x & \in\Lambda_{N+1}.
\end{align}
This leads to:
\begin{align}
\label{eq:isingham3}
H^{(N)}_{0} & \!=\! \tfrac{i}{4}\vep_{N+1}^{-1}\sum_{j=1}^{2}\sum_{x\in\Lambda_{N}}\big((1-\lambda_{N})\psi^{(j)}_{x+\vep_{N+1}}\psi^{(j)}_{x}+\psi^{(j)}_{x+\vep_{N}}\psi^{(j)}_{x+\vep_{N+1}}\big) \\ \nonumber
& \!=\! \tfrac{1}{2\sqrt{2}}\vep_{N+1}^{-1}\sum_{j=1}^{2}\sum_{x\in\Lambda_{N}}\big((1-\lambda_{N})e^{(j)}_{x} + e^{(j)}_{x+\vep_{N+1}}\big) - (2-\lambda_{N})2^{N}\vep_{N+1}^{-1}
\end{align}
Therefore, the KS approximants at criticality ($\lambda_{N}=0$) can be equivalently computed using the Temperley-Lieb generators $e^{(j)}_{x}$ (see \cite{StottmeisterAnyonBraidingAnd} for further explanations concerning this perspective in the setting of OAR).

\subsubsection{A modification of the Koo-Saleur formula}
\label{sec:KSmod}
In Section \ref{sec:KSconv}, we prove the convergence of the KS approximants in representations with non-trivial central charge $c\neq0$ (see Lemma \ref{lem:KSconvc}). But, this requires a certain modification of the original formula by Koo and Saleur to avoid spurious contributions due to the periodic boundary conditions introduced by the Fourier transform for one-particle functions on the dual momentum-space lattice $\Gamma_{N,{\color{blue}\pm}}$. While this modification has a concise form in terms of one-particle operators, it is also possible to formulate it in terms of the Hamiltonian density \eqref{eq:lathdfermisym}, and, moreover, in terms of Hamiltonians expressed via Temperley-Lieb generators as in Section \ref{sec:tl}.\\
In the formalism of second quantization, the Fourier transform of the Hamiltonian density is given by a one-particle operator $h^{(N)}_{k}$ according to formula  \eqref{eq:Flathdfermisym}. The modification of the Koo-Saleur formula is then given by:
\begin{align}
\label{eq:hdmod1p}
h^{(N)}_{k} & = \tfrac{1}{2L_{N}}\sum_{l\in\Gamma_{N,{\color{blue}\pm}}}h^{(N)}_{k}(l)\chi_{\Gamma_{N}}(l+k)e_{l+k}\otimes \overline{e}_{l} \\ \nonumber
 & = \sum_{x\in\Lambda_{N}}\Big(\tfrac{1}{2L_{N}}\sum_{l'\in\Gamma_{N,{\color{blue}\pm}}}e^{il'x}\chi_{\Gamma_{N}}(l'+k)\Big)\Big(\tfrac{1}{2L_{N}}\sum_{l\in\Gamma_{N,{\color{blue}\pm}}}h^{(N)}_{k}(l)e^{-ilx}e_{l+k}\otimes \overline{e}_{l}\Big) \\ \nonumber
  & = \sum_{x\in\Lambda_{N}}\overline{\widehat{\chi_{\Gamma_{N}}(\cdot+k)}}(x)\underbrace{\Big(\tfrac{1}{2L_{N}}\sum_{l\in\Gamma_{N,{\color{blue}\pm}}}h^{(N)}_{k}(l)e_{l+k}\otimes \overline{\tau^{(N)}_{{\color{blue}\pm} |-x}e_{l}}\Big)}_{=h^{(N)}_{k}(x)} \\ \nonumber
  & = \Big(\overline{\widehat{\chi_{\Gamma_{N}}(\cdot+k)}}\ast_{\Lambda_{N}}h^{(N)}_{k}\Big)(0),
\end{align}
where $\overline{e}_{l}$ is short hand for $\langle e_{l},\!\ \cdot\!\ \rangle_{N}$ for all $l\in\Gamma_{N,{\color{blue}\pm}}$. The characteristic function $\chi_{\Gamma_{N}}$ is understood to remove additional contributions arising from periodic boundary conditions on $\Gamma_{N}$, i.e., $h^{(N)}_{k}$ is not understood to be $\tfrac{2\pi}{\vep_{N}}$-periodic in $k\in\Gamma_{N,{\color{blue}+}}$. This formula together with \eqref{eq:latviraSQ} implies that the modified one-particle KS approximants resulting from the modified $h^{(N)}_{k}$ can be written as:
\begin{align}
\label{eq:KSmod1p}
\ell^{(N)}_{k} & = \Big(\overline{\widehat{\chi_{\Gamma_{N}}(\cdot+k)}}\ast_{\Lambda_{N}}\ell^{(N)}_{k}\Big)(0) \\ \nonumber
& = \sum_{x\in\Lambda_{N}}\overline{\widehat{\chi_{\Gamma_{N}}(\cdot+k)}}(x)\tfrac{1}{2}\vep_{N}^{-1}\Big(h_{k}^{(N)}(x)+(2\sin(\tfrac{1}{2}\vep_{N}k))^{-1}\Big[h_{k}^{(N)}(x),h^{(N)}_{0}\Big]\Big) \\ \nonumber
& = \tfrac{1}{2L_{N}}\sum_{l\in\Gamma_{N,{\color{blue}\pm}}}\ell^{(N)}_{k}(l)\chi_{\Gamma_{N}}(l+k)e_{l+k}\otimes \overline{e}_{l},
\end{align}
with the obvious definition of $\ell^{(N)}_{k}(x)$. An analogous formula holds for $\overline{\ell}^{(N)}_{k}$. At the level of the second-quantized operators, we obtain the following modifications:
\begin{align}
\label{eq:hdmod}
H^{(N)}_{k} & = \Big(\overline{\widehat{\chi_{\Gamma_{N}}(\cdot+k)}}\ast_{\Lambda_{N}}H^{(N)}_{k}\Big)(0), \\ \nonumber
H^{(N)}_{k}(x) & = \tfrac{L}{2\pi}\epsilon_{N}^{-1}\!\!\sum_{y\in\Lambda_{N}}\!\!\!e^{iky}\Big(a^{(1)\!\!\ \dagger}_{y+\vep_{N}}a^{(2)}_{y-x}\!-\!a^{(1)\!\!\ \dagger}_{y}a^{(2)}_{y-x}\!+\!a^{(1)\!\!\ \dagger}_{y}a^{(2)}_{y-\vep_{N}-x}\!-\!a^{(1)\!\!\ \dagger}_{y}a^{(2)}_{y-x} \\[-0.25cm] \nonumber
 & \hspace{2.6cm} \!+\!a^{(2)\!\!\ \dagger}_{y-\vep_{N}}a^{(1)}_{y-x}\!-\!a^{(2)\!\!\ \dagger}_{y}a^{(1)}_{y-x}\!+\!a^{(2)\!\!\ \dagger}_{y-\vep_{N}}a^{(1)}_{y-x}\!-\!a^{(2)\!\!\ \dagger}_{y}a^{(1)}_{y+\vep_{N}-x}\Big),
\end{align}
and
\begin{align}
\label{eq:KSmod}
L^{(N)}_{k} & = \Big(\overline{\widehat{\chi_{\Gamma_{N}}(\cdot+k)}}\ast_{\Lambda_{N}}L^{(N)}_{k}\Big)(0), \\ \nonumber
L^{(N)}_{k}(x) & = \tfrac{1}{2}\Big(H_{k}^{(N)}(x)+\tfrac{\pi\vep_{N}}{2L\sin(\frac{1}{2}\vep_{N}k)}\Big[H_{k}^{(N)}(x),H_0^{(N)}\Big]\Big),
\end{align}
and analogously for $\overline{L}^{(N)}_{k}$. It follows from \eqref{eq:hdmod} and Section \ref{sec:tl} that $H^{(N)}_{k}(x)$ and, thus, $L^{(N)}_{k}(x)$, $\overline{L}^{(N)}_{k}(x)$ can be expressed in terms of Temperley-Lieb generators and the action of the braid group (realized in terms of the Kauffman bracket or Jones representation, cf.~\cite{SchuetzDualityTwistedBoundary, StottmeisterAnyonBraidingAnd}).\\
It is evident from these formulas that we recover the original Koo-Saleur formulas for $N\rightarrow\infty$ because $\chi_{\Gamma_{\infty}}\equiv1$, and $H^{(N)}_{k}(0)$ agrees with the unmodified $H^{(N)}_{k}$ (and similarly all the other modified quantities).

\subsubsection{The formal scaling limit}
\label{sec:formallim}
Before we turn to the analysis of the convergence of the KS approximants in the scaling limit in the sense of Definition \ref{def:convseq}, we study their formal limit in terms of the asymptotic renormalization group elements $\alpha^{N}_{\infty}$. This allows us to determine candidates for limit operators which should correspond to representatives of the Virasoro algebra \eqref{eq:viralg}.
\paragraph{Complex fermions $\fA_{N,{\color{blue}\pm}}$.}
Let us consider the complex fermion case first. As $L^{(N)}_{k}$, $\overline{L}^{(N)}_{k}$ have convenient expression in momentum space, we first observe that:
\begin{align}
\label{eq:rglimalgFT}
\alpha^{N}_{\infty}(\hat{a}^{(j)}_{k}) & = \sum_{x\in\Lambda_{N}}e^{ikx}\alpha^{N}_{\infty}(a^{(j)}_{x}) = \sum_{x\in\Lambda_{N}}e^{ikx}a^{(j)}(s^{(\vep_{N})}_{x}).
\end{align}
Thus, because $\alpha^{N}_{\infty}$ is $*$-morphism, we find:
\begin{align}
\label{eq:limrglatvira}
\alpha^{N}_{\infty}\left(L^{(N)}_{k}\right) & = \tfrac{\vep_{N}}{2\pi}\sum_{l\in\Gamma_{N,{\color{blue}\pm}}}\sum_{i,j=1,2}\alpha^{N}_{\infty}(\hat{a}^{(i)\dag}_{l+k})\ell^{(N)}_{k}(l)_{ij}\alpha^{N}_{\infty}(\hat{a}^{(j)}_{l}) \\ \nonumber
 & = \tfrac{1}{2\pi}\vep_{N}^{2}\!\!\!\!\sum_{x,y\in\Lambda_{N}}\sum_{i,j=1,2}\!\!\!\!a^{(i)\dag}(\vep_{N}^{-\frac{1}{2}} s^{(\vep_{N})}_{x})\!\!\left(\sum_{l\in\Gamma_{N,{\color{blue}\pm}}} \!\!e^{i(l+k)x}\ell^{(N)}_{k}(l)_{ij}e^{-ily} \!\!\right) \!\!a^{(j)}(\vep_{N}^{-\frac{1}{2}} s^{(\vep_{N})}_{y}),\\ \nonumber
\alpha^{N}_{\infty}\left(\overline{L}^{(N)}_{k}\right) & = \tfrac{\vep_{N}}{2\pi}\sum_{l\in\Gamma_{N,{\color{blue}\pm}}}\sum_{i,j=1,2}\alpha^{N}_{\infty}(\hat{a}^{(i)\dag}_{l-k})\overline{\ell}^{(N)}_{k}(l)_{ij}\alpha^{N}_{\infty}(\hat{a}^{(j)}_{l}) \\ \nonumber
 & = \tfrac{1}{2\pi}\vep_{N}^{2}\!\!\!\!\sum_{x,y\in\Lambda_{N}}\sum_{i,j=1,2}\!\!\!\!a^{(i)\dag}(\vep_{N}^{-\frac{1}{2}} s^{(\vep_{N})}_{x}) \!\!\left(\sum_{l\in\Gamma_{N,{\color{blue}\pm}}} \!\!e^{i(l-k)x}\overline{\ell}^{(N)}_{k}(l)_{ij}e^{-ily} \!\!\right) \!\!a^{(j)}(\vep_{N}^{-\frac{1}{2}} s^{(\vep_{N})}_{y}).
\end{align}
Now, we can formally take the scaling limit $N\rightarrow\infty$ using $\vep_{N}^{2}\sum_{x,y\in\Lambda_{N}}\rightarrow\int_{S^{1}_{L}\times S^{1}_{L}}\!dx\!\ dy$ and $\vep_{N}^{-\frac{1}{2}}s^{(\vep_{N})}_{x}\rightarrow\delta_{x}$ (in the sense of distributions), as well as,
\begin{align}
\label{eq:scalelim1pvira}
 \ell^{(N)}_{k}(l) & \rightarrow \ell_{k}(l) := -(l+\tfrac{k}{2})p_{-}, & \overline{\ell}^{(N)}_{k}(l) & \rightarrow \overline{\ell}_{k}(l) := (l-\tfrac{k}{2})p_{+}.
\end{align}
The latter can equally be expressed as a formal convergence of operators on the one-particle space $\fh_{\infty,{\color{blue}\pm}}$:
\begin{align}
\label{eq:latviraconv1pformal}
\tilde{\ell}^{(N)}_{k} & := R^{N}_{\infty}\ell^{(N)}_{k}R^{N\!\ *}_{\infty} \rightarrow \ell_{k}, & \tilde{\overline{\ell}}^{(N)}_{k} & := R^{N}_{\infty}\overline{\ell}^{(N)}_{k}R^{N\!\ *}_{\infty} \rightarrow \overline{\ell}_{k}.
\end{align}
Therefore, the formal scaling limit of the KS approximants is:
\begin{align}
\label{eq:scalelimvira}
L_{k} & = \tfrac{L}{\pi}dF_{0}(\ell_{k}), & \overline{L}_{k} & = \tfrac{L}{\pi}dF_{0}(\overline{\ell}_{k}) \\ \nonumber
& = \tfrac{1}{2\pi}\sum_{l\in\Gamma_{\infty,{\color{blue}\pm}}}\sum_{i,j=1,2}\hat{a}^{(i)\dag}_{l+k}\ell_{k}(l)_{ij}\hat{a}^{(j)}_{l} &
 & = \tfrac{1}{2\pi}\sum_{l\in\Gamma_{\infty,{\color{blue}\pm}}}\sum_{i,j=1,2}\hat{a}^{(i)\dag}_{l-k}\overline{\ell}_{k}(l)_{ij}\hat{a}^{(j)} _{l} \\ \nonumber
 & = \tfrac{-1}{2\pi}\sum_{l\in\Gamma_{\infty,{\color{blue}\pm}}}(l+\tfrac{k}{2})\hat{\psi}^{\dag}_{-|l+k}\hat{\psi}_{-|l} &  & = \tfrac{1}{2\pi}\sum_{l\in\Gamma_{\infty,{\color{blue}\pm}}}(l-\tfrac{k}{2})\hat{\psi}^{\dag}_{+|l-k}\hat{\psi}_{+|l} \\ \nonumber
 & = \tfrac{L}{\pi}dF_{0}(\ell_{-,k}) = L_{-,k}, &  & = \tfrac{L}{\pi}dF_{0}(\ell_{+,k}) = L_{+,k},
\end{align}
where we invoked the chiral decomposition \eqref{eq:chiralF} in the last line. We give a precise definition of $\ell_{\pm,k}$ as an unbounded operator on one-particle $\fh_{\infty,{\color{blue}\pm}})$ in the next subsection. The definition of its (normal-ordered) second quantization as unbounded operator on Fock space $\fF_{\ua}(\fh_{\infty,{\color{blue}\pm}})$ is consequence of the results provided in Section \ref{sec:bogoliubov}, see especially \cite{CareyOnFermionGauge}. Similarly, a formal computation of the commutators in the scaling limit, again exploiting that $\alpha^{N}_{\infty}$ is a $*$-morphism, yields:
\begin{align}
\label{eq:scalelimviraCRex}
\left[L_{k}, L_{k’}\right] &\!=\!\tfrac{L}{\pi}(k-k’)L_{k+k’}, & \left[\overline{L}_{k}, \overline{L}_{k’}\right] &\!=\!\tfrac{L}{\pi}(k-k’)\overline{L}_{k+k’}, & \left[L_{k}, \overline{L}_{k’}\right] &\!=\!0.
\end{align}
Thus, we formally recover a representation of the Virasoro algebra \eqref{eq:viralg} (with central charge $c=0$) in the scaling limit of complex lattice fermions. The vanishing of the central charge is expected because the computation pertains to $\fA_{\infty,{\color{blue}\pm}}$ itself respectively to the Fock representation of $\fA_{\infty,{\color{blue}\pm}}$ with vacuum $\Omega_{0}$, i.e.~the Schwinger cocycle \eqref{eq:schwinger} vanishes: $c_{S=0} = 0$.

\medskip

It is also possible to recover representations of the Virasoro algebra with non-vanishing central charge ($c\neq0$) in the formal scaling limit by passing to the GNS representations of $\fA_{\infty,{\color{blue}\pm}}$ given by the scaling limits $\omega^{(\infty)}_{\infty,{\color{blue}\pm}} = \omega_{{\color{blue}\pm}}$ of the lattice vacua \eqref{eq:scalinglim} \& \eqref{eq:chiral2plimit}. But, at this point, we are forced to distinguish between (${\color{blue}-}$)- and (${\color{blue}+}$)-boundary conditions as the scaling limit $\omega_{{\color{blue}-}}$ is a pure state while $\omega_{{\color{blue}+}}$ is not because of the degeneracies at zero momentum.

\medskip

Let us first discuss the Neveu-Schwarz sector ((${\color{blue}-}$)-boundary condition). According to \eqref{eq:masslesslimproj}, the scaling limit $\omega_{{\color{blue}-}}$ is a pure quasi-free state on $\fA_{\infty,-}$ determined by the Hardy projection $P^{+}$ and its complement $P^{-} = 1- P^{+}$ given in \eqref{eq:pmom} via $S = P^{-}\oplus P^{+}$. Moreover, by \eqref{eq:chiral2plimit} $\omega_{{\color{blue}-}}$ is compatible with the factorization, $\fA_{\infty,-}\cong\fA^{(+)}_{\infty,{\color{blue}-}}\otimes_{\Z_{2}}\fA^{(-)}_{\infty,-}$, in to chiral components. By \eqref{eq:repqfpcar} the GNS representation $\pi_{{\color{blue}-}}$ of the restrictions of $\omega_{{\color{blue}-}} = \omega^{(+)}_{{\color{blue}-}}\otimes\omega^{(-)}_{{\color{blue}-}}$ can be realized on Fock space, $\fF_{\ua}(\fh_{\infty,{\color{blue}-}})\cong\fF_{\ua}(\fh^{(+)}_{\infty,{\color{blue}-}})\otimes\fF_{\ua}(\fh^{(-)}_{\infty,{\color{blue}-}})$, as:
\begin{align}
\label{eq:cxrepchi}
\pi_{{\color{blue}-}}(\psi_{\pm}(\xi)) & = \psi_{\pm}(P^{\pm}\xi)+\psi^{\dag}_{\pm}(JP^{\mp}\xi), & \xi & \in\fh^{(\pm)}_{\infty,{\color{blue}-}}
\end{align}
where we choose the conjugation $J\hat{\xi} = \overline{\hat{\xi}}$ which satisfies $JP^{\pm} = P^{\pm}J$.

\medskip

By analogy with \eqref{eq:limrglatvira}, we compute the formal limit of $\tfrac{L}{\pi}dF_{{\color{blue}-}}(\tilde{\ell}^{(N)}_{\pm,k})=\!\ :\!\pi_{{\color{blue}-}}(\alpha^{N}_{\infty}(\tfrac{L}{\pi}dF_{0}(\ell^{(N)}_{\pm,k})))$: resulting in:
\begin{align}
\label{eq:viraCRcxrep}
\tfrac{L}{\pi}dF_{{\color{blue}-}}(\tilde{\ell}^{(N)}_{\pm,k}) & \rightarrow \tfrac{L}{\pi}dF_{{\color{blue}-}}(\ell_{\pm,k}) =: L_{\pm,k}, \\ \nonumber
[L_{\pm,k},L_{\pm,k'}] & = \tfrac{L}{\pi}(k-k')L_{\pm,k+k'} + \big(\tfrac{L}{\pi}\big)^{2}c_{P^{\mp}}(\ell_{\pm,k},\ell_{\pm,k'}) \\ \nonumber
& = \tfrac{L}{\pi}(k-k')L_{\pm,k+k'} + \delta_{k+k’,0}\tfrac{1}{12}(\tfrac{L}{\pi}k)((\tfrac{L}{\pi}k)^2-1),
\end{align}
where we used the fact that $\ell_{\pm,k}$ has off-diagonal parts in the Hilbert-Schmidt class. This shows that we formally obtain a representation of the Virasoro algebra with central charge $c=1$ as expected, and the GNS vector of $\omega_{{\color{blue}-}}$ has conformal weight $h=0$, i.e.~$L_{\pm,0}\Omega_{0} = 0$.

\bigskip

In the Ramond sector ((${\color{blue}+}$)-boundary condition), we purify $\omega_{{\color{blue}+}}$ by doubling according to \eqref{eq:doubleproj}. The doubled projections, taking the role of $P^{\pm}$ in the Neveu-Schwarz sector, are explicitly given by,
\begin{align}
\label{eq:doubleprojR}
P_{S_{\pm}} & = \tfrac{1}{2L}\!\!\sum_{k\in\Gamma_{\infty,{\color{blue}+}}\setminus\{0\}}\!\!\begin{pmatrix} \theta(\pm k) & 0 \\ 0 & \theta(\mp k) \end{pmatrix} e_{k}\otimes \langle e_{k},\!\ \cdot\!\ \rangle_{\infty} + \tfrac{1}{2L}\!\!\begin{pmatrix} \tfrac{1}{2} & \pm\tfrac{1}{2} \\ \pm\tfrac{1}{2} & \tfrac{1}{2} \end{pmatrix} e_{0}\otimes \langle e_{0},\!\ \cdot\!\ \rangle_{\infty},
\end{align}
on $\fh^{(\pm)}_{\infty,{\color{blue}+}}\oplus\fh^{(\pm)}_{\infty,{\color{blue}+}}$. The GNS representation $\pi_{{\color{blue}+}}$ of $\omega_{{\color{blue}+}} = \omega^{(+)}_{{\color{blue}+}}\otimes\omega^{(-)}_{{\color{blue}+}}$ is then realized on $\fF_{\ua}(\fh_{\infty,{\color{blue}+}}\oplus\fh_{\infty,{\color{blue}+}})\cong\fF_{\ua}(\fh^{(+)}_{\infty,{\color{blue}+}})^{\otimes 2}\otimes\fF_{\ua}(\fh^{(-)}_{\infty,{\color{blue}+}})^{\otimes 2}$ by:
\begin{align}
\label{eq:cxrepchiR}
\pi_{{\color{blue}+}}(\psi_{\pm}(\xi)) & = \psi_{\pm}(P_{S_{\pm}}(\xi\oplus0))+\psi^{\dag}_{\pm}((J\oplus J)P_{S_{\mp}}(\xi\oplus0)), & \xi & \in\fh^{(\pm)}_{\infty,{\color{blue}+}},
\end{align}
where $\psi_{\pm}$ on the right hand side is the obvious two-component extension.

\medskip

Computing the formal limit as in \eqref{eq:viraCRcxrep} of the KS approximants with respect to $\pi_{{\color{blue}+}}$ leads to:
\begin{align}
\label{eq:viraCRcxrepR}
\tfrac{L}{\pi}dF_{{\color{blue}+}}(\tilde{\ell}^{(N)}_{\pm,k}) & \rightarrow \tfrac{L}{\pi}dF_{{\color{blue}+}}(\ell_{\pm,k}) =: L_{\pm,k}, \\ \nonumber
[L_{\pm,k},L_{\pm,k'}] & = \tfrac{L}{\pi}(k-k')L_{\pm,k+k'} + \big(\tfrac{L}{\pi}\big)^{2}c_{P_{S_{\mp}}}(\ell_{\pm,k},\ell_{\pm,k'}) \\ \nonumber
& = \tfrac{L}{\pi}(k-k')L_{\pm,k+k'} + \delta_{k+k’,0}\tfrac{1}{12}(\tfrac{L}{\pi}k)((\tfrac{L}{\pi}k)^2+2).
\end{align}
We recover the Virasoro algebra \eqref{eq:viralg} with $c=1$ by redefining the limit at zero momentum,
\begin{align}
\label{eq:zeromodeR}
L_{\pm,0} & := \tfrac{L}{\pi}dF_{{\color{blue}+}}(\ell_{\pm,k})+\tfrac{1}{8},
\end{align}
which shows that the GNS vector of $\omega_{{\color{blue}+}}$ has conformal weight $h=\tfrac{1}{8}$, i.e.~$L_{\pm,0}\Omega_{0}\otimes\Omega_{0} = \tfrac{1}{8}\Omega_{0}\otimes\Omega_{0}$.
\paragraph{Majorana fermions $\fB_{N,{\color{blue}\pm}}$.} The computation of the formal limit of the KS approximants in the Majorana case is structurally identical, hinging only on the formal convergence \eqref{eq:latviraconv1pformal} of the one-particle operators \eqref{eq:latviramaj} \& \eqref{eq:latviramajanti}:
\begin{align}
\label{eq:scalelim1pviramaj}
 \ell^{(N)}_{k}(l) & \rightarrow \ell_{k}(l) := -(l+\tfrac{k}{2})\tfrac{1}{2}(\1-\sigma_{z}), & \overline{\ell}^{(N)}_{k}(l) & \rightarrow \overline{\ell}_{k}(l) := (l-\tfrac{k}{2})\tfrac{1}{2}(\1+\sigma_{z}).
\end{align}
Therefore, we only state the final results pertaining to $\fB_{\infty,{\color{blue}\pm}}$ itself,
\begin{align}
\label{eq:scalelimviramaj}
L_{k} & = \tfrac{L}{\pi}dQ_{0}(\ell_{k}), & \overline{L}_{k} & = \tfrac{L}{\pi}dQ_{0}(\overline{\ell}_{k})  \\ \nonumber
 & = \tfrac{-1}{4\pi}\sum_{l\in\Gamma_{\infty,{\color{blue}\pm}}}(l+\tfrac{k}{2})\hat{\Psi}_{-|-(l+k)}\hat{\Psi}_{-|l} &  & = \tfrac{1}{4\pi}\sum_{l\in\Gamma_{\infty,{\color{blue}\pm}}}(l-\tfrac{k}{2})\hat{\Psi}_{+|-(l-k)}\hat{\Psi}_{+|l} \\ \nonumber
 & = \tfrac{L}{\pi}dQ_{0}(\ell_{-,k}) = L_{-,k}, &  & = \tfrac{L}{\pi}dQ_{0}(\ell_{+,k}) = L_{+,k},
\end{align}
and the scaling limits of the lattice vacua $\omega_{{\color{blue}-}}$,
\begin{align}
\label{eq:scalelimviramajNS}
L_{\pm, k} & = \tfrac{L}{\pi}dQ_{{\color{blue}-}}(\ell_{\pm, k}),
\end{align}
and $\omega_{{\color{blue}+}}$,
\begin{align}
\label{eq:scalelimviramajR}
L_{\pm, k} & = \tfrac{L}{\pi}dQ_{{\color{blue}+}}(\ell_{\pm, k}) + \tfrac{1}{16}\delta_{k,0},
\end{align}
both of which have $c=\tfrac{1}{2}$ and conformal weight $h=0$ respectively $h=\tfrac{1}{16}$. The normal-ordered quantizations $dQ_{{\color{blue}+}}$ and $dQ_{{\color{blue}-}}$ are defined with respect to the self-dual analogues of \eqref{eq:cxrepchi}:
\begin{align}
\label{eq:majrepchi}
\pi_{{\color{blue}-}}(\Psi_{\pm}(\xi)) & = \psi_{\pm}(P^{\pm}\xi)+\psi^{\dag}_{\pm}(CP^{\mp}\xi), & \xi & \in\fh^{(\pm)}_{\infty,{\color{blue}-}}, \\
\label{eq:majrepchiR}
\pi_{{\color{blue}+}}(\Psi_{\pm}(\xi)) & = \psi_{\pm}(P_{S_{\pm}}(\xi\oplus0))+\psi^{\dag}_{\pm}((C\oplus (-C))P_{S_{\mp}}(\xi\oplus0)), & \xi & \in\fh^{(\pm)}_{\infty,{\color{blue}+}}.
\end{align}

\subsubsection{Convergence of the KS approximants}
\label{sec:KSconv}
In the computation of the formal scaling limit in Section \ref{sec:formallim}, we have seen that establishing convergence of one-particle operators \eqref{eq:latviraconv1pformal} in a rigorous way is central to deduce the convergence of the KS approximant to the Virasoro generators in the scaling limit. Now, we make this idea precise and argue that such a result shows that $\{L^{(N)}_{k}\}_{N\in\N_{0}}$ and $\{\overline{L}^{(N)}_{k}\}_{N\in\N_{0}}$ are convergent sequences in the sense of Definition \ref{def:convseq} with limits given by $L_{k}$ and $\overline{L}_{k}$. We focus on the convergence of the chiral and anti-chiral components, $\{L^{(N)}_{\pm, k}\}_{N\in\N_{0}}$, to $L_{\pm,k}$ as the convergence of the other components follows along the same line albeit the limits are the zero operators. Moreover, we show that the result extends to smeared KS approximant \eqref{eq:latvirasmeared} and Virasoro generators \eqref{eq:virasmeared}.

\medskip

We explicitly provide the statements for the complex fermion algebra $\fA_{\infty,{\color{blue}\pm}}\cong\fA^{(+)}_{\infty,{\color{blue}\pm}}\otimes_{\Z_{2}}\fA^{(-)}_{\infty,{\color{blue}\pm}}$. The statements in Majorana case $\fB_{\infty,{\color{blue}\pm}}\cong\fB^{(+)}_{\infty,{\color{blue}\pm}}\otimes_{\Z_{2}}\fB^{(-)}_{\infty,{\color{blue}\pm}}$ are completely analogous.

\paragraph{Virasoro generators.} As a preparation for the convergence statement, we need to properly define the one-particle operator $\ell_{\pm,k}$ of the chiral and anti-chiral Virasoro generators as unbounded operators on $\fh^{(\pm)}_{\infty,{\color{blue}\pm}}\cong L^{2}(S^{1}_{L})_{{\color{blue}\pm}}$.

\begin{defn}[One-particle Virasoro generators]
\label{def:vir}
The operators $\ell_{\pm,k}$ associated with $\fA^{(\pm)}_{\infty,{\color{blue}\pm}}$ and $\fB^{(\pm)}_{\infty,{\color{blue}\pm}}$ are given on the standard dense invariant domain $\cD_{\std}\subset L^{2}(S^{1}_{L})_{{\color{blue}\pm}}$ spanned by the plane waves $\{e_{m}\}_{m\in\Gamma_{\infty,{\color{blue}\pm}}}$ by the expressions:
\begin{align}
\label{eq:contvira1p}
\ell_{\pm,k} & = \tfrac{\pm1}{2\pi}\sum_{l\in\Gamma_{\infty,{\color{blue}\pm}}}(l\mp\tfrac{k}{2})e_{l\mp k}\otimes \overline{e}_{l},
\end{align}
where $\overline{e}_{l}$ is short hand for $\langle e_{l},\!\ \cdot\!\ \rangle_{\infty}$ for all $l\in\Gamma_{\infty,{\color{blue}\pm}}$. Since we are particularly interested in associated unitary groups and,
\begin{align}
\label{eq:contvira1pad}
\ell_{\pm,k}^{*} & = \ell_{\pm,-k},
\end{align}
holds on $\cD_{\std}$, we define the real and imaginary parts of $\ell_{\pm,k}$,
\begin{align}
\label{eq:contvira1pRI}
r_{\pm,k} & = \tfrac{1}{2}(\ell_{\pm,k}+\ell_{\pm,-k}), & \iota_{\pm,k} & = \tfrac{1}{2i}(\ell_{\pm,k}-\ell_{\pm,-k}),
\end{align}
on $\cD_{\std}$ as well.
\end{defn}
We note that \eqref{eq:contvira1pad} implies that $\ell_{\pm,k}$ is closable. In view of this definition, we use the following notation for the one-particle KS approximants \eqref{eq:latvirasubchi},
\begin{align}
\label{eq:latvira1p}
\ell^{(N)}_{\pm,k} & = \tfrac{\pm1}{2\pi}\sum_{l\in\Gamma_{N,{\color{blue}\pm}}}\cos(\tfrac{1}{4}\vep_{N}k)^{2}\sin(\vep_{N}(l\mp\tfrac{k}{2}))e_{l\mp k}\otimes \overline{e}_{l},
\end{align}
and for the real and imaginary parts,
\begin{align}
\label{eq:latvira1pRI}
r^{(N)}_{\pm,k} & = \tfrac{1}{2}(\ell^{(N)}_{\pm,k}+\ell^{(N)}_{\pm,-k}), & \iota^{(N)}_{\pm,k} & = \tfrac{1}{2i}(\ell^{(N)}_{\pm,k}-\ell^{(N)}_{\pm,-k}).
\end{align}
as these evidently satisfy:
\begin{align}
\label{eq:latvira1pad}
\ell^{(N)}_{\pm,k}{}^{*} & = \ell^{(N)}_{\pm,-k}.
\end{align}

The following bounds for all $j\in\N_{0}$ are straightforward:
\begin{align}
\label{eq:analyticvector}
\|\ell_{\pm,k}^{j}e_{m}\|_{\infty}^{2} & \leq 2L(\tfrac{L}{\pi})^{2j}\left\{\begin{matrix} |m|^{2j} & k=0 \\ |k|^{2j}\left(\tfrac{\Gamma(|\frac{m}{k}|+\frac{1}{2}+j)}{\Gamma(|\frac{m}{k}|+\frac{1}{2})}\right)^{2} & k\neq0 \end{matrix}\right.,\\ \nonumber
\|r_{\pm,k}^{j}e_{m}\|_{\infty}^{2} & \leq 2L(\tfrac{L}{2\pi})^{2j}\left\{\begin{matrix} |2 m|^{2j} & k=0 \\ \binom{2j}{j}|k|^{2j}\left(\tfrac{\Gamma(|\frac{m}{k}|+\frac{1}{2}+j)}{\Gamma(|\frac{m}{k}|+\frac{1}{2})}\right)^{2} & k\neq0 \end{matrix}\right.,
\end{align}
with an analogous bound for $\|\iota_{\pm,k}^{j}e_{m}\|_{\infty}^{2}$. Here, $\Gamma(x+1) = x\Gamma(x)$ denotes the Gamma function. These bounds imply that each plane wave $e_{m}$ is an analytic vector for $\ell_{\pm,k}$, $r_{\pm,k}$ and $\iota_{\pm,k}$ because the series,
\begin{align*}
\sum_{j=0}^{\infty}\tfrac{t^{j}}{j!}\|\ell_{\pm,k}^{j}e_{m}\|_{\infty} & < \infty, & \sum_{j=0}^{\infty}\tfrac{t^{j}}{j!}\|r_{\pm,k}^{j}e_{m}\|_{\infty} & < \infty, & \sum_{j=0}^{\infty}\tfrac{t^{j}}{j!}\|\iota_{\pm,k}^{j}e_{m}\|_{\infty} & < \infty,
\end{align*}
are convergent for some $t>0$. By Nelson’s analytic vector theorem \cite{NelsonAnalyticVectors, ReedMethodsOfModern2} we have:
\begin{cor}
\label{cor:viraesa}
The plane waves $\{e_{m}\}_{m\in\Gamma_{\infty,{\color{blue}\pm}}}\subset\cD_{\std}$ are a total set of analytic vectors for $\ell_{\pm,k}$, $r_{\pm,k}$ and $\iota_{\pm,k}$. Moreover, $r_{\pm,k}$ and $\iota_{\pm,k}$ are essentially self-adjoint on $\cD_{\std}$.
\end{cor}

Because of these observations, we do not distinguish between the $\ell_{\pm,k}$, $r_{\pm,k}$, $\iota_{\pm,k}$ and their closures in what follows.

\begin{rem}[Majorana one-particle Virasoro generators]
\label{rem:impcond}
In the Majorana case the one-particle Virasoro generators need to satisfy the implementability condition \eqref{eq:sdcconstraints} which is met by the one-particle KS approximants \eqref{eq:latvira1p} and the one-particle Virasoro generators \eqref{eq:contvira1p}:
\begin{align}
\label{eq:impcondKS}
C\ell^{(N)}_{\pm,k}{}^{*}C & = -\ell^{(N)}_{\pm,k}, & C\ell_{\pm,k}^{*}C & = -\ell_{\pm,k} \;.
\end{align}
\end{rem}

We are now in a position to establish the convergence of the one-particle KS approximants to the one-particle Virasoro generators in the sense of Definition \ref{def:convseq} using the asymptotic maps $R^{N}_{\infty}:\ltwo(\Gamma_{N,{\color{blue}\pm}},(2L_{N})^{-1})\cong\fh^{(\pm)}_{N,{\color{blue}\pm}}\rightarrow\fh^{(\pm)}_{\infty,{\color{blue}\pm}}\cong L^{2}(S^{1}_{L})_{{\color{blue}\pm}}$, either from the wavelet or the momentum-cutoff renormalization group. In particular, we show the convergence of \eqref{eq:latviraconv1pformal} in the strong operator topology on a dense, common core $\cD\subset\fh^{(\pm)}_{\infty,{\color{blue}\pm}}$, i.e.:
\begin{align}
\label{eq:latviraconv1p}
\lim_{N\rightarrow\infty}\|(\tilde{\ell}^{(N)}_{\pm,k} - \ell_{\pm,k})\xi\|_{\infty} & = 0,
\end{align}
for all $\xi\in\cD$. A similar state holds for the real and imaginary parts \eqref{eq:contvira1pRI} \& \eqref{eq:latvira1pRI}.

\medskip

For the momentum-cutoff renormalization group, we use $\cD=\cD_{\std}$ because:
\begin{align}
\label{eq:momdominv}
R^{N}_{\infty}(e_{k}) & = \vep_{N}^{-\frac{1}{2}}e_{k}, & k&\in\Gamma_{N,{\color{blue}\pm}}\;,
\end{align}
where $e_{k}\in\fh^{(\pm)}_{N,{\color{blue}\pm}}$ on the left but $e_{k}\in\fh^{(\pm)}_{\infty,{\color{blue}\pm}}$ on the right, i.e.~$\bigcup_{N\in\N_{0}}R^{N}_{\infty}(\fh^{(\pm)}_{N,{\color{blue}\pm}})=\cD_{\std}$.

\medskip

For the wavelet renormalization group, we use $\cD=\cD_{W}$ defined as:
\begin{align}
\label{eq:waveletdom}
\cD_{W} = \bigcup_{N\in\N_{0}}R^{N}_{\infty}(\fh^{(\pm)}_{N,{\color{blue}\pm}}),
\end{align}
which agrees with the span of the wavelet basis associated with the scaling function $s^{(\vep_{0})}\in L^{2}(S^{1}_{L})_{+}$ by construction \cite{DaubechiesTenLecturesOn}. That $\cD_{W}$ is dense common core for the one-particle Virasoro generators and their real and imaginary parts is due to the following argument: 

We observe that the graph norms of $\ell_{\pm,k} $, $r_{\pm,k}$ and $\iota_{\pm,k}$ are bounded by the norm of the first-order Sobolev space $h^{1}(\Gamma_{\infty,{\color{blue}\pm}})$:
\begin{align*}
\|\xi\|_{\ell_{\pm,k}}^{2} & = \|\xi\|_{\infty}^{2}+\|\ell_{\pm,k}\xi\|_{\infty}^{2} \leq \max\{1+(\tfrac{L k}{2\pi})^{2},(\tfrac{L}{\pi})^{2}\}\|\xi\|_{h^{1}(\Gamma_{\infty,{\color{blue}\pm}})}^{2},
\end{align*}
and similarly for $\|\xi\|_{r_{\pm,k}}$ and $\|\xi\|_{\iota_{\pm,k}}$. But, it is known that $\cD_{\std}$ and $\cD_{W}$ (provided the scaling function $s$ is sufficiently regular) are dense in $h^{1}(\Gamma_{\infty,{\color{blue}\pm}})$ \cite{WASA}. It follows that the graph-norm closures of $\cD_{\std}$ and $\cD_{W}$ agree, and we conclude that $\cD_{W}$ is a core for $\ell_{\pm,k}$, $r_{\pm,k}$ and $\iota_{\pm,k}$ \cite{ReedMethodsOfModern1}.

\paragraph{Wavelet renormalization.}
\label{sec:KSconvwav}
In this paragraph $R^{N}_{M}$ and $R^{N}_{\infty}$ as well as $\alpha^{N}_{M}$ and $\alpha^{N}_{\infty}$ denote the wavelet renormalization group of Definition \ref{def:waveletrg}.

We formulate the essential convergence statement \eqref{eq:latviraconv1p} as a lemma:
\begin{lemma}[Convergence of the one-particle KS approximants]
\label{lem:KSconv}
Let $\xi\in D_{W}$ and $s\in C^{\alpha}(\R)$ a sufficiently regular, compactly supported orthonormal scaling function. Then:
\begin{align*}
\lim_{N\rightarrow\infty}\|(\tilde{\ell}^{(N)}_{\pm,k} - \ell_{\pm,k})\xi\|_{\infty} & = 0,
\end{align*}
where $\tilde{\ell}^{(N)}_{\pm,k}:=R^{N}_{\infty}\ell^{(N)}_{\pm,k}R^{N\!\ *}_{\infty}$, and similarly:
\begin{align*}
\lim_{N\rightarrow\infty}\|(\tilde{r}^{(N)}_{\pm,k} - r_{\pm,k})\xi\|_{\infty} & = 0, & \lim_{N\rightarrow\infty}\|(\tilde{\iota}^{(N)}_{\pm,k} - \iota_{\pm,k})\xi\|_{\infty} & = 0.
\end{align*} 
\end{lemma}
\begin{proof}
We only spell out the proof for $\tilde{\ell}^{(N)}_{\pm,k}$ and $\ell_{\pm,k}$. We first observe that:
\begin{align*}
\tilde{\ell}^{(N)}_{\pm,k}(R^{M}_{\infty}(\xi)) & = \tilde{\ell}^{(N)}_{\pm,k}(R^{N}_{\infty}(R^{M}_{N}(\xi))) = R^{N}_{\infty}(\ell^{(N)}_{\pm,k}(R^{M}_{N}(\xi))),
\end{align*}
for $N>M$ and $\xi\in\fh_{M}$ because $(R^{N}_{\infty})^{*}R^{N}_{\infty} = \1_{\fh_{N}}$. Now, the Cauchy-Schwarz inequality implies:
\begin{align*}
& \|R^{N}_{\infty}(\ell^{(N)}_{\pm,k}(R^{M}_{N}(\xi))) - \ell_{\pm,k}(R^{M}_{\infty}(\xi))\|_{\infty}^{2} \\ 
& \leq \|\xi\|_{\fh_{M}}^{2}\tfrac{1}{2L}\sum_{m\in\Gamma_{\infty,{\color{blue}\pm}}}\tfrac{1}{2L_{M}}\sum_{l\in\Gamma_{M,{\color{blue}\pm}}}|\underbrace{\langle e_{m},R^{N}_{\infty}(\ell^{(N)}_{\pm,k}(R^{M}_{N}(e_{l})))\rangle_{\infty}}_{:=f^{(N)}_{k,\pm}(M,m,l)}-\underbrace{\langle e_{m},\ell_{\pm,k}(R^{M}_{\infty}(e_{l}))\rangle_{\infty}}_{:=f_{k,\pm}(M,m,l)}|^{2}.
\end{align*}
Next, we derive explicit expression for $f^{(N)}_{k,\pm}(M,m,l)$ and $f_{k,\pm}(M,m,l)$ to use the dominated convergence theorem with respect to $\ltwo(\Gamma_{\infty,{\color{blue}\pm}},(2L)^{-1})$ to deduce the convergence of the limit $N\rightarrow\infty$:
\begin{align*}
f^{(N)}_{k,\pm}(M,m,l) & = \tfrac{\pm1}{2\pi}\sum_{n\in\Gamma_{N,{\color{blue}\pm}}}\cos(\tfrac{1}{4}\vep_{N}k)^{2}\sin(\vep_{N}(n\mp\tfrac{k}{2}))\langle e_{m},R^{N}_{\infty}(e_{n\mp k})\rangle_{\infty}\langle e_{n},R^{M}_{N}(e_{l})\rangle_{N}, \\
f_{k,\pm}(M,m,l) & = \tfrac{\pm1}{2\pi}\sum_{n\in\Gamma_{\infty,{\color{blue}\pm}}}(n\mp\tfrac{k}{2})\underbrace{\langle e_{m},e_{n\mp k}\rangle_{\infty}}_{=2L\delta_{m,n\mp k}}\langle e_{n},R^{M}_{\infty}(e_{l})\rangle_{\infty}.
\end{align*}
Explicitly evaluating the inner product as in \eqref{eq:waveletrginner}, this leads to:
\begin{align*}
f^{(N)}_{k,\pm}\!(M,m,l) &\!=\!\tfrac{\pm1}{4\pi}\vep_{M}^{\frac{1}{2}} \cos(\tfrac{1}{4}\vep_{N}k)^{2}\hat{s}(\vep_{N}m)\!\!\sum_{n\in\Gamma_{N,{\color{blue}\pm}}}\!\!2L_{N}\sin(\vep_{N}(n\mp\tfrac{k}{2}))\delta_{0,\frac{L}{\pi}(m-(n\mp k))\!\!\!\!\mod 2L_{N}} \\
& \hspace{5cm} \times 2L_{M}\delta_{0,\frac{L}{\pi}(n-l)\!\!\!\!\mod 2L_{M}}\prod_{j=1}^{N-M}m_{0}(\vep_{M+j}n) \\
&\!=\!\tfrac{\pm1}{4\pi}\vep_{M}^{\frac{1}{2}} \cos(\tfrac{1}{4}\vep_{N}k)^{2}\hat{s}(\vep_{N}m)\!\!\sum_{n\in\Gamma_{\infty,{\color{blue}\pm}}}\!\!2L_{N}\sin(\vep_{N}(n\mp\tfrac{k}{2}))\delta_{0,\frac{L}{\pi}(m-(n\mp k))\!\!\!\!\mod 2L_{N}} \\
& \hspace{5cm} \times 2L_{M}\delta_{0,\frac{L}{\pi}(n-l)\!\!\!\!\mod 2L_{M}}\chi_{\Gamma_{N}}(n)\prod_{j=1}^{N-M}\!\!\!m_{0}(\vep_{M+j}n) \\
&\!=\!\tfrac{\pm L}{2\pi}2L_{M}\vep_{M}^{\frac{1}{2}} \cos(\tfrac{1}{4}\vep_{N}k)^{2}\hat{s}(\vep_{N}m)\sum_{n\in\Z}\tfrac{\sin(\vep_{N}(m\pm\frac{k}{2})\!)}{\vep_{N}} \delta_{0,\frac{L}{\pi}(m\pm k+\frac{2\pi}{\vep_{N}}n-l)\!\!\!\!\mod 2L_{M}}\\
& \hspace{5.5cm} \times\!\!\chi_{\Gamma_{N}}\!(m\!\pm\!k\!+\!\tfrac{2\pi}{\vep_{N}}n)\!\!\!\prod_{j=1}^{N-M}\!\!\!m_{0}(\vep_{M+j}(m\!\pm\!k\!+\!\tfrac{2\pi}{\vep_{N}}n)\!), \\
f_{k,\pm}\!(M,m,l) &\!=\!\tfrac{\pm L}{2\pi}2L_{M}\vep_{M}^{\frac{1}{2}}\hat{s}(\vep_{M}(m\!\pm\!k))(m\!\pm\!\tfrac{k}{2})\delta_{0,\frac{L}{\pi}(m\pm k-l)\!\!\!\!\mod 2L_{M}},
\end{align*}
where $\chi_{\Gamma_{N}}$ is the indicator function of the sublattice $\Gamma_{N,{\color{blue}\pm}}\subset\Gamma_{\infty,{\color{blue}\pm}}$, and we evaluated the sum over $n\in\Gamma_{\infty,{\color{blue}\pm}}$ using $\delta_{0,\frac{L}{\pi}(m-(n\mp k))\!\!\!\!\mod 2L_{N}}$ form the second to third line. Finally, we use the periodicity, $m_{0}(\vep_{M+j}(m+\tfrac{2\pi}{\vep_{N}}n)) = m_{0}(\vep_{M+j}m)$ for all $m\in\Gamma_{\infty,{\color{blue}\pm}}$, $n\in\Z$ and $j=1,\dots,N-M$, and:
\begin{align*}
\delta_{0,\frac{L}{\pi}(m\pm k+\frac{2\pi}{\vep_{N}}n-l)\!\!\!\!\mod 2L_{M}} & = \delta_{0,\frac{L}{\pi}(m\pm k-l)+2^{N-M}L_{M}n\!\!\!\!\mod 2L_{M}} = \delta_{0,\frac{L}{\pi}(m\pm k-l)\!\!\!\!\mod 2L_{M}}.
\end{align*}
Thus, we arrive at:
\begin{align*}
f^{(N)}_{k,\pm}\!(M,m,l) &\!=\!\tfrac{\pm L}{2\pi}2L_{M}\vep_{M}^{\frac{1}{2}} \cos(\tfrac{1}{4}\vep_{N}k)^{2}\hat{s}(\vep_{N}m)\sum_{n\in\Z}\tfrac{\sin(\vep_{N}(m\pm\frac{k}{2}))}{\vep_{N}} \delta_{0,\frac{L}{\pi}(m\pm k-l)\!\!\!\!\mod 2L_{M}}\\
& \hspace{5.5cm} \times \chi_{\Gamma_{N}}\!(m\!\pm\!k\!+\!\tfrac{2\pi}{\vep_{N}}n)\!\!\prod_{j=1}^{N-M}\!\!\!m_{0}(\vep_{M+j}(m\!\pm\!k)\!) \\
&\!=\!\tfrac{\pm L}{2\pi}2L_{M}\vep_{M}^{\frac{1}{2}} \cos(\tfrac{1}{4}\vep_{N}k)^{2}\hat{s}(\vep_{N}m)\tfrac{\sin(\vep_{N}(m\pm\frac{k}{2})\!)}{\vep_{N}} \delta_{0,\frac{L}{\pi}(m\pm k-l)\!\!\!\!\mod 2L_{M}}\tfrac{\hat{s}(\vep_{M}(m\pm k)\!)}{\hat{s}(\vep_{N}(m\pm k)\!)},
\end{align*}
because $\sum_{n\in\Z}\chi_{\Gamma_{N}}(m\pm k+\tfrac{2\pi}{\vep_{N}}n) = 1$ for any fixed $m\pm k\in\Gamma_{\infty,{\color{blue}\pm}}$. We also used \eqref{eq:infprod}. Next, we write:
\begin{align*}
f^{(N)}_{k,\pm}(M,m,l) & = \tfrac{\pm L}{2\pi}2L_{M}\vep_{M}^{\frac{1}{2}}f^{(N)}_{\pm k}(M,m)\delta_{0,\frac{L}{\pi}(m\pm k-l)\!\!\!\!\mod 2L_{M}}, \\
f_{k,\pm}(M,m,l) & =  \tfrac{\pm L}{2\pi}2L_{M}\vep_{M}^{\frac{1}{2}}f_{\pm k}(M,m)\delta_{0,\frac{L}{\pi}(m\pm k-l)\!\!\!\!\mod 2L_{M}}.
\end{align*}
where we defined:
\begin{align*}
f^{(N)}_{k}(M,m) & := \hat{s}(\vep_{M}(m+k)\!)\cos(\tfrac{1}{4}\vep_{N}k)^{2}\tfrac{\sin(\vep_{N}(m+\frac{k}{2})\!)}{\vep_{N}}\tfrac{\hat{s}(\vep_{N}m)}{\hat{s}(\vep_{N}(m+k)\!)}, \\
f_{k}(M,m) & := \hat{s}(\vep_{M}(m+k)\!)(m\!+\!\tfrac{k}{2}).
\end{align*}
This allows us to conclude that we have pointwise convergence (in $m\in\Gamma_{\infty,{\color{blue}\pm}}$):
\begin{align*}
\lim_{N\rightarrow\infty}f^{(N)}_{k}(M,m) & = f_{k}(M,m),
\end{align*}
for arbitrary $M<N$ and $k\in\tfrac{\pi}{L}\Z$, and, thus, also pointwise (in $m\in\Gamma_{\infty}$):
\begin{align*}
\lim_{N\rightarrow\infty}f^{(N)}_{k,\pm}(M,m,l) & = f_{k,\pm}(M,m,l),
\end{align*}
for arbitrary $l\in\Gamma_{M,{\color{blue}\pm}}$. It remains to show that $\lim_{N\rightarrow\infty}f^{(N)}_{k,\pm}(M,\!\ .\!\ ,l) = f_{k,\pm}(M,\!\ .\!\ ,l)$ in $\ltwo(\Gamma_{\infty,{\color{blue}\pm}},(2L)^{-1})$, which follows from the dominated convergence theorem if we provide a $g_{k}(M,\!\ .\!\ )\in\ltwo(\Gamma_{\infty,{\color{blue}\pm}}, (2L)^{-1})$ such that:
\begin{align*}
|f^{(N)}_{k}(M,m)| & \leq g_{k}(M,m).
\end{align*}
The existence of such a $g_{k}(N,\!\ .\!\ )$ is implied by,
\begin{align*}
|f^{(N)}_{k}(M,m)| & \leq |(m+\tfrac{k}{2})\hat{s}(\vep_{M}(m+k))|\Big|\tfrac{\hat{s}(\vep_{N}m)}{\hat{s}(\vep_{N}(m+k)\!)}\Big| \leq C_{k}|(m+\tfrac{k}{2})\hat{s}(\vep_{M}(m+k))|,
\end{align*}
for some $C_{k}>0$, because of the uniform continuity of $\hat{s}$, and the decay estimate \eqref{eq:decaysf} for sufficiently regular $s\in C^{\alpha}(\R)$. $\hat{s}$ is uniformly continuous since $\lim_{|k|\rightarrow\infty}|\hat{s}(k)|=0$ and $\hat{s}(k) = \prod_{j=1}^{\infty}m_{0}(2^{-j}k)$ is locally uniformly convergent.
\end{proof}

The lemma implies the convergence of the KS approximants in the standard Fock space representation ($c=0$).
\begin{thm}[Covergence of the KS approximants]
\label{thm:KSconv}
Let $s\in C^{\alpha}(\R)$ be a sufficiently regular, compactly supported orthonormal Daubechies scaling function. The KS approximants, $L^{(N)}_{\pm,k}$, converge strongly to the continuum Virasoro generators, $L_{\pm,k}$, on the dense, common core $\fF^{\alg}_{\ua}(\cD_{W})\subset\fF_{\ua}(\fh^{(\pm)}_{\infty,{\color{blue}\pm}})$ spanned by anti-symmetric Fock vectors with finitely many one-particle excitations in $\cD_{W}$ (finite $\cD_{W}$-particle number), i.e.:
\begin{align}
\label{eq:latviraconv}
\lim_{N\rightarrow\infty}\|(\alpha^{N}_{\infty}(L^{(N)}_{\pm,k}) - L_{\pm,k})a^{\dagger}(\xi_{1})\dots a^{\dagger}(\xi_{n})\Omega_{0}\| & = 0,
\end{align}
for all $n\in\N_{0}$ and $\xi_{1},\dots,\xi_{n}\in\cD_{W}$.
\end{thm}
\begin{proof}
We have by the definition of second quantization:
\begin{align*}
dF_{0}(o)a^{\dagger}(\xi_{1})\dots a^{\dagger}(\xi_{n})\Omega_{0} & = \sum^{n}_{k=1}a^{\dagger}(\xi_{1})\dots a^{\dagger}(o\xi_{k})\dots a^{\dagger}(\xi_{n})\Omega_{0},
\end{align*}
for any one-particle operator $o$ on $\cD_{o}\subset\fh^{(\pm)}_{\infty,{\color{blue}\pm}}$ such that $\xi_{1},...,\xi_{n}\in\cD_{o}$. Thus, because of \eqref{eq:secondquantboundvect}, we know that \eqref{eq:latviraconv} is implied by \eqref{eq:latviraconv1p}.
\end{proof}

The lemma also implies the convergence of the unitary Bogoliubov transformations associated with the KS approximants to those of the Virasoro generators (a more general statement for smeared KS approximant and Virasoro generators is stated below). Let us note that a quasi-free representation $\pi_{S}$ admits an implementation of the Virasoro generators $L_{\pm,k}$ by normal-ordered second quantization \eqref{eq:secondquantarakino} provided the one-particle operators $\ell_{\pm,k}$ satisfy the off-diagonal Hilbert-Schmidt bounds \eqref{eq:hilbertschmidt}, for example, the representations $\pi_{{\color{blue}\pm}}$ of the scaling limit states \eqref{eq:cxrepchi} \& \eqref{eq:cxrepchiR}. Thus, for admissible quasi-free representations, we have:
\begin{align}
\label{eq:KSconvgroup}
\tilde{\sigma}^{(N)}_{t} & = \Ad_{e^{it dF_{S}(\tilde{r}^{(N)}_{\pm,k})}}(\pi_{S}(a^{\dagger}(\xi)\!)\!) = \pi_{S}(a^{\dagger}(e^{it \tilde{r}^{(N)}_{\pm,k}}\xi)\!), \\ \nonumber
\sigma_{t} & = \Ad_{e^{it dF_{S}(r_{\pm,k})}}(\pi_{S}(a^{\dagger}(\xi)\!)\!) = \pi_{S}(a^{\dagger}(e^{it r_{\pm,k}}\xi)\!),
\end{align}
where the unitary $e^{it dF_{S}(r_{\pm,k})}$ is well-defined by essential self-adjointness of $dF_{S}(r_{\pm,k})$ (see below Theorem \ref{thm:KSconvc}), and similarly for $\tilde{\iota}^{(N)}_{\pm,k}$ and $\iota_{\pm,k}$ by Corollary \ref{cor:viraesa}. By construction the relation with the KS approximants is:
\begin{align}
\label{eq:KShermiteanparts}
dF_{S}(\tilde{\ell}^{(N)}_{k,\pm}) & =\ :\pi_{S}(\alpha^{N}_{\infty}(dF_{0}(\ell^{(N)}_{k,\pm}))):\ =\ :\pi_{S}(\alpha^{N}_{\infty}(L^{(N)}_{k,\pm})):. 
\end{align}
The analogous statements hold in the Majorana setting, for example, for the scaling limit representations \eqref{eq:majrepchi} \& \eqref{eq:majrepchiR}.

\begin{cor}[Convergence of KS Bogoliubov transformations]
\label{cor:KSconvbog}
Let $s\in C^{\alpha}(\R)$ be a sufficiently regular, compactly supported orthonormal Daubechies scaling function. In any admissible quasi-free representation $\pi_{S}$, we have:
\begin{align*}
\lim_{N\rightarrow\infty}\|\tilde{\sigma}^{(N)}_{t}(\pi_{S}(A)) - \sigma_{t}(\pi_{S}(A))\| & = 0, & \forall A\in\fA^{(\pm)}_{\infty,{\color{blue}\pm}},
\end{align*}
and uniformly on compact intervals in $t\in\R$.
\end{cor}
\begin{proof}
It is sufficient to prove the statement for $a = a^{\dagger}(\xi)$ with $\xi\in\fh^{(\pm)}_{\infty,{\color{blue}\pm}}$, because the multinomials $A = a^{\dagger}(\xi_{1})\dots a^{\dagger}(\xi_{n})a(\xi_{n+1})\dots a(\xi_{n+m})$ form a norm total set in $\fA^{(\pm)}_{\infty,{\color{blue}\pm}}$ and we have:
\begin{align*}
\|\tilde{\sigma}^{(N)}_{t}(AB) - \sigma_{t}(AB)\| & \leq \|\tilde{\sigma}^{(N)}_{t}(A)-\sigma_{t}(A)\|\|B\| +\|A\|\|\tilde{\sigma}^{(N)}_{t}(B)-\sigma_{t}(B)\|,
\end{align*}
for all $A,B\in\pi_{S}(\fA^{(\pm)}_{\infty,{\color{blue}\pm}})$. Now, let $o = r_{\pm,k}, \iota_{\pm,k}$ and $\tilde{o}^{(N)} = r^{(N)}_{\pm,k}, \tilde{\iota}^{(N)}_{\pm,k}$. Then, because of \eqref{eq:KSconvgroup}, we have:
\begin{align*}
\|\tilde{\sigma}^{(N)}_{t}(\pi_{S}(a^{\dagger}(\xi))) - \sigma_{t}(\pi_{S}(a^{\dagger}(\xi)))\| & = \|\pi_{S}(a^{\dagger}(e^{it \tilde{o}^{(N)}}\xi)) - \pi_{S}(a^{\dagger}(e^{ito}\xi))\|\\
& \leq \|(e^{it \tilde{o}^{(N)}} - e^{it o})\xi\|_{\infty}.
\end{align*}
Now, $\lim_{N\rightarrow\infty}\|(e^{it \tilde{o}^{(N)}} - e^{it o})\xi\|_{\infty} = 0$ uniformly on compact intervals in $t\in\R$ because $\lim_{N\rightarrow\infty}\|(\tilde{o}^{(N)} - o)\xi\|_{\infty} = 0$ for all $\xi\in\mathcal{D}_{W}$ by Lemma \ref{lem:KSconv} which entails strong resolvent convergence \cite{ReedMethodsOfModern1}.
\end{proof}

\begin{rem}
\label{rem:KSunitaryconv}
For multinomials $A = a^{\dagger}(\xi_{1})\dots a^{\dagger}(\xi_{n})a(\xi_{n+1})\dots a(\xi_{n+m})$ we have the estimate:
\begin{align}
\label{eq:KSmultiestimate}
\|\tilde{\sigma}^{(N)}_{t}(\pi_{S}(A)) - \sigma_{t}(\pi_{S}(A))\| & \leq \left(\prod_{l=1}^{n+m}\|\xi_{l}\|_{\infty}\right)\sum_{k=1}^{n+m}\frac{\|(e^{it \tilde{o}^{(N)}} - e^{it o})\xi_{k}\|_{\infty}}{\|\xi_{k}\|_{\infty}}.
\end{align}
\end{rem}

\begin{cor}[Convergence of KS Bogoliubov derivations]
\label{cor:KSderivationconv}
With the same assumptions and notation as in Corollary \ref{cor:KSconvbog}, we have the convergence of the quasi-free derivations,
\begin{align*}
\tfrac{d}{dt}_{|t=0}\tilde{\sigma}^{(N)}_{t}(\pi_{S}(a(\xi))) & = i[:\!dF_{S}(\tilde{o}^{(N)})\!:,a(\xi)] = a(i\tilde{o}^{(N)}\xi), \\
\tfrac{d}{dt}_{|t=0}\sigma_{t}(\pi_{S}(a(\xi))) & = i[:\!dF_{S}(o)\!:,a(\xi)] = a(io\xi),
\end{align*}
for $A\in\fA^{(\pm)}_{\infty,{\color{blue}\pm}}$ in the algebraic span of $a^{\dag}(\xi), a(\xi)$ with $\xi\in\cD_{W}$. Moreover, we have:
\begin{align}
\label{eq:KSmultiestimatederivation}
\|[:\!dF_{S}(\tilde{o}^{(N)})\!:,A] - [:\!dF_{S}(o)\!:,A]\| & \leq \left(\prod_{l=1}^{n+m}\|\xi_{l}\|_{\infty}\right)\sum_{k=1}^{n+m}\frac{\|(\tilde{o}^{(N)} - o)\xi_{k}\|_{\infty}}{\|\xi_{k}\|_{\infty}}\;
\end{align}
for $A = a^{\dagger}(\xi_{1})\dots a^{\dagger}(\xi_{n})a(\xi_{n+1})\dots a(\xi_{n+m})$.
\end{cor}

\paragraph{Momentum-cutoff renormalization.}
\label{sec:momconv}
The convergence results for the wavelet renormalization group lead to the convergence of the KS approximants to the Virasoro generators in the standard Fock representation ($c=0$). Clearly, we would like to achieve a similar statement for non-vanishing central charge ($c\neq 0$). Specifically, we would like to deduce the convergence of the KS approximants in the representations $\pi_{{\color{blue}\pm}}$ of the scaling limit states \eqref{eq:cxrepchi} \& \eqref{eq:cxrepchiR} as well as \eqref{eq:majrepchi} \& \eqref{eq:majrepchiR}. Such a result is achieved by using the momentum-cutoff renormalization group of Definition \ref{def:momrg}. Although, we initially lose the compatibility with the quasi-local structure (see Proposition \ref{prop:antilocal}) in this way, we explain in the next paragraph of this subsection how to remedy this issue by combining the wavelet renormalization group with the momentum-cutoff renormalization group for smeared KS approximants and Virasoro generators.

Moreover, it is necessary to slightly modify the KS approximants \eqref{eq:latvira1p},
\begin{align}
\label{eq:latvirachi1pmommod}
\ell^{(N)}_{k,\pm} & = \tfrac{\pm1}{2\pi}\sum_{l\in\Gamma_{N,{\color{blue}\pm}}} \cos(\tfrac{1}{4}\vep_{N}k)^{2}\sin(\vep_{N}(l\mp\tfrac{k}{2}))\chi_{\Gamma_{N}}(l\mp k)e_{l\mp k}\otimes \overline{e}_{l}\;,
\end{align}
where the additional factor $\chi_{\Gamma_{N}}$ has been introduced to avoid spurious contributions due to the periodic boundary conditions for one-particle functions on the dual momentum-space lattice $\Gamma_{N}$. As the scale $N$ can be chosen large compared to intrinsic scale of any fixed $k\in\Gamma_{M,{\color{blue}+}}$, i.e., $N\gg M$, this modification does not affect the asymptotic properties of $\ell^{(N)}_{k,\pm}$. A direct computation shows that the modified KS approximants still satisfy the adjointness relations \eqref{eq:latvira1pad} and implementability conditions \eqref{eq:impcondKS}\footnote{In the Ramond sector some additional care is necessary, i.e., $\chi_{\Gamma_{N}}$ needs to be defined as $\chi_{(-\frac{\pi}{\vep_{N}},\frac{\pi}{\vep_{N}})\cap \Gamma_{N}}$ to ensure is invariance under reflections about $0$.}. We note that the modification of the KS approximants at the one-particle level can also be translated into modification of the formula of the KS approximants in Definition \ref{def:KSapprox}.

\bigskip

Denoting by $R^{N}_{M}$ and $R^{N}_{\infty}$ as well as $\alpha^{N}_{M}$ and $\alpha^{N}_{\infty}$ the momentum-cutoff renormalization group throughout this paragraph, the one-particle operators $\tilde{\ell}^{(N)}_{k,\pm}$ are now given by:
\begin{align}
\label{eq:latvirachi1pmom}
\tilde{\ell}^{(N)}_{k,\pm} & = \tfrac{\pm1}{2\pi}\sum_{l\in\Gamma_{N,{\color{blue}\pm}}} \cos(\tfrac{1}{4}\vep_{N}k)^{2}\tfrac{\sin(\vep_{N}(l\mp\frac{k}{2}))}{\vep_{N}}\chi_{\Gamma_{N}}(l\mp k)e_{l\mp k}\otimes \overline{e}_{l}\;.
\end{align}
To prove the analogue of Theorem \ref{thm:KSconv} for $\pi_{{\color{blue}\pm}}$ (with $c=1,\tfrac{1}{2}$), we need to compare the matrix elements of $\tilde{\ell}^{(N)}_{\pm,k}$ with those of $\ell_{\pm,k}$:
\begin{align}
\label{eq:latvirachi1pmommat}
(\tilde{\ell}^{(N)}_{\pm,k})_{mn} &\!=\!\pm\tfrac{L}{\pi}2L\delta_{m,n\mp k}\cos(\tfrac{1}{4}\vep_{N}k)^{2}\tfrac{\sin(\vep_{N}(n\mp\frac{k}{2}))}{\vep_{N}}\chi_{\Gamma_{N}(n\!\mp\!k)}\chi_{\Gamma_{N}}(n), \\ \nonumber
(\ell_{\pm,k})_{mn} &\!=\!\pm\tfrac{L}{\pi}2L\delta_{m,n\mp k}(n\!\mp\!\tfrac{k}{2}).
\end{align}
If we have an approximating sequence $\{o^{(N)}\}_{N\in\N_{0}}$ of $o$ on the one-particle space $\fh^{(\pm)}_{\infty,{\color{blue}\pm}}$, the convergence of their normal-ordered second quantizations can be analyzed using the following estimate because of \eqref{eq:secondquantbound} \& \eqref{eq:hilbertschmidt}:
\begin{align}
\label{eq:noquantconv} \nonumber
\|dF_{S}(\tilde{o}^{(N)}\!-\!o)a^{\dag}\!(R^{M}_{\infty}(\eta_{1})\!)...a^{\dag}\!(R^{M}_{\infty}(\eta_{n})\!)\Omega_{0}\| &\!\leq\!\!\Big(\!\sum_{k=1}^{n}\!\Big\{\!\tfrac{\|(\tilde{o}^{(N)} - o)_{++}R^{M}_{\infty}(\eta_{k})\|_{\infty}+\|(\tilde{o}^{(N)} - o)^{t}_{- -}R^{M}_{\infty}(\eta_{k})\|_{\infty}}{\|\eta_{k}\|_{M}}\!\Big\} \\ \nonumber
 & \hspace{1cm} +\!(n\!+\!2)\|(\tilde{o}^{(N)}\!-\!o)_{+-}\|_{2}\!+\!n\|(\tilde{o}^{(N)}\!-\!o)_{-+}\|_{2}\!\Big) \\
 & \hspace{0.5cm} \times \prod_{l=1}^{n}\|\eta_{l}\|_{M}.
\end{align}
For $\tilde{o}^{(N)}=\tilde{\ell}^{(N)}_{\pm,k}$ and $o = \ell_{\pm,k}$, we find ($N>M$):
\begin{align}
\label{eq:momrgestimates}
& \|((\tilde{\ell}^{(N)}_{k,\pm}-\ell_{k,\pm})_{\pm\pm})^{(t)}R^{M}_{\infty}(\eta)\|_{\infty}^{2} \\ \nonumber
& \leq \|\eta\|_{\fh_{M}}^{2}(\tfrac{L}{\pi})^{2}\!\!\!\sum_{n\in\Gamma_{\infty,{\color{blue}\pm}}}\!\!\!\chi_{\Gamma_{M}}(n)\theta(\pm(n\!\mp\!k))\theta(\pm n)\Big|\!\cos(\tfrac{1}{4}\vep_{N}k)^{2}\tfrac{\sin(\vep_{N}(n\mp\frac{k}{2}))}{\vep_{N}}\chi_{\Gamma_{N}(n\!\mp\!k)}\!-\!(n\!\mp\!\tfrac{k}{2})\Big|^{2}, \\ \nonumber
& \|(\tilde{\ell}^{(N)}_{k,\pm}-\ell_{k,\pm})_{\pm\mp}\|_{2}^{2} \\ \nonumber
& \leq (\tfrac{L}{\pi})^{2}\!\!\!\sum_{n\in\Gamma_{\infty,{\color{blue}\pm}}}\!\!\theta(\pm(n\!\mp\!k))\theta(\mp n)\Big|\cos(\tfrac{1}{4}\vep_{N}k)^{2}\tfrac{\sin(\vep_{N}(n\mp\frac{k}{2}))}{\vep_{N}}\chi_{\Gamma_{N}(n\!\mp\!k)}\chi_{\Gamma_{N}}(n)\!-\!(n\!\mp\!\tfrac{k}{2})\Big|^{2}.
\end{align}
The sums in these estimates are finite (for fixed $M$) which entails convergence for $N\rightarrow\infty$. It is evident from \eqref{eq:momrgestimates}, that we have the analogue of Lemma \ref{lem:KSconv}:
\begin{lemma}
\label{lem:KSconvc}
Let $\xi\in D_{\std}$. Then:
\begin{align*}
\lim_{N\rightarrow\infty}\|(\tilde{\ell}^{(N)}_{\pm,k}\!-\!\ell_{\pm,k})\xi\|_{\infty} & = 0, & \lim_{N\rightarrow\infty}\|(\tilde{\ell}^{(N)}_{\pm,k}\!-\!\ell_{\pm,k})_{\pm\pm}\xi\|_{\infty} & = 0, \lim_{N\rightarrow\infty}\|(\tilde{\ell}^{(N)}_{\pm,k}\!-\!\ell_{\pm,k})_{\pm\mp}\|_{2} & = 0,
\end{align*}
where $\tilde{\ell}^{(N)}_{\pm,k}:=R^{N}_{\infty}\ell^{(N)}_{\pm,k}R^{N\!\ *}_{\infty}$, and similarly:
\begin{align*}
\lim_{N\rightarrow\infty}\|(\tilde{r}^{(N)}_{\pm,k}\!-\!r_{\pm,k})\xi\|_{\infty} & = 0, & \lim_{N\rightarrow\infty}\|(\tilde{r}^{(N)}_{\pm,k}\!-\!r_{\pm,k})_{\pm\mp}\xi\|_{\infty} & = 0, & \lim_{N\rightarrow\infty}\|(\tilde{r}^{(N)}_{\pm,k}\!-\!r_{\pm,k})_{\pm\mp}\|_{\infty} & = 0, \\
\lim_{N\rightarrow\infty}\|(\tilde{\iota}^{(N)}_{\pm,k}\!-\!\iota_{\pm,k})\xi\|_{\infty} & = 0, & \lim_{N\rightarrow\infty}\|(\tilde{\iota}^{(N)}_{\pm,k}\!-\!\iota_{\pm,k})_{\pm\pm}\xi\|_{\infty} & = 0, & \lim_{N\rightarrow\infty}\|(\tilde{\iota}^{(N)}_{\pm,k}\!-\!\iota_{\pm,k})_{\pm\mp}\|_{\infty} & = 0.
\end{align*} 
\end{lemma}
Because of Remark \ref{rem:impcond}, similar estimates as \eqref{eq:noquantconv} and \eqref{eq:momrgestimates} also hold for $dQ_{S}$ in Majorana setting. Applying the lemma, we have:
\begin{thm}[Covergence of the KS approximants for $\pi_{{\color{blue}\pm}}$]
\label{thm:KSconvc}
The KS approximants, $L^{(N)}_{\pm,k}$, converge strongly to the continuum Virasoro generators, $L_{\pm,k}$, on the dense, common core $\fF^{\alg}_{\ua}(\cD_{\std})\subset\fF_{\ua}(\fh^{(\pm)}_{\infty,{\color{blue}\pm}})$ spanned by anti-symmetric Fock vectors with finitely many one-particle excitations in $\cD_{\std}$ (finite $\cD_{\std}$-particle number):
\begin{align}
\label{eq:latviraconvc}
\lim_{N\rightarrow\infty}\|(:\!\pi_{{\color{blue}\pm}}(\alpha^{N}_{\infty}(L^{(N)}_{\pm,k})\!)\!: - L_{\pm,k})a^{\dagger}(\xi_{1})\dots a^{\dagger}(\xi_{n})\Omega_{0}\| & = 0,
\end{align}
for all $n\in\N_{0}$ and $\xi_{1},\dots,\xi_{n}\in\cD_{\std}$.
\end{thm}
By Corollary \ref{cor:viraesa} and \eqref{eq:noquantconv}, $\fF^{\alg}_{\ua}(\cD_{\std})$ contains a total set of analytic vectors for  $L_{\pm,k}$, and by Nelson's analytic vector theorem \cite{NelsonAnalyticVectors} the real and imaginary parts of $L_{\pm,k}$ are essentially self-adjoint on $\fF^{\alg}_{\ua}(\cD_{\std})$. This implies strong resolvent convergence of the KS approximants and, therefore \cite{ReedMethodsOfModern1}:
\begin{cor}[Convergence of KS unitaries]
\label{cor:KSunitariesc}
Consider the unitaries,
\begin{align*}
U^{(N)}_{t} & = e^{it dF_{{\color{blue}\pm}}(\tilde{o}^{(N)})}, & U_{t} & = e^{it dF_{{\color{blue}\pm}}(o)},
\end{align*}
for $o^{(N)} = \tilde{r}^{(N)}_{\pm,k}, \tilde{\iota}^{(N)}_{\pm,k}$ and $o = r_{\pm,k}, \iota_{\pm,k}$. Then,
\begin{align*}
\lim_{N\rightarrow\infty}\|(U^{(N)}_{t}-U_{t})\phi\| & = 0,
\end{align*}
for all $\phi\in\fF_{\ua}(\fh^{(\pm)}_{\infty,{\color{blue}\pm}})$ uniformly on compact intervals in $t\in\R$.
\end{cor}
As in the previous section, Lemma \ref{lem:KSconvc} implies the following convergence statement for the Bogoliubov transformations \eqref{eq:KSconvgroup}:
\begin{cor}[Convergence of KS Bogoliubov transformations for $\pi_{{\color{blue}\pm}}$]
\label{cor:KSconvbogc}
In the quasi-free representations $\pi_{{\color{blue}\pm}}$, we have:
\begin{align*}
\lim_{N\rightarrow\infty}\|\tilde{\sigma}^{(N)}_{t}(\pi_{{\color{blue}\pm}}(A)) - \sigma_{t}(\pi_{{\color{blue}\pm}}(A))\| & = 0, & \forall A\in\fA^{(\pm)}_{\infty,{\color{blue}\pm}},
\end{align*}
and uniformly on compact intervals in $t\in\R$.
\end{cor}

\begin{rem}[Moebius group in the Neveu-Schwarz sector]
\label{rem:moeberror}
We note that the Hilbert-Schmidt norms in \eqref{eq:momrgestimates} vanish in the Neveu-Schwarz sector ((${\color{blue}-}$)-boundary condition) for $k=0,\pm\tfrac{\pi}{L}$, i.e.~$\|(\tilde{\ell}^{(N)}_{k,\pm}-\ell_{k,\pm})_{\pm\mp}\|_{2}^{2}=0$ for the approximation of the Virasoro generators corresponding to the Moebius group. Therefore, a simple extension of Lemma \ref{lem:KSconv} is sufficient to prove Theorem \ref{thm:KSconvc} in this case using the wavelet renormalization group.
\end{rem}

\paragraph{Smeared KS approximants.}
\label{sec:KSsmeared}
We combine the wavelet and momentum-cutoff renormalization groups to analyze the convergence of the smeared KS approximants \eqref{eq:latvirasmeared} to the smeared Virasoro generators \eqref{eq:virasmeared}. In this way, it is possible to exploit the convergence of the KS approximants in the non-trivial quasi-free representations $\pi_{{\color{blue}\pm}}$ while, at the same time, preserving localization in real space in the sense of Proposition \ref{prop:antilocal}.

\bigskip

Throughout this paragraph we denote by $R^{N}_{M}$ and $R^{N}_{\infty}$ as well as $\alpha^{N}_{M}$ and $\alpha^{N}_{\infty}$ the momentum-cutoff renormalization group, while $S^{N}_{M}$ and $S^{N}_{\infty}$ denote the wavelet renormalization group for differential loops (see Definition \ref{def:waveletrgcur}).
\begin{lemma}[Convergence of smeared one-particle KS approximants]
\label{lem:KSconvsmeared}
Let $s\in C^{\alpha}(\R)$ be a sufficiently regular, compactly supported orthonormal Daubechies scaling function.  For $M\in\N_{0}$ and $X\in\fl(\Lambda_{M},\C)$, we consider,
\begin{align*}
\tilde{\ell}^{(N)}_{\pm}(S^{M}_{N}(X)) & = \tfrac{1}{2L_{N}}\sum_{k\in\Gamma_{N,{\color{blue}+}}}S^{M}_{N}(\hat{X})_{k}\tilde{\ell}^{(N)}_{\pm,k}, & \ell_{\pm}(S^{M}_{\infty}(X)) & = \tfrac{1}{2L}\sum_{k\in\Gamma_{\infty,{\color{blue}+}}}S^{M}_{\infty}(\hat{X})_{k}\ell_{\pm,k},
\end{align*}
for any $N\geq M$.Then,
\begin{align*}
\lim_{N\rightarrow\infty}\|(\tilde{\ell}^{(N)}_{\pm}(S^{M}_{N}(X)\!)\!-\!\ell_{\pm}(S^{M}_{\infty}(X)\!)\!)_{\pm\pm}\xi\|_{\infty} & = 0, & \lim_{N\rightarrow\infty}\|(\tilde{\ell}^{(N)}_{\pm}(S^{M}_{N}(X)\!)\!-\!\ell_{\pm}(S^{M}_{\infty}(X)\!)\!)\xi\|_{\infty} & = 0,
\end{align*}
for $\xi\in\cD_{\std}$, and,
\begin{align*}
\lim_{N\rightarrow\infty}\|(\tilde{\ell}^{(N)}_{\pm}(S^{M}_{N}(X)\!)\!-\!\ell_{\pm}(S^{M}_{\infty}(X)\!)\!)_{\pm\mp}\|_{2} & = 0.
\end{align*}
\end{lemma}
\begin{proof}
For sufficiently regular $s$, we know that,
\begin{align*}
\ell_{\pm}(S^{M}_{\infty}(X)) & = \tfrac{1}{2L_{M}}\sum_{k\in\Gamma_{\infty,{\color{blue}+}}}\hat{s}(\vep_{M}k)\hat{X}_{\per|k}\ell_{\pm,k},
\end{align*}
is well-defined on $\cD_{\std}$ because of \eqref{eq:decay}. Now, we use,
\begin{align*}
\tilde{\ell}^{(N)}_{\pm}(S^{M}_{N}(X)\!)\!-\!\ell_{\pm}(S^{M}_{\infty}(X)\!) & \!=\!\tilde{\ell}^{(N)}_{\pm}(S^{M}_{N}(X)\!-\!\chi_{\Gamma_{N}}S^{M}_{\infty}(X)\!)\!+\!\tilde{\ell}^{(N)}_{\pm}(\chi_{\Gamma_{N}}S^{M}_{\infty}(X)\!)\!-\!\ell_{\pm}(S^{M}_{\infty}(X)\!),
\end{align*}
to write:
\begin{align*}
\|(\tilde{\ell}^{(N)}_{\pm}(S^{M}_{N}(X)\!)\!-\!\ell_{\pm}(S^{M}_{\infty}(X)\!)\!)_{\pm\pm}\xi\|_{\infty} & \!\leq\!\|\tilde{\ell}^{(N)}_{\pm}(S^{M}_{N}(X)\!-\!\chi_{\Gamma_{N}}S^{M}_{\infty}(X)\!)_{\pm\pm}\xi\|_{\infty} \\
& \hspace{0.5cm} + \|(\tilde{\ell}^{(N)}_{\pm}(\chi_{\Gamma_{N}}S^{M}_{\infty}(X)\!)\!-\!\ell_{\pm}(S^{M}_{\infty}(X)\!)\!)_{\pm\pm}\xi\|_{\infty}, \\[0.25cm]
\|(\tilde{\ell}^{(N)}_{\pm}(S^{M}_{N}(X)\!)\!-\!\ell_{\pm}(S^{M}_{\infty}(X)\!)\!)\xi\|_{\infty} & \!\leq\!\|\tilde{\ell}^{(N)}_{\pm}(S^{M}_{N}(X)\!-\!\chi_{\Gamma_{N}}S^{M}_{\infty}(X)\!)\xi\|_{\infty} \\
& \hspace{0.5cm} + \|(\tilde{\ell}^{(N)}_{\pm}(\chi_{\Gamma_{N}}S^{M}_{\infty}(X)\!)\!-\!\ell_{\pm}(S^{M}_{\infty}(X)\!)\!)\xi\|_{\infty}, \\[0.25cm]
\|(\tilde{\ell}^{(N)}_{\pm}(S^{M}_{N}(X)\!)\!-\!\ell_{\pm}(S^{M}_{\infty}(X)\!)\!)_{\pm\mp}\|_{2} & \!\leq\!\|\tilde{\ell}^{(N)}_{\pm}(S^{M}_{N}(X)\!-\!\chi_{\Gamma_{N}}S^{M}_{\infty}(X)\!)_{\pm\mp}\|_{2} \\
& \hspace{0.5cm} + \|(\tilde{\ell}^{(N)}_{\pm}(\chi_{\Gamma_{N}}S^{M}_{\infty}(X)\!)\!-\!\ell_{\pm}(S^{M}_{\infty}(X)\!)\!)_{\pm\mp}\|_{2}
\end{align*}
Now, let $\xi = R^{K}_{\infty}(\eta)$ for $\eta\in\fh_{K}$. By the Cauchy-Schwarz inequality we have:
\begin{align*}
 & \|\tilde{\ell}^{(N)}_{\pm}(S^{M}_{N}(X)\!-\!\chi_{\Gamma_{N}}S^{M}_{\infty}(X)\!)_{\pm\pm}R^{K}_{\infty}(\eta)\|_{\infty}^{2} \\
 & \leq \|\eta\|_{K}^{2}\tfrac{1}{2L_{K}}\!\!\!\sum_{l\in\Gamma_{K,{\color{blue}\pm}}}\!\|\tfrac{1}{2L}\!\!\sum_{k\in\Gamma_{N,{\color{blue}+}}}\!\!\vep_{M}\Bigg(\prod_{j=1}^{N-M}m_{0}(\vep_{M+j}k)\!-\!\hat{s}(\vep_{M}k)\Bigg)\hat{X}_{\per|k}(\tilde{\ell}^{(N)}_{\pm,k})_{\pm\pm}R^{K}_{\infty}(e_{l})\|_{\infty}^{2} \\
 & \leq \|\eta\|_{K}^{2}\tfrac{1}{2L_{K}}\!\!\!\sum_{l\in\Gamma_{K,{\color{blue}\pm}}}\!\!\!\Bigg(\tfrac{1}{2L}\!\!\sum_{k\in\Gamma_{N,{\color{blue}+}}}\!\!\!\big(1\!+\!|\vep_{M}k|\big)^{2\delta}\Big|\vep_{M}\Bigg(\!\prod_{j=1}^{N-M}m_{0}(\vep_{M+j}k)\!-\!\hat{s}(\vep_{M}k)\!\Bigg)\hat{X}_{\per|k}\Big|^{2}\!\Bigg) \\
 &\hspace{2.4cm} \times\Bigg(\!\tfrac{1}{2L}\!\!\sum_{k\in\Gamma_{N,{\color{blue}+}}}\!\!\!\big(1\!+\!|\vep_{M}k|\big)^{-2\delta}\|(\tilde{\ell}^{(N)}_{\pm,k})_{\pm\pm}R^{K}_{\infty}(e_{l})\|_{\infty}^{2}\!\!\Bigg),
\end{align*}
and,
\begin{align*}
 & \|(\tilde{\ell}^{(N)}_{\pm}(\chi_{\Gamma_{N}}S^{M}_{\infty}(X)\!)\!-\!\ell_{\pm}(S^{M}_{\infty}(X)\!)\!)_{\pm\pm}R^{K}_{\infty}(\eta)\|_{\infty}^{2} \\
 & \leq \|\eta\|_{K}^{2}\tfrac{1}{2L_{K}}\!\!\!\sum_{l\in\Gamma_{K,{\color{blue}\pm}}}\!\|\tfrac{1}{2L}\!\!\sum_{k\in\Gamma_{\infty,{\color{blue}+}}}\!\!\vep_{M}\hat{s}(\vep_{M}k)\hat{X}_{\per|k}(\chi_{\Gamma_{N}}(k)(\tilde{\ell}^{(N)}_{\pm,k})\!-\!(\ell_{\pm,k})\!)_{\pm\pm}R^{K}_{\infty}(e_{l})\|_{\infty}^{2} \\
 & \leq \|\eta\|_{K}^{2}\tfrac{1}{2L_{K}}\!\!\!\sum_{l\in\Gamma_{K,{\color{blue}\pm}}}\!\!\!\Bigg(\tfrac{1}{2L}\!\!\sum_{k\in\Gamma_{\infty,{\color{blue}+}}}\!\!\!\big(1\!+\!|\vep_{M}k|\big)^{2\delta}\Big|\vep_{M}\hat{s}(\vep_{M}k)\hat{X}_{\per|k}\Big|^{2}\!\Bigg) \\
 &\hspace{2.4cm} \times\Bigg(\!\tfrac{1}{2L}\!\!\sum_{k\in\Gamma_{\infty,{\color{blue}+}}}\!\!\!\big(1\!+\!|\vep_{M}k|\big)^{-2\delta}\|(\chi_{\Gamma_{N}}(k)(\tilde{\ell}^{(N)}_{\pm,k})\!-\!(\ell_{\pm,k})\!)_{\pm\pm}R^{K}_{\infty}(e_{l})\|_{\infty}^{2}\!\!\Bigg).
\end{align*}
For sufficiently large $\delta>0$ and $\alpha>0$, this implies the first convergence statement because of \eqref{eq:momrgestimates}, Lemma \ref{lem:convergence}, and the polynomial boundedness of $\|\tilde{\ell}^{(N)}_{\pm, k}R^{K}_{\infty}(e_{l})\|_{\infty}$ and $\|\ell_{\pm,k}R^{K}_{\infty}(e_{l})\|_{\infty}$ in $k\in\Gamma_{\infty,{\color{blue}\pm}}$ due to \eqref{eq:analyticvector}.

The case without the restriction to the diagonal parts $(\!\ \cdot\!\ )_{\pm\pm}$ of the block-diagonal decomposition \eqref{eq:blockmatrix} is analogous. Similarly, we find:
\begin{align*}
\|\tilde{\ell}^{(N)}_{\pm}\!(S^{M}_{N}\!(X)\!-\!\chi_{\Gamma_{N}}S^{M}_{\infty}\!(X)\!)_{\pm\mp}\|_{2}^{2} & \!\leq\!\!\Bigg(\!\!\tfrac{1}{2L}\!\!\!\!\sum_{k\in\Gamma_{N,{\color{blue}+}}}\!\!\!\!\!\big(1\!+\!|\vep_{M}k|\big)^{2\delta}\!\vep_{M}\Big|\!\!\!\!\prod_{j=1}^{\ N-M}\!\!\!\!m_{0}(\vep_{M+j}k)\!-\!\hat{s}(\vep_{M}k)\Big|^{2}\!|\hat{X}_{\per|k}|^{2}\!\Bigg) \\
 &\hspace{1cm} \times\Bigg(\!\tfrac{1}{2L}\!\!\!\!\!\sum_{k\in\Gamma_{N,{\color{blue}\pm}}}\!\!\!\big(1\!+\!|\vep_{M}k|\big)^{-2\delta}\|(\tilde{\ell}^{(N)}_{\pm,k})_{\pm\mp}\|_{2}^{2}\!\!\Bigg),
\end{align*}
and,
\begin{align*}
\|(\tilde{\ell}^{(N)}_{\pm}(\chi_{\Gamma_{N}}S^{M}_{\infty}(X)\!)\!-\!\ell_{\pm}(S^{M}_{\infty}(X)\!)\!)_{\pm\mp}\|_{2}^{2} & \!\leq\!\!\Bigg(\tfrac{1}{2L}\!\!\sum_{k\in\Gamma_{\infty,{\color{blue}+}}}\!\!\!\big(1\!+\!|\vep_{M}k|\big)^{2\delta}\Big|\vep_{M}\hat{s}(\vep_{M}k)\hat{X}_{\per|k}\Big|^{2}\!\Bigg) \\
 &\hspace{0.25cm}\times\!\!\Bigg(\!\tfrac{1}{2L}\!\!\!\!\!\sum_{k\in\Gamma_{\infty,{\color{blue}+}}}\!\!\!\!\!\big(1\!+\!|\vep_{M}k|\big)^{-2\delta}\|(\chi_{\Gamma_{N}}\!(k)\tilde{\ell}^{(N)}_{\pm,k}\!-\!\ell_{\pm,k})_{\pm\mp}\|_{2}^{2}\!\!\Bigg).
\end{align*}
Again for sufficiently large $\delta>0$ and $\alpha>0$, the convergence of the second statement follows using \eqref{eq:momrgestimates}, the observation that $\|(\tilde{\ell}^{(N)}_{\pm,k})_{\pm\mp}\|_{2}$ and $\|(\ell_{\pm,k})_{\pm\mp}\|_{2}$ are polynomially bounded in $k\in\Gamma_{\infty,{\color{blue}\pm}}$:
\begin{align*}
\|(\tilde{\ell}^{(N)}_{\pm,k})_{\pm\mp}\|_{2}^{2} & = (\tfrac{L}{\pi})^{2}\!\!\!\!\sum_{n\in\Gamma_{N,{\color{blue}\pm}}}\!\!\!\!\theta(\pm n\!-\!k)\theta(\mp n)\cos(\tfrac{1}{4}\vep_{N}k)^{4}\Big(\tfrac{\sin(\vep_{N}(n\mp\frac{k}{2}))}{\vep_{N}}\Big)^{2}\chi_{\Gamma_{N}(n\!\mp\!k)}, \\
\|(\ell_{\pm,k})_{\pm\mp}\|_{2}^{2} & = (\tfrac{L}{\pi})^{2}\!\!\!\!\sum_{n\in\Gamma_{\infty,{\color{blue}\pm}}}\!\!\!\!\theta(\pm n\!-\!k)\theta(\mp n)(n\!\mp\!\tfrac{k}{2})^{2}.
\end{align*}
\end{proof}
As before, we use the preceding lemma to deduce the following result on the convergence of the KS approximants (cf.~\eqref{eq:noquantconv}).
\begin{thm}[Convergence of smeared KS approximants for $\pi_{{\color{blue}\pm}}$]
\label{thm:KSconvsmeared}
Let $s\in C^{\alpha}(\R)$ be a sufficiently regular, compactly supported orthonormal Daubechies scaling function, and let $X\in\fl(\Lambda_{M},\C)$ for some $M\in\N_{0}$. Then, the smeared KS approximants, $L^{(N)}_{\pm}(S^{M}_{N}(X))$, converge strongly to the continuum Virasoro generators, $L_{\pm}(S^{M}_{\infty}(X))$, on the dense domain $\fF^{\alg}_{\ua}(\cD_{\std})\subset\fF_{\ua}(\fh^{(\pm)}_{\infty,{\color{blue}\pm}})$ spanned by anti-symmetric Fock vectors with finitely many one-particle excitations in $\cD_{\std}$ (finite $\cD_{\std}$-particle number):
\begin{align}
\label{eq:latviraconvsmeared}
\lim_{N\rightarrow\infty}\|(:\!\pi_{{\color{blue}\pm}}(\alpha^{N}_{\infty}(L^{(N)}_{\pm}(S^{M}_{N}(X)))\!)\!: - L_{\pm}(S^{M}_{\infty}(X)))a^{\dagger}(\xi_{1})\dots a^{\dagger}(\xi_{n})\Omega_{0}\| & = 0,
\end{align}
for all $n\in\N_{0}$ and $\xi_{1},\dots,\xi_{n}\in\cD_{\std}$.
\end{thm}
By the results in \cite{CarpiOnTheUniqueness}, $\fF^{\alg}_{\ua}(\cD_{\std})$ is a domain of essential self-adjointness of $L_{\pm}(X)$ for real smearing functions $X\in\fl^{\frac{3}{2}}(S^{1},\R)$ (the representations $\pi_{{\color{blue}\pm}}$ have positive energy and $c>0$, see also \cite{GoodmanProjectiveUnitaryPositive}). Thus, the preceding theorem implies strong resolvent convergence of the smeared KS approximants and, therefore \cite{ReedMethodsOfModern1}:
\begin{cor}[Convergence of smeared KS unitaries]
\label{cor:KSconvsmeared}
For $X\in\fl(\Lambda_{M},\R)$, consider the unitaries,
\begin{align*}
U^{(N)}_{t}(S^{M}_{N}(X)) & = e^{it:\pi_{{\color{blue}\pm}}(\alpha^{N}_{\infty}(L^{(N)}_{\pm}(S^{M}_{N}(X)))):}, & U_{t}(X) & = e^{itL_{\pm}(S^{M}_{\infty}(X))},
\end{align*}
Then,
\begin{align*}
\lim_{N\rightarrow\infty}\|(U^{(N)}_{t}(X)-U_{t}(X))\phi\| & = 0,
\end{align*}
for all $\phi\in\fF_{\ua}(\fh^{(\pm)}_{\infty,{\color{blue}\pm}})$ uniformly on compact intervals in $t\in\R$.
\end{cor}
\begin{rem}
\label{rem:KSconvsmeared}
We can lift the restriction to sequences $\{S^{M}_{N}(X)\}_{N\in\N_{0}}$ with $X\in\fl(\Lambda_{M},\C)$ (called basic sequences in \cite{DuffieldMeanFieldDynamical}) by:
\begin{align*}
\tilde{\ell}^{(N)}_{\pm}(X_{N})\!-\!\ell_{\pm}(X) & \!=\!\tilde{\ell}^{(N)}_{\pm}(X_{N}\!-\!S^{M}_{N}(X_{M})\!)\!+\!\tilde{\ell}^{(N)}_{\pm}(S^{M}_{N}(X_{M})\!)\!-\!\ell_{\pm}(S^{M}_{\infty}(X)\!)\!+\!\ell_{\pm}(S^{M}_{\infty}(X)\!-\!X),
\end{align*}
for $X_{N}\in\fl(\Lambda_{N},\C)$ and $X\in\fl^{\alpha}(S^{1}_{L},\C)$ ($\alpha>0$). Then, assuming we have a $S$-convergent sequence of smearing functions $X_{N}\stackrel{N\rightarrow\infty}{\rightarrow} X$ in the sense of Definition \ref{def:convseq} for a suitable sequence of Sobolev-type semi-norm $\{p_{N,\delta}\}_{N\in\N_{0}}$,
\begin{align*}
\lim_{M\rightarrow\infty}\limsup_{N\rightarrow\infty}p_{N,\delta}(X_{N}-S^{M}_{N}(X_{M})\!) & = 0,
\end{align*}
we conclude the convergence of the KS approximants as in \eqref{eq:latviraconvsmeared}. In principle this allows us to reconstruct the smeared Virasoro generators for smooth smearing functions $X\in\fl^{\infty}(S^{1}_{L},\C)$.
\end{rem}
In analogy with the Corollaries \ref{cor:KSconvbog} \& \ref{cor:KSconvbogc}, Lemma \ref{lem:KSconvsmeared} implies the convergence of the Bogoliubov transformation with smeared generators, see \eqref{eq:latviradyn}:
\begin{cor}[Convergence of smeared KS Bogoliubov transformations]
\label{cor:KSconvbogsmeared}
Let $s\in C^{\alpha}(\R)$ be a sufficiently regular, compactly supported orthonormal Daubechies scaling function. In the quasi-free representations $\pi_{{\color{blue}\pm}}$, we have:
\begin{align*}
\lim_{N\rightarrow\infty}\|\tilde{\sigma}^{(N)}_{tS^{M}_{N}(X)}(\pi_{{\color{blue}\pm}}(A)) - \sigma_{tS^{M}_{\infty}(X)}(\pi_{{\color{blue}\pm}}(A))\| & = 0, & \forall A\in\fA^{(\pm)}_{\infty,{\color{blue}\pm}},
\end{align*}
uniformly on compact intervals in $t\in\R$, for any $M\in\N_{0}$ and $X\in\fl(\Lambda_{M},\R)$.
\end{cor}

\paragraph{Error estimates.}
In this paragraph, we discuss some basic estimates of the approximation errors occurring in the proofs of Theorems \ref{thm:KSconv} \& \ref{thm:KSconvc}. Explicit control of the approximation errors is particularly interesting in view of the potential simulation of conformal field theories on quantum computers using a lattice discretization.

\bigskip

Let us begin with a discussion of the approximation error in Lemma \ref{lem:KSconv}. The structure of the proof, and specifically considering,
\begin{align}
\label{eq:uniformestimate}
& \sum_{m\in\Gamma_{\infty,{\color{blue}\pm}}}|f^{(N)}_{k}(M,m)-f_{k}(M,m)|^{2} \\ \nonumber
& = \sum_{m\in\Gamma_{\infty,{\color{blue}\pm}}}\left|\hat{s}(\vep_{M}(m\!+\!k))(m\!+\!\tfrac{k}{2})\right|^{2}\left|\cos(\tfrac{1}{4}\vep_{N}k)^{2}\tfrac{\sin(\vep_{N}(m+\frac{k}{2}))}{\vep_{N}(m+\frac{k}{2})}\tfrac{\hat{s}(\vep_{N}m)}{\hat{s}(\vep_{N}(m+k))}-1\right|^{2},
\end{align}
suggests a simple strategy to obtain a specific error bound by estimating against the decay of $\hat{s}$:
\begin{align}
\label{eq:sobolevtrade}
& \sum_{m\in\Gamma_{\infty,{\color{blue}\pm}}}\left|\hat{s}(\vep_{M}(m\!+\!k))(m\!+\!\tfrac{k}{2})\right|^{2}\left| \cos(\tfrac{1}{4}\vep_{N}k)^{2}\tfrac{\sin(\vep_{N}(m+\frac{k}{2}))}{\vep_{N}(m+\frac{k}{2})}\tfrac{\hat{s}(\vep_{N}m)}{\hat{s}(\vep_{N}(m+k))}-1\right|^{2} \\ \nonumber
& \leq\!\!\left(\!\tfrac{2\pi^{2}}{L}\!\!\sup_{m\in\Gamma_{\infty,{\color{blue}\pm}}}\!\!(1\!+\!\tfrac{L}{\pi}|m\!+\!\tfrac{k}{2}|)^{-2\delta}\!\left| \cos(\tfrac{1}{4}\vep_{N}k)^{2}\tfrac{\sin(\vep_{N}(m+\frac{k}{2}))}{\vep_{N}(m+\frac{k}{2})}\tfrac{\hat{s}(\vep_{N}m)}{\hat{s}(\vep_{N}(m+k))}\!-\!1\right|^{2}\right)\|\hat{s}(\vep_{M}(\!\ .\!\ \!+\!\tfrac{k}{2}))\|_{h^{1+\delta}}^{2},
\end{align}
for some $\delta>0$ and $\|\hat{s}(\vep_{M}(\!\ .\!\ +\tfrac{k}{2}))\|_{h^{1+\delta}}^{2}\!=\!\tfrac{1}{2L}\sum_{m\in\Gamma_{\infty,{\color{blue}\pm}}}(1\!+\!\tfrac{L}{\pi}|m|)^{2(1+\delta)}|\hat{s}(\vep_{M}(m\!+\!\tfrac{k}{2}))|^{2}$. The regularizing factor $(1\!+\!\tfrac{L}{\pi}|m\!+\!\tfrac{k}{2}|)^{-2\delta}$ has the effect that,
\begin{align}
\label{eq:errordecay}
\textup{Error}^{2}(\delta,L, k, N) &\!:=\!\tfrac{2\pi^2}{L}\!\!\sup_{m\in\Gamma_{\infty,{\color{blue}\pm}}}\!\!\!\textup{error}^{2}(\delta,L, k, N,m) \stackrel{N\rightarrow\infty}\longrightarrow 0 \\ \nonumber
\textup{error}^{2}(\delta,L, k, N,m) &\!:=\!(1+\tfrac{L}{\pi}|m+\tfrac{k}{2}|)^{-2\delta}\left| \cos(\tfrac{1}{4}\vep_{N}k)^{2}\tfrac{\sin(\vep_{N}(m+\frac{k}{2}))}{\vep_{N}(m+\frac{k}{2})}\tfrac{\hat{s}(\vep_{N}m)}{\hat{s}(\vep_{N}(m+k))}\!-\!1\right|^{2},
\end{align}
with a rate determined by $\delta = \delta(L,k)$. To obtain a precise numerical bound we also need to estimate $\|\hat{s}(\vep_{M}(\!\ .\!\ +\tfrac{k}{2}))\|_{h^{1+\delta}}^{2}$, which will decrease with increasing regularity of $s$.

\medskip

According to the proof of Theorem \ref{thm:KSconv}, we can estimate the Fock-space norm in \eqref{eq:latviraconv} using \eqref{eq:uniformestimate}:
\begin{align}
\label{eq:uniformestimateFock}
\|(\alpha^{N}_{\infty}(L^{(N)}_{k,\pm})\!-\!L_{k,\pm})a^{\dagger}(\xi_{1})...a^{\dagger}(\xi_{n})\Omega_{0}\| & \!\!\leq\!\!\tfrac{nL}{\pi}\textup{Error}(\delta,L, k, N)\|\hat{s}(\vep_{M}(\!\ .\!\ \!\!+\!\tfrac{k}{2}))\|_{h^{1+\delta}}\!\!\prod_{p=1}^{n}\!\|\eta_{p}\|_{M},
\end{align}
where $\xi_{p} = R^{M}_{\infty}(\eta_{p})$ for $p=1,\dots,n$ and $M<N$\footnote{Note that the states $a^{\dagger}(\xi_{1})\dots a^{\dagger}(\xi_{n})\Omega_{0}$ have the norm $\|a^{\dagger}(\xi_{1})\dots a^{\dagger}(\xi_{n})\Omega_{0}\| = \prod_{p=1}^{n}\|\xi_{p}\|_{\infty}$, for $\xi_{1},\dots,\xi_{n}$ orthogonal.}.

\medskip

The supremum in \eqref{eq:errordecay} can be estimated for $\delta\in[0,2]$ as follows:
\begin{align}
\label{eq:errordecayexp}
& \textup{Error}^{2}(\delta,L, k, N) \\ \nonumber
& = \tfrac{2\pi^2}{L}(\tfrac{\pi}{L_{N}})^{2\delta}\!\!\sup_{m\in\Gamma_{\infty}}\!\!(\tfrac{\pi}{L_{N}}\!+\!\vep_{N}|m+\tfrac{k}{2}|)^{-2\delta}\!\left| \cos(\tfrac{1}{4}\vep_{N}k)^{2}\tfrac{\sin(\vep_{N}(m+\frac{k}{2}))}{\vep_{N}(m+\frac{k}{2})}\tfrac{\hat{s}(\vep_{N}m)}{\hat{s}(\vep_{N}(m+k))}-1\right|^{2} \\ \nonumber
& \leq \tfrac{2\pi^2}{L}(\tfrac{\pi}{L_{N}})^{2\delta}\sup_{x\in\R}(\tfrac{\pi}{L_{N}}\!+\!|x+\vep_{N}\tfrac{k}{2}|)^{-2\delta}\!\left| \cos(\tfrac{1}{4}\vep_{N}k)^{2}\tfrac{\sin(x+\vep_{N}\frac{k}{2})}{x+\vep_{N}\frac{k}{2}}\tfrac{\hat{s}(x)}{\hat{s}(x+\vep_{N}k)}-1\right|^{2}.
\end{align}
Thus, we find for $k=0$:
\begin{align}
\label{eq:errordecayexp0}
\textup{Error}^{2}(\delta,L, 0, N) & \leq \tfrac{2\pi^2}{L}(\tfrac{\pi}{L_{N}})^{2\delta}\sup_{x\in\R}(\tfrac{\pi}{L_{N}}\!+\!|x|)^{-2\delta}\!\left|\tfrac{\sin(x)}{x}-1\right|^{2} \\ \nonumber
 & \leq \tfrac{2\pi^2}{L}(\tfrac{\pi}{L_{N}})^{2\delta}\sup_{x\in\R}|x|^{-2\delta}\!\left|\tfrac{\sin(x)}{x}-1\right|^{2},
\end{align}
and $\sup_{x\in\R}|x|^{-2\delta}\!\left|\tfrac{\sin(x)}{x}-1\right|^{2}$ is a finite constant of order unity for $\delta\in[0,2]$ since $|\tfrac{\sin(x)}{x}-1|^{2}\approx\tfrac{x^{4}}{36}+\mathcal{O}(x^{6})$. As long as $N$ is sufficiently large, we expect essentially the same behavior for $k\neq 0$ because of the uniform continuity of $\hat{s}$ in \eqref{eq:errordecayexp}. This suggests:
\begin{align}
\label{eq:errorscaling}
\textup{Error}^{2}(\delta,L, k, N) & \lesssim C_{\delta,L,k}\!\ 2^{-2\delta N},
\end{align}
with an explicitly computable proportionality factor $C_{\delta,L,k}$.

\begin{rem}[CFT simulation]
\label{rem:stateapprox}
Concerning the simulation of the a free-fermion CFT on a quantum computer, the scale parameter $N$ reflects the number of qubits used in the simulation because it determines the size of the real-space lattice $\Lambda_{N}$ which forms the basis for the Hilbert space on which the KS approximant $L^{(N)}_{\pm,k}$ acts. The second scale parameter $M$ reflects the resolution at which the Fock states $a^{\dagger}(\xi_{1})\dots a^{\dagger}(\xi_{n})\Omega_{0}$ are resolved. It should, therefore, be strictly less than $N$ reflecting at least a doubling of the sampling frequency in accordance with the Nyquist-Shannon theorem. At the same time, if we fix the length scale $L$ of the spatial circle $S^{1}_{L}$ on which the CFT is supposed to live, we are forced to choose $M$ in accordance with the regularity parameter $K\geq 2$ of the Daubechies scaling function $s=\!\ _{K}s$ since $\supp(\!\ _{K}s)=[0,2K-1]$. This suggest a lower bound on $M$ in terms of $K$ according to $\vep_{M}(2K-1)\sim L$ to ensure that $s^{(\vep_{M})}$ is well-localized in $S^{1}_{L}$.

The choice, $\xi_{p}=R^{M}_{\infty}(\eta_{p})$, $p=1,\dots,n$, for the one-particle states means that an arbitrary $n$-particle state $a^{\dagger}(\xi_{1})\dots a^{\dagger}(\xi_{n})\Omega_{0}$ is approximated by projecting its one-particle amplitudes onto the subspace generated by the wavelet basis determined by the scaling function $s$ up to a resolution scale $\vep_{M} = 2^{-M}\vep_{0}$, i.e.:
\begin{align*}
\eta_{p} & = R^{M}_{\infty}{}^{*}\xi_{p} = \sum_{x\in\Lambda_{M}}\langle s^{(\vep_{M})}(\!\ .\!\ -x),\xi_{p}\rangle s^{(\vep_{M})}(\!\ .\!\ -x) \\ 
& = \sum_{x\in\Lambda_{0}}\langle s(\!\ .\!\ -x),\xi_{p}\rangle s(\!\ .\!\ -x) + \sum_{j=0}^{M}\sum_{x\in\Lambda_{j}}\langle \psi^{(\vep_{j})}(\!\ .\!\ -x),\xi_{p}\rangle \psi^{(\vep_{j})}(\!\ .\!\ -x),
\end{align*}
where $\psi$ is the wavelet associated with $s$.

The inner products of the states $\pi_{{\color{blue}\pm}}(a^{\dagger}(\xi_{1}))...\pi_{{\color{blue}\pm}}(a^{\dagger}(\xi_{n}))\Omega_{0}$ are accessible according to Lemma \ref{lem:stateconv} in terms of the lattice correlation functions of $\omega^{(N)}_{0,{\color{blue}\pm}}$:
\begin{align*}
\langle \pi_{{\color{blue}\pm}}(a^{\dagger}(R^{M}_{\infty}(\eta)))\Omega_{0},\pi_{{\color{blue}\pm}}(a^{\dagger}(R^{M}_{\infty}(\eta')))\Omega_{0}\rangle & = \lim_{N\rightarrow\infty}\langle a^{\dagger}_{N}(R^{M}_{N}(\eta))\omega^{(N)}_{0,{\color{blue}\pm}},a^{\dagger}_{N}(R^{M}_{N}(\eta'))\omega^{(N)}_{0,{\color{blue}\pm}}\rangle,
\end{align*}
where we explicitly indicated that the creation operators $a^{\dagger}_{N}$ are to be realized by the quantum simulation at scale $N$.
\end{rem}

Let us continue with a few remarks on the estimation of errors occurring in Lemma \ref{lem:KSconvc} using the momentum-cutoff renormalization group.	
\begin{rem}[Error of time-evolution]
\label{rem:timeerror}
The error bound on the one-particle generators directly translates to an error bound on the unitary implementer of time evolution ($k=0$) by the fundamental theorem of calculus since $[\tilde{\ell}^{(N)}_{\pm,0},\exp(it\ell_{\pm,0})]=0$ implies:
\begin{align}
\label{eq:momrgestimatetime}
\|(e^{it\tilde{\ell}^{(N)}_{\pm,0}}-e^{it\ell_{\pm,0}})R^{M}_{\infty}(\eta_{p})\|_{\infty} & \leq t \|(\tilde{\ell}^{(N)}_{\pm,0}-\ell_{\pm,0})R^{M}_{\infty}(\eta_{p})\|_{\infty}.
\end{align}
Thus, we can explicitly estimate the error for the analogue of \eqref{eq:KSmultiestimate} as well as the approximation of time-dependent correlation functions as in Theorem \ref{thm:corapprox} below. For the same reason, we can estimate the approximation of correlation functions involving $U_{t}=\exp(itL_{\pm,0})$ by those involving $U^{(N)}_{t}=\exp(itdF_{{\color{blue}-}}(\tilde{\ell}^{(N)}_{\pm,0}))$, cf.~Corollary \ref{cor:KSunitariesc}. By a similar reasoning that leads to \eqref{eq:errorscaling}, we have for $\delta\in[0,2]$:
\begin{align}
\label{eq:momrgtimeerrorscaling}
\|((\tilde{\ell}^{(N)}_{\pm,0}\!-\!\ell_{\pm,0})^{(t)}_{\pm\pm}R^{M}_{\infty}(\eta_{p})\|_{\infty}^{2} & \!\leq\!\vep_{N}^{2\delta}\underbrace{(\sup_{x\in\R}\tfrac{|\sinc(x)\!-\!1|^{2}}{|x|^{2\delta}})}_{=\frac{1}{36}\textup{ for }\delta=2}\|\eta\|_{M}^{2}(\tfrac{L}{\pi})^{2}\!\!\!\!\sum_{n\in\Gamma_{M,{\color{blue}\pm}}}\!\!\!\!\theta(\pm n)|n|^{2(1+\delta)},
\end{align}
which shows that the asymptotic approximation error for the time-evolution behaves as $\sim 2^{-2N}$ for fixed $M$ ($\delta=2$).
\end{rem}

\begin{rem}[Error estimates for $k\neq0$]
\label{rem:conformalerror}
For $k\neq 0$, the error bound on the generators can be used to obtain an error bound (asymptotic in the time $t$) for the associated unitary groups acting on analytic vectors $\{e_{m}\}_{m\in\N_{0}}$ which are in the image of $R^{M}_{\infty}$. To this end, we observe that:
\begin{align*}
\|(e^{it\tilde{o}^{(N)}}-e^{ito})\xi\|_{\infty} & \leq t \|(\tilde{o}^{(N)}-o)\xi\|_{\infty}+\sum_{j=2}^{\infty}\tfrac{t^{j}}{j!}\|((\tilde{o}^{(N)})^{j}-o^{j})\xi\|_{\infty},
\end{align*}
which is convergent for some $t>0$ whenever $\xi$ is analytic for $\tilde{o}^{(N)}$, $o$. Thus, we can, for example, consider $\tilde{o}^{(N)}=\tilde{r}^{(N)}_{\pm,k}$ and $o = r_{\pm,k}$ and exploit \eqref{eq:analyticvector} to deduce for $k\neq0$:
\begin{align*}
\|((\tilde{r}^{(N)}_{\pm,k})^{j}-r_{\pm,k}^{j})e_{m}\|_{\infty} & \leq 4L(\tfrac{L}{2\pi})^{j}\binom{2j}{j}^{\frac{1}{2}}|k|^{j}\tfrac{\Gamma(|\frac{m}{k}|+\frac{1}{2}+j)}{\Gamma(|\frac{m}{k}|+\frac{1}{2})}.
\end{align*}
This immediately implies (for sufficiently small $t>0$ if $k\neq 0$):
\begin{align}
\label{eq:conformalerror}
\|(e^{it\tilde{r}^{(N)}_{\pm,k}}-e^{itr_{\pm,k}})e_{m}\|_{\infty} & \leq t \|(\tilde{r}^{(N)}_{\pm,k}-r_{\pm,k})e_{m}\|_{\infty} + 4L\frac{(\tfrac{L}{\pi}(|m|+\frac{1}{2}|k|)t)^{2}}{1-\tfrac{L}{\pi}(|m|+\frac{1}{2}|k|)t} ,
\end{align}
because $\binom{2j}{j}^{\frac{1}{2}}|k|^{j}\tfrac{\Gamma(|\frac{m}{k}|+\frac{1}{2}+j)}{\Gamma(|\frac{m}{k}|+\frac{1}{2})} \leq (2(|m|+\tfrac{1}{2}|k|))^{j}j!$. An analogous bound holds for the approximation of the unitary group associated with $\iota_{\pm,k}$.

\medskip

Thus, we find that for small $t>0$ (independent of $N$) the approximation error (in the limit $N\rightarrow\infty$) for the unitary groups can be explicitly estimated in terms of the error given for the generators.

\medskip

Being less explicit about the one-particle vectors used in \eqref{eq:momrgestimates} (or \eqref{eq:KSmultiestimate} in general), we can always get the following bound for the approximation error, valid for all $t>0$, when restricting to a spectral subspace $\fh_{\infty,{\color{blue}\pm}}^{\leq M} = P_{[M,-M]}\fh^{(\pm)}_{\infty,{\color{blue}\pm}}$ of $o$ for some finite cut-off $M<\infty$, 
\begin{align*}
\|(e^{it\tilde{o}^{(N)}}\!-\!e^{ito})\xi\|_{\infty} & \!\leq\! t \|(\tilde{o}^{(N)}\!-\!o)\xi\|_{\infty}\!+\!(1\!+\!a)(e^{M t}\!-\!(1\!+\!M t))\|\xi\|_{\infty}\!+\!b(e^{t}\!-\!(1\!+\!t))\|\xi\|_{\infty},
\end{align*}
with $\xi\in\fh_{\infty,{\color{blue}\pm}}^{\leq M}$, and assuming the powers of the approximants $(\tilde{o}^{(N)})^{j}$ are constructed as uniformly $o^{j}$-bounded operators, i.e.~$\|(\tilde{o}^{(N)})^{j}\xi\|_{\infty}\leq a\|o^{j}\xi\|_{\infty}+b\|\xi\|_{\infty}$.
\end{rem}

\subsubsection{Koo-Saleur approximants of the equivalent XY model}
\label{sec:latvirastag}
For completeness and in view of our companion article \cite{OsborneCFTsim}, we also provide the expressions for the KS approximants in terms of the single-component fermion $b$ using \eqref{eq:stagfermionfourier}:
\begin{align}
\label{eq:latvirastag}
L_k^{(N)} & = \tfrac{1}{8\pi}\sum_{l,l’\in\Gamma_{N+1,{\color{blue}\pm}}} \!\!\! \begin{pmatrix} \hat{b}_{l’} \\ \hat{b}^{\dag}_{-l’} \end{pmatrix}^{\dag} e^{-\frac{i}{4}\varepsilon_{N}k} \delta_{k,l’-l\!\!\!\!\mod\frac{2\pi}{\varepsilon_{N}}}  \ell^{(N)}_{k}(l’,l) \begin{pmatrix} \hat{b}_{l} \\ \hat{b}^{\dag}_{-l} \end{pmatrix}, \\ \nonumber
\ell^{(N)}_{k}(l’,l) & = \tfrac{1}{2}\begin{pmatrix} \hspace{-2cm} -e^{\frac{i}{4}\varepsilon_{N}k}\sin(\varepsilon_{N}(l+\tfrac{k}{2})) & \hspace{-2cm} -i(e^{\frac{i}{4}\varepsilon_{N}k}\sin(\varepsilon_{N}\tfrac{l}{2})+e^{-\frac{i}{4}\varepsilon_{N}k}\sin(\varepsilon_{N}\tfrac{l+k}{2})) \\[0.1cm] i(e^{\frac{i}{4}\varepsilon_{N}k}\sin(\varepsilon_{N}\tfrac{l’-k}{2})+e^{-\frac{i}{4}\varepsilon_{N}k}\sin(\varepsilon_{N}\tfrac{l’}{2})) & -e^{-\frac{i}{4}\varepsilon_{N}k}\sin(\varepsilon_{N}\tfrac{l+l’}{2}) \end{pmatrix}, \\ \nonumber
\overline{L}_k^{(N)} & = \tfrac{1}{8\pi}\sum_{l,l’\in\Gamma_{N+1,{\color{blue}\pm}}} \!\!\! \begin{pmatrix} \hat{b}_{l’} \\ \hat{b}^{\dag}_{-l’} \end{pmatrix}^{\dag} e^{\frac{i}{4}\varepsilon_{N}k}\delta_{-k,l’-l\!\!\!\!\mod\frac{2\pi}{\varepsilon_{N}}} \overline{\ell}^{(N)}_{k}(l’,l) \begin{pmatrix} \hat{b}_{l} \\ \hat{b}^{\dag}_{-l} \end{pmatrix}, \\ \nonumber
\overline{\ell}^{(N)}_{k}(l’,l) & = \tfrac{1}{2}\begin{pmatrix} \hspace{-2cm} e^{-\frac{i}{4}\varepsilon_{N}k}\sin(\varepsilon_{N}(l-\tfrac{k}{2})) & \hspace{-2cm} -i(e^{-\frac{i}{4}\varepsilon_{N}k}\sin(\varepsilon_{N}\tfrac{l}{2})+e^{\frac{i}{4}\varepsilon_{N}k}\sin(\varepsilon_{N}\tfrac{l-k}{2})) \\[0.1cm] i(e^{-\frac{i}{4}\varepsilon_{N}k}\sin(\varepsilon_{N}\tfrac{l’+k}{2})+e^{\frac{i}{4}\varepsilon_{N}k}\sin(\varepsilon_{N}\tfrac{l’}{2})) & e^{\frac{i}{4}\varepsilon_{N}k}\sin(\varepsilon_{N}\tfrac{l+l’}{2}) \end{pmatrix}.
\end{align}
Alternatively, we may derive the KS approximants directly in terms of the single-component fermion on the doubled lattice $\Lambda_{N+1}$, which results in:
\begin{align}
\label{eq:latvirastag1c}
L_k^{(N)} & = \tfrac{1}{4\pi}\sum_{l,l’\in\Gamma_{N+1}} \!\!\! \begin{pmatrix} \hat{b}_{l’} \\ \hat{b}^{\dag}_{-l’} \end{pmatrix}^{\dag} \ell^{(N)}_{k}(l’,l) \begin{pmatrix} \hat{b}_{l} \\ \hat{b}^{\dag}_{-l} \end{pmatrix}, \\ \nonumber
\ell^{(N)}_{k}(l’,l) & = \delta_{k,l’-l} \begin{pmatrix} -\sin(\varepsilon_{N+1}(l+l’)) & e^{-i\varepsilon_{N+1}(l+\frac{k}{2})}\cos(\tfrac{1}{2}\vep_{N+1}k) \\ e^{i\varepsilon_{N+1}(l’-\frac{k}{2})}\cos(\tfrac{1}{2}\vep_{N+1}k) & -\sin(\varepsilon_{N+1}(l+l’)) \end{pmatrix}, \\ \nonumber
\overline{L}_k^{(N)} & = \tfrac{1}{4\pi}\sum_{l,l’\in\Gamma_{N+1}} \!\!\! \begin{pmatrix} \hat{b}_{l’} \\ \hat{b}^{\dag}_{-l’} \end{pmatrix}^{\dag} \overline{\ell}^{(N)}_{k}(l’,l) \begin{pmatrix} \hat{b}_{l} \\ \hat{b}^{\dag}_{-l} \end{pmatrix}, \\ \nonumber
\overline{\ell}^{(N)}_{k}(l’,l) & = \delta_{-k,l’-l} \begin{pmatrix} \sin(\varepsilon_{N+1}(l+l’)) & e^{-i\varepsilon_{N+1}(l-\frac{k}{2})}\cos(\tfrac{1}{2}\vep_{N+1}k) \\ e^{i\varepsilon_{N+1}(l’+\frac{k}{2})}\cos(\tfrac{1}{2}\vep_{N+1}k) & \sin(\varepsilon_{N+1}(l+l’)) \end{pmatrix}.
\end{align}
But, it should be noted that these alternative expressions are not directly comparable to those in \eqref{eq:latvirastag} because they arise through the Fourier transform on $\Lambda_{N+1}$, not $\Lambda_{N}$, and, therefore, refer to the (symmetrized) Hamiltonian density 
\begin{align*}
\tilde{h}^{(N)}_{x} & = \tfrac{1}{2}\vep_{N+1}^{-2}(b_{x}b_{x+\vep_{N+1}}-b^{\dag}_{x}b^{\dag}_{x+\vep_{N+1}}+b_{x-\vep_{N+1}}b_{x}-b^{\dag}_{x-\vep_{N+1}}b^{\dag}_{x}),
\end{align*}
of the single-component fermion which is related to the Hamiltonian density of the two-component fermion by the half-shift $x\mapsto x+\vep_{N+1}$:
\begin{align}
\label{eq:hdstag1c}
h^{(N)}_{x} & = \tilde{h}^{(N)}_{x} + \tilde{h}^{(N)}_{x+\vep_{N+1}}, & x & \in\Lambda_{N}.
\end{align}

\section{Approximation of Wess-Zumino-Witten currents}
\label{sec:WZWcur}
We use the same approach that shows the convergence of the KS approximants to the Virasoro generators in the scaling limit to approximate the currents of some Wess-Zumino-Witten models. Similar to the Virasoro generators and the KS approximants the currents of these WZW models can be obtained as normal-ordered second quantizations of certain one-particle operators \cite{DiFrancescoCFTBook, EvansQuantumSymmetriesOn, WassermannOperatorAlgebrasAnd}, i.e.~normal-ordered fermion bilinears.

\bigskip

\paragraph{The $U(1)$-current.}
For simplicity, we describe the adaptation of the procedure for the (chiral) $U(1)$-current,
\begin{align}
\label{eq:U1cur}
J_{{\color{blue}-}}(x) & = :\pi_{{\color{blue}-}}(\psi_{+|x}^{\dag})\pi_{{\color{blue}-}|x}(\psi_{\pm}):, & x & \in S^{1}_{L},
\end{align}
in the Neveu-Schwarz sector of $\fA^{(+)}_{\infty,{\color{blue}-}}$ first. $J_{{\color{blue}-}}(x)$ is a local bosonic field on $\fF_{\ua}(\fh^{(+)}_{\infty,{\color{blue}-}})$ with the following commutation relations \cite{RehrenCFTLectures}:
\begin{align}
\label{eq:U1curcom}
[J_{{\color{blue}-}}(x),J_{{\color{blue}-}}(y)] = \tfrac{i}{2\pi}\delta'_{0}(x-y)\;,
\end{align}
which reflect the presence of a non-trivial central charge $c=1$. At finite scales, the direct analogue of $J_{{\color{blue}-}}(x)$ is given by the \emph{lattice current}:
\begin{align}
\label{eq:U1curfinite}
J^{(N)}(x) & = \psi_{+|x}^{\dag}\psi_{+|x}, & x & \in\Lambda_{N}, 
\end{align}
as an element of $\fA^{(+)}_{N,{\color{blue}-}}$. As for the Virasoro generators, $J_{{\color{blue}-}}$ can be approximated by $J^{(N)}$ in the scaling limit representation $\pi_{{\color{blue}-}}$ on $\fF_{\ua}(\fh^{(+)}_{\infty,{\color{blue}-}})$ using the asymptotic maps $\alpha^{N}_{\infty}:\fA^{(+)}_{N,{\color{blue}-}}\rightarrow\fA^{(+)}_{\infty,{\color{blue}-}}$ as in the Theorem \ref{thm:KSconv} using the wavelet renormalization group, 
\begin{align}
\label{eq:U1curapprox}
:\pi_{{\color{blue}-}}(\alpha^{N}_{\infty}(J^{(N)}(x)\!)\!):\!\ & =\ :\!\pi_{{\color{blue}-}}(\psi_{+}^{\dag}(s^{(\vep_{N})}_{x})\!)\pi_{{\color{blue}-}}(\psi_{+}(s^{(\vep_{N})}_{x})\!):\;,
\end{align}
and similar for the momentum-cutoff renormalization group as in Theorems \ref{thm:KSconvc} and \ref{thm:KSconvsmeared}.

We infer from \eqref{eq:U1curapprox} that the approximants for the wavelet renormalization group are localized operators in the even subalgebra of $\fA^{(+)}_{\infty,{\color{blue}-}}$ with a localization region determined by the support of the scaling function $s^{(\vep_{N})}_{x}$. 

\bigskip

In momentum space, we have the following picture:
\begin{align}
\label{eq:U1curFourier}
\hat{J}_{{\color{blue}-},k} & =\ :\!dF_{{\color{blue}-}}(j_{k})\!:\ = \tfrac{1}{2L}\!\!\!\sum_{l\in\Gamma_{\infty,{\color{blue}-}}}\!\!\!:\!\pi_{{\color{blue}-}}(\hat{\psi}^{\dag}_{+|l+k})\pi_{{\color{blue}-}}(\hat{\psi}_{+|l})\!:, \\ \nonumber
\hat{J}^{(N)}_{k} & = dF_{0}(j^{(N)}_{k}) = \tfrac{1}{2L_{N}}\!\!\!\sum_{l\in\Gamma_{N,{\color{blue}-}}}\!\hat{\psi}^{\dag}_{+|l+k}\hat{\psi}_{+|l},
\end{align}
such that $:\!\pi_{{\color{blue}-}}(\alpha^{N}_{\infty}(J^{(N)}_{k}))\!:\ =\ :\!dF_{C}(\tilde{j}^{(N)}_{k})\!:$ with $\tilde{j}^{(N)}_{k} = R^{N}_{\infty}j^{(N)}_{k}(R^{N}_{\infty})^{*}$.
The one-particle operators $j^{(N)}_{k}$, $j_{k}$ are the (unitary) momentum-space translations on $\fh^{(+)}_{N,{\color{blue}-}}$ respectively $\fh^{(+)}_{\infty,{\color{blue}-}}$:
\begin{align}
\label{eq:U1cur1p}
 (j^{(N)}_{k}\hat{\xi})_{m} & = \hat{\xi}_{m-k}, & (j_{k}\hat{\xi})_{m} & = \hat{\xi}_{m-k},
\end{align}
for $\xi\in\fh^{(+)}_{N,{\color{blue}-}}$ or $\fh^{(+)}_{\infty,{\color{blue}-}}$, satisfying the implementability condition \eqref{eq:sdcconstraints} and with matrix elements:
\begin{align}
\label{eq:U1cur1pmatrix}
(j^{(N)}_{k})_{mn} & = 2L_{N} \delta^{(N)}_{m,n+k}, & (j_{k})_{mn} & = 2L \delta_{m,n+k}.
\end{align}
Since the off-diagonal parts of $j_{k}$ satisfy the Hilbert-Schmidt condition,
\begin{align}
\label{eq:U1curimp}
\|(j_{k})_{\pm\mp}\|_{2}^{2} & = \sum_{n\in\Gamma_{\infty,{\color{blue}-}}}\theta(\pm n)\theta(\mp(n-k)) < \infty,
\end{align}
and,
\begin{align}
\label{eq:U1curadjoint}
j^{(N)}_{k}{}^{*} & = j^{(N)}_{-k}, & j_{k}^{*} & = j_{k},
\end{align}
the expression for $\hat{J}_{{\color{blue}-},k}$ in \eqref{eq:U1curFourier} results in a densely defined, closable operator on $\fF_{\ua}(\fh^{(+)}_{\infty,{\color{blue}-}})$ \cite{CareyOnFermionGauge, EvansQuantumSymmetriesOn}. As before, we define the smeared $U(1)$- and lattice current,
\begin{align}
\label{eq:U1cursmeared}
J^{(N)}(X) & = \tfrac{1}{2L_{N}}\sum_{k\in\Gamma_{N,{\color{blue}+}}}\hat{X}_{k}J^{(N)}_{k}, & J_{{\color{blue}-}}(X) & = \tfrac{1}{2L}\sum_{k\in\Gamma_{\infty,{\color{blue}+}}}\hat{X}_{k}\hat{J}_{{\color{blue}-},k},
\end{align}
and their one-particle analogues,
\begin{align}
\label{eq:U1cur1psmeared}
j^{(N)}(X) & = \tfrac{1}{2L_{N}}\sum_{k\in\Gamma_{N,{\color{blue}+}}}\hat{X}_{k}j^{(N)}_{k}, & j(X) & = \tfrac{1}{2L}\sum_{k\in\Gamma_{\infty,{\color{blue}+}}}\hat{X}_{k}j_{k},
\end{align}
for with sufficiently regular $X\in\fl^{\alpha}(S^{1}_{L},\C)$ respectively $X\in\fl(\Lambda_{N},\C):$. For real $X = \overline{X}$, i.e.~$iX\in\fu(1)$, the Lie algebra of $U(1)$, $J_{{\color{blue}-}}(X)$ and $J^{(N)}(X)$ are self-adjoint and induce automorphic Bogoliubov transformations of $\fA^{(+)}_{\infty,{\color{blue}-}}$ respectively $\fA^{(+)}_{N,{\color{blue}-}}$:
\begin{align}
\label{eq:U1curBT}
\Ad_{e^{itJ^{(N)}(X)}}(a^{\dagger}(\xi)) & = a^{\dagger}(e^{itj^{(N)}(X)}\xi), & \Ad_{e^{itJ_{{\color{blue}-}}(X)}}(\pi_{{\color{blue}-}}(a^{\dagger}(\xi))) & = \pi_{{\color{blue}-}}(a^{\dagger}(e^{itj(X)}\xi)), 
\end{align}
for $\xi\in\fh^{(+)}_{N,{\color{blue}-}}$ or $\fh^{(+)}_{\infty,{\color{blue}-}}$, and $t\in\R$.

\bigskip

Comparing the matrix elements \eqref{eq:U1cur1pmatrix} with those of the one-particle operators of the Virasoro generators and the KS approximants \eqref{eq:latvirachi1pmommat}, we deduce that all convergence results of Section \ref{sec:KSconv} equally apply to the approximation of the $U(1)$-current by \eqref{eq:U1curfinite} in the scaling limit. Notably, some of the proof of the results for the $U(1)$-current simplify slightly because the one-particle operators $j_{k}$, $k\in\Gamma_{\infty,{\color{blue}+}}$, are bounded. As an illustration we state the analogue of Theorem \ref{thm:KSconvsmeared}:
\begin{thm}[Convergence of smeared lattice currents for $\pi_{{\color{blue}-}}$]
\label{thm:U1cursmearedconv}
Let $s\in C^{\alpha}(\R)$ be a sufficiently regular, compactly supported orthonormal Daubechies scaling function, and let $X\in\fl(\Lambda_{M},\C)$ for some $M\in\N_{0}$. Then, the smeared lattice currents, $J^{(N)}(S^{M}_{N}(X))$, converge strongly to the $U(1)$-current, $J_{{\color{blue}-}}(S^{M}_{\infty}(X))$, on the dense domain $\fF^{\alg}_{\ua}(\cD_{\std})\subset\fF_{\ua}(\fh^{(\pm)}_{\infty,{\color{blue}\pm}})$ spanned by anti-symmetric Fock vectors with finitely many one-particle excitations in $\cD_{\std}$ (finite $\cD_{\std}$-particle number):
\begin{align}
\label{eq:U1curconvsmeared}
\lim_{N\rightarrow\infty}\|(:\!\pi_{{\color{blue}-}}(\alpha^{N}_{\infty}(J^{(N)}(S^{M}_{N}(X)))\!)\!: - J_{{\color{blue}-}}(S^{M}_{\infty}(X)))a^{\dagger}(\xi_{1})\dots a^{\dagger}(\xi_{n})\Omega_{0}\| & = 0,
\end{align}
for all $n\in\N_{0}$ and $\xi_{1},\dots,\xi_{n}\in\cD_{\std}$.
\end{thm}
\begin{proof}
The proof is identical to that Theorem \ref{thm:KSconvsmeared} using the analogue of Lemma \ref{lem:KSconvsmeared} with the additional simplification that we have uniform bounds:
\begin{align*}
\|\tilde{j}^{(N)}_{k}R^{K}_{\infty}(e_{l})\|_{\infty} & = 1, & \|j_{k}R^{K}_{\infty}(e_{l})\|_{\infty} & = 1.
\end{align*}
\end{proof}

\paragraph{WZW-currents.}
We illustrate the adaptation of the procedure for the $U(1)$-current to non-abelian currents by the example of the level-$1$ (chiral) $\hat{\fu}(D)_{1}$-currents. Here, $\fu(D)$ denotes the Lie algebra of $U(D)$, the unitary group in $D$ dimensions. The $\hat{\fu}(D)_{1}$-currents are given by:
\begin{align}
\label{eq:WZWcur}
J^{\mu}_{{\color{blue}-}}(x) & = \sum_{i,j=1}^{D}t^{\mu}_{ij}:\!\pi_{{\color{blue}-}}(\psi^{(i)\dag}_{+|x})\pi_{{\color{blue}-}}(\psi^{(j)}_{+|x})\!:\ = \tfrac{1}{2L}\sum_{k\in\Gamma_{\infty,{\color{blue}-}}}e^{-ikx}\hat{J}^{\mu}_{{\color{blue}-},k}\;,
\end{align}
where the matrices $\{t^{\mu}\}_{\mu=1}^{D^2-1}$ are a basis of the Lie algebra $\fu(D)$, and we use $D$-component (chiral) complex fermions $\fA^{(+)}_{\infty,{\color{blue}-}}(D)=\fA_{\CAR}(\fh^{(+)}_{\infty,{\color{blue}-}}\otimes\C^{D})$. The finite-scale analogues of the non-abelian currents and the associated approximants are given by:
\begin{align}
\label{eq:WZWcurfinite}
J^{(N),\mu}(x) & = \sum_{i,j=1}^{D}t^{\mu}_{ij}\psi^{(i)\dag}_{+|x}\psi^{(j)}_{+|x} = \tfrac{1}{2L_{N}}\!\!\!\sum_{k\in\Gamma_{N,{\color{blue}-}}}e^{-ikx}\hat{J}^{(N),\mu}_{k}, & \hat{J}^{(N),\mu}_{k} & = dF_{0}(j^{(N),\mu}_{k}),
\end{align}
with the one-particle operators that also satisfy the implementability condition \eqref{eq:sdcconstraints}:
\begin{align}
\label{eq:WZWcur1p}
(j^{(N),\mu}_{k}\hat{\xi})_{i,m} & = \sum_{j=1}^{D}t^{\mu}_{ij}\hat{\xi}_{j,m-k}, & (j^{\mu}_{k}\hat{\xi})_{i,m} & = \sum_{j=1}^{D}t^{\mu}_{ij}\hat{\xi}_{j,m-k},
\end{align}
for $\fh^{(+)}_{N,{\color{blue}-}}\otimes\C^{D}$ respectively $\xi\in\fh^{(+)}_{\infty,{\color{blue}-}}\otimes\C^{D}$. The finite-scale one-particle operators are mapped to the scaling limit by $\tilde{j}^{(N),\mu}_{k} = R^{N}_{\infty}(j^{(N),\mu}_{k})(R^{N}_{\infty})^{*}$. where asymptotic maps are used componentwise.

It is natural to use the Hilbert-Schmidt norm $\|\!\ .\!\ \|_{2,D}$ on the basis elements $t^{\mu}$ as it is directly related to the Killing form on $\fu(D)$, and by the Cauchy-Schwarz inequality we have:
\begin{align}
\label{eq:WZWcur1pestimate}
\|j^{(N),\mu}_{k}\hat{\eta}\|_{\fh^{(+)}_{N,{\color{blue}-}}\otimes\C^{D}}^{2} & \leq \|t^{\mu}\|_{2,D}^{2}\sum_{i=1}^{D}\|(j^{(N)}_{k}\hat{\eta})_{i}\|_{\fh_{N}}^{2} = \|t^{\mu}\|_{2,D}^{2}\|\eta\|_{\fh^{(+)}_{N,{\color{blue}-}}\otimes\C^{D}}^{2}, \\ \nonumber
\|j^{\mu}_{k}\hat{\xi}\|_{\fh^{(+)}_{\infty,{\color{blue}-}}\otimes\C^{D}}^{2} & \leq \|t^{\mu}\|_{2,D}^{2}\sum_{i=1}^{D}\|(j_{k}\hat{\xi})_{i}\|_{\fh_{\infty}}^{2} = \|t^{\mu}\|_{2,D}^{2}\|\xi\|_{\fh^{(+)}_{\infty,{\color{blue}-}}\otimes\C^{D}}^{2}.
\end{align}
The convergence results for the non-abelian lattice currents as, for example, in Theorem \ref{thm:U1cursmearedconv} follow directly from the estimate:
\begin{align}
\label{eq:WZWcur1pconvestimate}
\|(\tilde{j}^{(N),\mu}_{k}-j^{\mu}_{k})R^{K}_{\infty}(\eta)\|_{\fh^{(+)}_{\infty,{\color{blue}-}}\otimes\C^{D}}^{2} & \leq \|t^{\mu}\|_{2,D}^{2}\sum_{i=1}^{D}\|(\tilde{j}^{(N)}_{k}-j_{k})R^{K}_{\infty}(\eta_{i})\|_{\infty}^{2},
\end{align}
as this reduce the problem to setting of the $U(1)$-current.

\bigskip

A similar reasoning applies to any non-abelian $\hat{\fg}_{k}$-current of the form \eqref{eq:WZWcur} associated with a (simply-laced) Lie algebra $\fg$ and representation $\lambda$ determining the level $k = k(\lambda)$, possibly by using the Majorana algebras $\fB^{(+)}_{N,{\color{blue}-}}$ and $\fB^{(+)}_{\infty,{\color{blue}-}}$ instead.

\section{Approximation of correlation functions}
\label{sec:corr}
The convergence results of the previous sections for the approximation of Virasoro generator, their unitaries and the associated Bogoliubov transformations allow for the approximation of (chiral) correlation functions of the resulting fermionic continuum field theories in the scaling limit. Explicit error bounds can be obtained along the lines outlined in the last paragraph of Section \ref{sec:KSconv}.

\paragraph{Fermion correlation function.}
Using either the wavelet or the momentum-cutoff renormalization group, we obtain the following approximation theorem concerning the dynamical (chiral) correlation functions of fermions. We state the theorem for complex fermions $\fA^{(\pm)}_{\infty,{\color{blue}\pm}}$, but an analogous statement holds for the Majorana fermions $\fB^{(\pm)}_{\infty,{\color{blue}\pm}}$ as is clear from the structure of the proofs in Section \ref{sec:KSconv}.
\begin{thm}[Convergence of fermion correlation functions]
\label{thm:corapprox}
Given the quasi-free representations $\pi_{{\color{blue}\pm}}$ of the fermion algebra $\fA^{(\pm)}_{\infty,{\color{blue}\pm}}$ arising from the scaling limit states $\omega_{{\color{blue}\pm}}$. Then, for any $M\in\N_{0}$, $A,B\in\fA^{(\pm)}_{M,{\color{blue}\pm}}$ and uniformly in $t\in\R$ on compact intervals, we have:
\begin{align*}
\lim_{N\rightarrow\infty}(\omega^{(N)}_{0,{\color{blue}\pm}},\pi^{(N)}_{{\color{blue}\pm}}(\alpha^{M}_{N}(A))\sigma^{(N)}_{t}(\pi^{(N)}_{{\color{blue}\pm}}(\alpha^{M}_{N}(B)))\omega^{(N)}_{0,{\color{blue}\pm}}) & = (\Omega_{0},\pi_{{\color{blue}\pm}}(\alpha^{M}_{\infty}(A))\sigma_{t}(\pi_{{\color{blue}\pm}}(\alpha^{M}_{\infty}(B)))\Omega_{0})\;,
\end{align*}
where $\omega^{(N)}_{0,{\color{blue}\pm}}$, $\pi^{(N)}_{{\color{blue}\pm}}$ are the GNS-vector and -representation of $\omega^{(N)}_{0,{\color{blue}\pm}}$, and $\Omega_{0}\in\fF_{\ua}(\fh^{(\pm)}_{\infty,{\color{blue}\pm}})$ is the standard Fock vacuum. $\sigma^{(N)}_{t}$ and $\sigma_{t}$ are $1$-parameter (semi-)groups of (automorphic) Bogoliubov transformations generated by a finite-scale KS approximant \eqref{eq:latvirasubchi} or a (non-)abelian current \eqref{eq:U1curfinite} \& \eqref{eq:WZWcurfinite} and their scaling limits respectively.
\end{thm}
\begin{proof}
Let us first introduce the short hands:
\begin{align*}
C^{(N)}_{t}(A,B) & = (\omega^{(N)}_{0,{\color{blue}\pm}},\pi^{(N)}_{{\color{blue}\pm}}(A)\sigma^{(N)}_{t}(\pi^{(N)}_{{\color{blue}\pm}}(B))\omega^{(N)}_{0,{\color{blue}\pm}}), & C_{t}(A,B) & = (\Omega_{0},\pi_{{\color{blue}\pm}}(A)\sigma_{t}(\pi_{{\color{blue}\pm}}(B))\Omega_{0})\;,
\end{align*}
for either $A,B\in\fA^{(\pm)}_{N,{\color{blue}\pm}}$ or $A,B\in\fA^{(\pm)}_{\infty,{\color{blue}\pm}}$, as well as,
\begin{align*}
\tilde{C}^{(N)}_{t}(A,B) & = (\Omega_{0},\pi_{{\color{blue}\pm}}(A)\tilde{\sigma}^{(N)}_{t}(\pi_{{\color{blue}\pm}}(B))\Omega_{0})\;,
\end{align*}
with the approximation $\tilde{\sigma}^{(N)}_{t}$ of $\sigma_{t}$ given by $\tilde{\sigma}^{(N)}_{t} = \Ad_{\pi_{{\color{blue}\pm}}(\alpha^{N}_{\infty}(U_{N}(t)))}$ because $\sigma^{(N)}_{t} = \Ad_{U_{N}(t)}$ is inner by construction, i.e.~$U_{N}(t)\in\fA^{(\pm)}_{N,{\color{blue}\pm}}$. Thus, the statement can be rephrased as:
\begin{align*}
\lim_{N\rightarrow\infty}C^{(N)}_{t}(\alpha^{M}_{N}(A),\alpha^{M}_{N}(B)) & = C_{t}(\alpha^{M}_{\infty}(A),\alpha^{M}_{\infty}(B)).
\end{align*}
We observe that:
\begin{align}
\label{eq:algcor}
C^{(N)}_{t}(\alpha^{M}_{N}(A),\alpha^{M}_{N}(B)) & = (\omega^{(N)}_{0,{\color{blue}\pm}},\pi^{(N)}_{{\color{blue}\pm}}(\alpha^{M}_{N}(A)\Ad_{U_{N}(t)}(\alpha^{M}_{N}(B)))\omega^{(N)}_{0,{\color{blue}\pm}}) \\ \nonumber
 & = \omega^{(N)}_{0,{\color{blue}\pm}}(\alpha^{M}_{N}(A)\Ad_{U_{N}(t)}(\alpha^{M}_{N}(B))), \\[0.25cm] \nonumber
\tilde{C}^{(N)}_{t}(\alpha^{M}_{\infty}(A),\alpha^{M}_{\infty}(B))  & = (\Omega_{0},(\pi_{{\color{blue}\pm}}\circ\alpha^{N}_{\infty})(\alpha^{M}_{N}(A)\Ad_{U_{N}(t)}(\alpha^{M}_{N}(B)))\Omega_{0}) \\ \nonumber
 & = \omega_{{\color{blue}\pm}}(\alpha^{N}_{\infty}(\alpha^{M}_{N}(A)\Ad_{U_{N}(t)}(\alpha^{M}_{N}(B))))\;,
\end{align}
Thus, we  can estimate:
\begin{align*}
 & |C^{(N)}_{t}\!(\alpha^{M}_{N}\!(A),\alpha^{M}_{N}\!(B)\!)\!-\!C_{t}(\alpha^{M}_{\infty}\!(A),\alpha^{M}_{\infty}\!(B)\!)| \\
 &\!\leq\!|C^{(N)}_{t}\!(\alpha^{M}_{N}\!(A),\alpha^{M}_{N}\!(B)\!)\!-\!\tilde{C}^{(N)}_{t}\!(\alpha^{M}_{\infty}\!(A),\alpha^{M}_{\infty}\!(B)\!)|\!+\!|\tilde{C}^{(N)}_{t}\!(\alpha^{M}_{\infty}\!(A),\alpha^{M}_{\infty}\!(B)\!)\!-\!C_{t}(\alpha^{M}_{\infty}\!(A),\alpha^{M}_{\infty}\!(B)\!)| \\
&\!\leq\!|(\Omega_{0},\pi_{{\color{blue}\pm}}(\alpha^{M}_{\infty}\!(A)\!)(\tilde{\sigma}^{(N)}_{t}\!-\!\sigma_{t})(\pi_{{\color{blue}\pm}}(\alpha^{M}_{\infty}\!(B)\!)\!)\Omega_{0})|\!+\!|(\omega^{(N)}_{{\color{blue}\pm}}\!\!-\!\omega^{(N)}_{0,{\color{blue}\pm}})(\alpha^{M}_{N}\!(A)\!\Ad_{U_{N}(t)}(\alpha^{M}_{N}\!(B)\!)\!)| \\
&\!\leq\!\|A\|\|(\tilde{\sigma}^{(N)}_{t}\!-\!\sigma_{t})(\pi_{{\color{blue}\pm}}(\alpha^{M}_{\infty}\!(B)\!)\!)\|\!+\!|(\omega^{(N)}_{{\color{blue}\pm}}\!\!-\!\omega^{(N)}_{0,{\color{blue}\pm}})(\alpha^{M}_{N}\!(A)\!\Ad_{U_{N}(t)}(\alpha^{M}_{N}\!(B)\!)\!)|
\end{align*}
Thus, the result follows from Corollary \ref{cor:KSconvbog} (or \ref{cor:KSconvbogc} \& \ref{cor:KSconvbogsmeared}) provided:
\begin{align*}
\lim_{N\rightarrow\infty}|(\omega^{(N)}_{{\color{blue}\pm}}\!\!-\!\omega^{(N)}_{0,{\color{blue}\pm}})(\alpha^{M}_{N}(A)\Ad_{U_{N}(t)}(\alpha^{M}_{N}(B))))| & = 0\;.
\end{align*}
But, this is itself a consequence of Corollary \ref{cor:KSconvbog} and the fact that $\omega_{{\color{blue}\pm}}$ is the scaling limit of $\omega^{(N)}_{0,{\color{blue}\pm}}$. To see this, we observe that Corollary \ref{cor:KSconvbog} implies for $\pi_{S} = \pi_{0}$ (the standard Fock-space representation) that $O_{N} = \alpha^{M}_{N}(A)\Ad_{U_{N}(t)}(\alpha^{M}_{N}(B))$ converges to an element $O\in\fA^{(\pm)}_{\infty,{\color{blue}\pm}}$ for any $M\in\N_{0}: A,B\in\fA_{M}$ and $t\in\R$ (uniformly on compact intervals):
\begin{align*}
\lim_{K\rightarrow\infty}\lim_{N\rightarrow\infty}\!\!\|O_{N}\!-\!\alpha^{K}_{N}(O_{K})\| &\!=\!\lim_{K\rightarrow\infty}\lim_{N\rightarrow\infty}\!\!\|\alpha^{N}_{\infty}(O_{N})\!-\!\alpha^{N}_{\infty}(\alpha^{K}_{N}(O_{K})\!)\| \\
& \!\leq\!\lim_{K\rightarrow\infty}\lim_{N\rightarrow\infty}\!\!\|\alpha^{N}_{\infty}\!(O_{N})\!-\!O\|\!+\!\lim_{K\rightarrow\infty}\lim_{N\rightarrow\infty}\!\!\|\alpha^{N}_{\infty}\!(O_{N})\!-\!\alpha^{N}_{\infty}\!(\alpha^{K}_{N}\!(O_{K})\!)\| \\
& = \lim_{N\rightarrow\infty}\|\alpha^{N}_{\infty}(O_{N})-O\| + \lim_{K\rightarrow\infty}\|O-\alpha^{K}_{\infty}(O_{K})\| \\
& = 0,
\end{align*}
with $O$ defined by $\pi_{0}(\alpha^{M}_{\infty}(A)\!)\sigma_{t}(\pi_{0}(\alpha^{M}_{\infty}(B)\!)\!)$, and $\pi_{0}(\alpha^{N}_{\infty}(O_{N})\!)\!=\!\pi_{0}(\alpha^{M}_{\infty}(A)\!)\tilde{\sigma}^{(N)}_{t}\!(\pi_{0}(\alpha^{M}_{\infty}(B)\!)\!)$. Finally, the scaling limit construction of $\omega_{{\color{blue}\pm}}$ gives (see Lemma \ref{lem:stateconv}):
\begin{align*}
\lim_{N\rightarrow\infty}\omega^{(N)}_{0,{\color{blue}\pm}}(\alpha^{K}_{N}(O_{K})\!) & = \omega_{{\color{blue}\pm}}(\alpha^{K}_{\infty}(O_{K})\!), & \forall K&\in\N_{0},
\end{align*}
which defines $\omega_{{\color{blue}\pm}}$ on the dense subalgebra $\bigcup_{K\in\N_{0}}\alpha^{K}_{\infty}(\fA^{(\pm)}_{K,{\color{blue}\pm}})$. The extension to $O\in\fA^{(\pm)}_{\infty,{\color{blue}\pm}}$ can be determined via the approximating sequence $\alpha^{K}_{\infty}(O_{K})$:
\begin{align*}
\omega_{{\color{blue}\pm}}(O_{\infty}) & = \lim_{K\rightarrow\infty}\omega_{{\color{blue}\pm}}(\alpha^{K}_{\infty}(O_{K})) = \lim_{K\rightarrow\infty}\lim_{N\rightarrow\infty}\omega^{(N)}_{0,{\color{blue}\pm}}(\alpha^{K}_{N}(O_{K})) \\
& = \lim_{K\rightarrow\infty}\lim_{N\rightarrow\infty}\omega^{(N)}_{0,{\color{blue}\pm}}(O_{N}) + \lim_{K\rightarrow\infty}\lim_{N\rightarrow\infty}\omega^{(N)}_{0,{\color{blue}\pm}}(O_{N}-\alpha^{K}_{N}(O_{K})) \\
& = \lim_{N\rightarrow\infty}\omega^{(N)}_{0,{\color{blue}\pm}}(O_{N})\;.
\end{align*}
which proves the result. To achieve uniform convergence in $t$ on compact intervals, we need the following estimate:
\begin{align*}
|(\omega^{(N)}_{{\color{blue}\pm}}\!\!-\!\omega^{(N)}_{0,{\color{blue}\pm}})(O_{N})| & \!\leq\!|\omega^{(N)}_{{\color{blue}\pm}}\!(O_{N}\!-\!\alpha^{K}_{N}(O_{K})\!)|\!+\!|(\omega^{(N)}_{{\color{blue}\pm}}\!\!-\!\omega^{(N)}_{0,{\color{blue}\pm}})(\alpha^{K}_{N}(O_{K})\!)|\!+\!|\omega^{(N)}_{0,{\color{blue}\pm}}\!(\alpha^{K}_{N}(O_{K})\!-\!O_{N})| \\ 
 & \!\leq\! 2\|O_{N}\!-\!\alpha^{K}_{N}(O_{K})\|\!+\!|(\omega^{(K)}_{\pm}\!-\!\omega^{(K)}_{N-K,{\color{blue}\pm}})(O_{K})| \\
 & \!\leq\! 2(\|\alpha^{N}_{\infty}(O_{N})\!-\!O\|\!+\!\|O\!-\!\alpha^{K}_{\infty}(O_{K})\|)+\|\omega^{(K)}_{{\color{blue}\pm}}-\omega^{(K)}_{N-K,{\color{blue}\pm}}\|\|O_{K}\|.
\end{align*}
Since, by Corollary \ref{cor:KSconvbog}, $\|\alpha^{N}_{\infty}(O_{N})-O\|$ and $\|O-\alpha^{K}_{\infty}(O_{K})\|$ can be made uniformly small for $t$ in a compact interval, we only need to show that this is possible for $\|\omega^{(K)}_{{\color{blue}\pm}}-\omega^{(K)}_{N-K,{\color{blue}\pm}}\|\|O_{K}\|$ as well. But, this follows from the fact that $\omega^{(K)}_{N-K,{\color{blue}\pm}}$ approximates $\omega^{(K)}_{{\color{blue}\pm}}$ in the sense of Lemma \ref{lem:stateconv}, and provided $\|U_{K}(t)\|\leq C^{(K)}_{I}$ is uniformly bounded in $t$ on compact intervals $I\subset\R$ (obviously true in the unitary case) because:
\begin{align*}
\|O_{K}\| & \leq \|U_{K}(t)\|^{2}\|A\|\|B\|.
\end{align*}
\end{proof}

\begin{cor}
\label{cor:corapprox}
The statement of Theorem \ref{thm:corapprox} remains valid, if $A,B\in\fA^{(\pm)}_{M,{\color{blue}\pm}}$ (for some $M\in\N_{0}$) are replaced by convergent sequences $\{A_{N}\}_{N\in\N_{0}}$, $\{B_{N}\}_{N\in\N_{0}}$ in the sense of \eqref{eq:convseq}. Then
\begin{align*}
\lim_{N\rightarrow\infty}(\omega^{(N)}_{0,{\color{blue}\pm}},\pi^{(N)}_{{\color{blue}\pm}}(A_{N})\sigma^{(N)}_{t}(\pi^{(N)}_{{\color{blue}\pm}}((B_{N})))\omega^{(N)}_{0,{\color{blue}\pm}}) & = (\Omega_{0},\pi_{{\color{blue}\pm}}(A)\sigma_{t}(\pi_{{\color{blue}\pm}}(B))\Omega_{0})\;,
\end{align*}
where $\lim_{N\rightarrow\infty}\alpha^{N}_{\infty}(A_{N}) = A$ and $\lim_{N\rightarrow\infty}\alpha^{N}_{\infty}(B_{N}) = B$.
\end{cor}
\begin{proof}
This follows using,
\begin{align*}
(\omega^{(N)}_{0,{\color{blue}\pm}},\pi^{(N)}_{{\color{blue}\pm}}(A_{N})\omega^{(N)}_{0,{\color{blue}\pm}}) - (\Omega_{0},\pi_{{\color{blue}\pm}}(A)\Omega_{0})& = (\omega^{(N)}_{0,{\color{blue}\pm}},\pi^{(N)}_{{\color{blue}\pm}}(A_{N})\omega^{(N)}_{0,{\color{blue}\pm}}) - (\Omega_{0},\pi_{{\color{blue}\pm}}(\alpha^{N}_{\infty}(A_{N}))\Omega_{0}) \\
& \hspace{0.2cm} + (\Omega_{0},\pi_{{\color{blue}\pm}}(\alpha^{N}_{\infty}(A_{N}))\Omega_{0}) - (\Omega_{0},\pi_{{\color{blue}\pm}}(A)\Omega_{0}) \\
& = (\omega^{(N)}_{0,{\color{blue}\pm}} - \omega^{(N)}_{{\color{blue}\pm}})(A_{N})+\omega_{{\color{blue}\pm}}(\alpha^{N}_{\infty}(A_{N})-A),
\end{align*}
repeatedly applying the triangle inequality, and the observation that,
\begin{align*}
\lim_{N\rightarrow\infty}\|\alpha^{N}_{\infty}(A_{N}\sigma^{(N)}_{t}(B_{N}))-A\sigma_{t}(B)\| & = 0,
\end{align*}
by applying Corollary \ref{cor:KSconvbog} to the standard Fock-space representation $\pi_{0}$.
\end{proof}

\paragraph{Virasoro correlation functions}
Since the (smeared) KS approximants converge to the (smeared) Virasoro generators in the scaling limit representations $\pi_{{\color{blue}\pm}}$ using the momentum-cutoff renormalization group (possibly in combination with the wavelet renormalization group for the smearing functions) by Theorems \ref{thm:KSconvc} \& \ref{thm:KSconvsmeared}, we obtain an analogue of Theorem \ref{thm:corapprox} for correlation functions involving the Virasoro generators or their associated unitaries, cp.~\cite{ZiniConformalFieldTheories}. A subtle difference arises from the fact that the KS approximants are only strongly operator-convergent on the Fock-space vectors with finite particle number in momentum space, $\fF^{\alg}_{\ua}(\cD_{\std})$, and not operator-convergent in the sense of Corollary \ref{cor:corapprox}. Again, we state the theorem for complex fermions $\fA^{(\pm)}_{\infty,{\color{blue}\pm}}$, but an analogous statement holds for the Majorana fermions $\fB^{(\pm)}_{\infty,{\color{blue}\pm}}$.
\begin{thm}[Convergence of Virasoro correlation functions]
\label{thm:corapproxvir}
Let $s\in C^{\alpha}(\R)$ be a sufficiently regular, compactly supported orthonormal Daubechies scaling function. Given the quasi-free representations $\pi_{{\color{blue}\pm}}$ of the fermion algebra $\fA^{(\pm)}_{\infty,{\color{blue}\pm}}$ arising from the scaling limit states $\omega_{{\color{blue}\pm}}$. Then, for any $n\in\N$ and $S$-convergent sequences of smearing functions $X_{N,p}\stackrel{N\rightarrow\infty}{\rightarrow} X_{p}$, $N\in\N_{0}$ and $p=1,...,n$ with sufficient regularity as in Remark \ref{rem:KSconvsmeared}, for example, $X_{N,p} = S^{M_{p}}_{N}(X_{p})$ for $X_{p}\in\fl(\Lambda_{M_{p}},\R)$, we have:
\begin{align*}
\lim_{N\rightarrow\infty}(\Omega_{0},\prod_{p=1}^{n}:\!(\pi_{{\color{blue}\pm}}\circ\alpha^{N}_{\infty})(L^{(N)}_{\pm}(X_{N,p}))\!:\Omega_{0}) & = (\Omega_{0},\prod_{p=1}^{n}L_{\pm}(X_{p})\Omega_{0})\;,
\end{align*}
and similarly,
\begin{align*}
\lim_{N\rightarrow\infty}(\Omega_{0},\prod_{p=1}^{n}e^{i:(\pi_{{\color{blue}\pm}}\circ\alpha^{N}_{\infty})(L^{(N)}_{\pm}(X_{N,p})):}\Omega_{0}) & = (\Omega_{0},\prod_{p=1}^{n}e^{iL_{\pm}(X_{p})}\Omega_{0})\;,
\end{align*}
where $\Omega_{0}$ is the standard Fock vacuum of $\fF_{\ua}(\fh^{(\pm)}_{\infty,{\color{blue}\pm}})$. Moreover, the finite-scale correlation functions of the scaling limit state $\omega_{{\color{blue}\pm}}$ can be approximated in terms of the renormalized finite-scale states:
\begin{align*}
& (\Omega_{0},\prod_{p=1}^{n}:\!(\pi_{{\color{blue}\pm}}\circ\alpha^{N}_{\infty})(L^{(N)}_{\pm}(X_{N,p})\!)\!:\Omega_{0}) = \omega^{(N)}_{{\color{blue}\pm}}(\prod_{p=1}^{n}(L^{(N)}_{\pm}(X_{N,p})-\omega^{(N)}_{{\color{blue}\pm}}(L^{(N)}_{\pm}(X_{N,p})\!)\!)\!) \\
& = \lim_{M\rightarrow\infty}\omega^{(N)}_{M,{\color{blue}\pm}}(\prod_{p=1}^{n}(L^{(N)}_{\pm}(X_{N,p})-\omega^{(N)}_{{\color{blue}\pm}}(L^{(N)}_{M,\pm}(X_{N,p})\!)\!)\!) \\
& = \lim_{M\rightarrow\infty}\omega^{(N+M)}_{0,{\color{blue}\pm}}\!(\prod_{p=1}^{n}(\alpha^{N}_{N+M}(L^{(N)}_{\pm}\!(X_{N,p})\!)\!-\!\omega^{(N+M)}_{0,{\color{blue}\pm}}\!(\alpha^{N}_{N+M}(L^{(N)}_{\pm}\!(X_{N,p})\!)\!)\!)\!) \;,
\end{align*}
and similarly,
\begin{align*}
& (\Omega_{0},\prod_{p=1}^{n}e^{i:(\pi_{{\color{blue}\pm}}\circ\alpha^{N}_{\infty})(L^{(N)}_{\pm}(X_{N,p})\!):}\Omega_{0}) = \omega^{(N)}_{{\color{blue}\pm}}(\prod_{p=1}^{n}e^{iL^{(N)}_{\pm}(X_{N,p})-i\omega^{(N)}_{{\color{blue}\pm}}(L^{(N)}_{\pm}(X_{N,p})\!)}) \\
& = \lim_{M\rightarrow\infty}\omega^{(N)}_{M,{\color{blue}\pm}}(\prod_{p=1}^{n}e^{iL^{(N)}_{\pm}(X_{N,p})-i\omega^{(N)}_{M,{\color{blue}\pm}}(L^{(N)}_{\pm}(X_{N,p})\!)}) \\
& = \lim_{M\rightarrow\infty}\omega^{(N+M)}_{0,{\color{blue}\pm}}\!(\prod_{p=1}^{n}e^{i\alpha^{N}_{N+M}(L^{(N)}_{\pm}(X_{N,p})\!)-i\omega^{(N+M)}_{0,{\color{blue}\pm}}\!(\alpha^{N}_{N+M}(L^{(N)}_{\pm}(X_{N,p})\!)\!)}) \;,
\end{align*}
where $\omega^{(N)}_{{\color{blue}\pm}} = \omega_{{\color{blue}\pm}}\circ\alpha^{N}_{\infty}$.
\end{thm}
\begin{proof}
The statements follow from a direct application of Theorem \ref{thm:KSconvsmeared} and Corollary \ref{cor:KSconvsmeared}, the identity,
\begin{align*}
:\!(\pi_{S}\circ\alpha^{N}_{\infty})(L^{(N)}_{\pm}(X_{N,p})\!)\!: & = (\pi_{S}\circ\alpha^{N}_{\infty})(L^{(N)}_{\pm}(X_{N,p})\!)-\omega^{(N)}_{S}(L^{(N)}_{\pm}(X_{N,p})\!),
\end{align*}
and the convergence, according to Lemma \ref{lem:stateconv}, of the renormalized states $\omega^{(N)}_{M,{\color{blue}\pm}}$ at finite scales to the scaling limit $\omega^{(N)}_{{\color{blue}\pm}}$:
\begin{align}
\label{eq:corapproxvirestimate}
& |(\Omega_{0},\!\Big(\prod_{p=1}^{n}\!:\!\pi_{{\color{blue}\pm}}(\alpha^{N}_{\infty}(L^{(N)}_{\pm}(X_{N,p})\!)\!)\!:\!-\!\prod_{p=1}^{n}L_{\pm}(X_{p})\!)\!\!\Big)\Omega_{0})| \\ \nonumber
& \!\leq\!\sum^{n}_{q=1}|(\Omega_{0},\!\Big(\!\!\prod_{p=q+1}^{n}\!\!\!\!:\!\pi_{{\color{blue}\pm}}(\alpha^{N}_{\infty}\!(L^{(N)}_{\pm}\!(X_{N,p})\!)\!)\!:\!\!\Big)\!\Big(\!\!:\!\pi_{{\color{blue}\pm}}(\alpha^{N}_{\infty}\!(L^{(N)}_{\pm}\!(X_{N,q})\!)\!)\!:\!-L_{\pm}(X_{q})\!\!\Big)\!\Big(\!\prod_{p=1}^{q-1}L_{\pm}(X_{p})\!\!\Big)\Omega_{0})| \\ \nonumber
& \!\leq\!\sum^{n}_{q=1}\Big\|\!\!\prod_{p=q+1}^{n}\!\!\!\!:\!\pi_{{\color{blue}\pm}}(\alpha^{N}_{\infty}\!(L^{(N)}_{\pm}\!(X_{N,p})\!)\!)\!:\!\Omega_{0}\Big\|\Big\|\!\Big(\!\!:\!\pi_{{\color{blue}\pm}}(\alpha^{N}_{\infty}(L^{(N)}_{\pm}\!(X_{N,q})\!)\!)\!:\!-L_{\pm}(X_{q})\!\!\Big)\!\Big(\prod_{p=1}^{q-1}\!L_{\pm}(X_{p})\!\!\Big)\Omega_{0}\Big\|,
\end{align}
and similarly,
\begin{align}
\label{eq:corapproxvirestimateunitary}
& |(\Omega_{0},\Big(\prod_{p=1}^{n}e^{i:\pi_{{\color{blue}\pm}}(\alpha^{N}_{\infty}(L^{(N)}_{\pm}(X_{N,p})\!)\!):}\!-\!\prod_{p=1}^{n}e^{iL_{\pm}(X_{p})}\!\Big)\Omega_{0})| \\ \nonumber
& \!\leq\!\sum^{n}_{q=1}|(\Omega_{0},\!\Big(\!\!\prod_{p=q+1}^{n}\!\!\!e^{i:\pi_{{\color{blue}\pm}}(\alpha^{N}_{\infty}(L^{(N)}_{\pm}(X_{N,p})\!)\!):}\!\Big)\!\Big(e^{i:\pi_{{\color{blue}\pm}}(\alpha^{N}_{\infty}(L^{(N)}_{\pm}(X_{N,q})\!)\!):}\!-e^{iL_{\pm}(X_{q})}\!\Big)\!\Big(\prod_{p=1}^{q-1}e^{iL_{\pm}(X_{p})}\!\Big)\Omega_{0})| \\ \nonumber
& \!\leq\! \sum^{n}_{q=1}\Big\|\!\!\prod_{p=q+1}^{n}\!\!\!e^{i:\pi_{{\color{blue}\pm}}(\alpha^{N}_{\infty}(L^{(N)}_{\pm}(X_{N,p})\!)\!):}\Omega_{0}\Big\|\Big\|\!\Big(e^{i:\pi_{{\color{blue}\pm}}(\alpha^{N}_{\infty}(L^{(N)}_{\pm}(X_{N,q})\!)\!):}\!-\!e^{iL_{\pm}(X_{q})}\!\Big)\!\Big(\prod_{p=1}^{q-1}e^{iL_{\pm}(X_{p})}\!\Big)\Omega_{0}\Big\| \\ \nonumber
& \!\leq\!\sum^{n}_{q=1}\!\Big\|\!\Big(e^{i:\pi_{{\color{blue}\pm}}(\alpha^{N}_{\infty}(L^{(N)}_{\pm}(X_{N,q})\!)\!):}\!-\!e^{iL_{\pm}(X_{q})}\!\Big)\!\Big(\prod_{p=1}^{q-1}e^{iL_{\pm}(X_{p})}\!\Big)\Omega_{0}\Big\|,
\end{align}
In the case of correlation functions of Virasoro generators, we need two additional observation: First, $\|\prod_{p=1}^{n}\!:\!\!(\pi_{{\color{blue}\pm}}(\alpha^{N}_{\infty}(L^{(N)}_{\pm}(X_{N,p})\!)\!)\!\!:\!\Omega_{0}\|$ can be bounded uniformly in $N$ in terms of $\|\prod_{p=1}^{n}L_{\pm}(X_{p})\Omega_{0}\|$ for sufficiently regular $X_{p}$ as powers of the approximants $:\!(\pi_{{\color{blue}\pm}}\circ\alpha^{N}_{\infty})(L^{(N)}_{\pm}(X_{N,p}))\!:$ are uniformly bounded (in $N$) by $L_{\pm}(X_{p})$. Second, $\prod_{p=1}^{n}L_{\pm}(X_{p}))\Omega_{0}$ is given by a convergent sum (depending on the regularity of $X_{p}$) of Fock-space vectors with finite particle number in momentum space.
\end{proof}
\begin{rem}[Mixed correlation functions]
\label{rem:corapproxvir}
Theorem \ref{thm:corapprox} \& \ref{thm:corapproxvir} can be combined to obtain convergent expressions for mixed correlation functions, i.e.~those containing insertions of operators in $\fA^{(\pm)}_{\infty,{\color{blue}\pm}}$ respectively $\fB^{(\pm)}_{\infty,{\color{blue}\pm}}$, Virasoro generators, WZW currents and their associated unitaries.
\end{rem}

\section*{Acknowledgements}
The authors would like to thank R.~F.~Werner for valuable discussion about inductive limits in quantum theory. AS would like to thank Y.~Tanimoto for helpful discussions concerning the essential self-adjointness of smeared Virasoro generators. Moreover, the authors would like to extend their gratitude to the unknown reviewers providing valuable feedback on a previous version of the manuscript, which improved the presentation and clarified the relation to existing work. Special thanks are extended to one of the reviewers for suggesting to clarify the relation with the Temperley-Lieb algebra and suggesting the case of symplectic fermions for future work. This work was supported, in part, by the DFG through SFB 1227 (DQ-mat), Quantum Valley Lower Saxony, and funded by the Deutsche Forschungsgemeinschaft (DFG, German Research Foundation) under Germanys Excellence Strategy EXC-2123 QuantumFrontiers 390837967.

\section*{Declarations}

The authors have no conflicts of interest to declare.


\end{document}